\documentclass[11pt]{article}
\pdfoutput=1
\usepackage[english]{babel}
\usepackage[centertags]{amsmath}
\usepackage{amsthm}
\usepackage[square, comma, sort&compress,numbers]{natbib}
\usepackage{array,multirow}
\numberwithin{equation}{section}
\usepackage{amssymb,amsfonts}
\usepackage{graphicx}
\usepackage{color}
\usepackage{mathtools,bm}
\usepackage{accents}

\usepackage{calligra}
\usepackage[T1]{fontenc}

\usepackage{epsfig}

\usepackage{wrapfig}
\usepackage{float}
\usepackage{soul}
\usepackage{tkz-euclide}
\usepackage{braket}
\usepackage{tikz,pgf}
\usetikzlibrary{decorations.pathreplacing,calligraphy}
\tikzset{snake it/.style={decorate, decoration=snake}}
\usetikzlibrary{decorations.pathreplacing,angles,quotes}
\usetikzlibrary{shapes}
\usetikzlibrary{calc}
\usetikzlibrary{decorations.pathmorphing}
\usetikzlibrary{decorations.pathreplacing,shapes.misc}
\usetikzlibrary{positioning}
\usetikzlibrary{arrows}
\usetikzlibrary{decorations.markings}
\usetikzlibrary{shadings}

\usetikzlibrary{intersections}

\def\be{\begin{equation}}
\def\ee{\end{equation}}
\def\ba{\begin{array}}
\def\ea{\end{array}}

\def\dps{\displaystyle}
\newcommand{\half}{\frac{1}{2}}

\def\1{\tilde{1}}
\def\2{\tilde{2}}
\def\3{\tilde{3}}

\newdimen\tableauside\tableauside=1.0ex
\newdimen\tableaurule\tableaurule=0.4pt
\newdimen\tableaustep
\def\phantomhrule#1{\hbox{\vbox to0pt{\hrule height\tableaurule
width#1\vss}}}
\def\phantomvrule#1{\vbox{\hbox to0pt{\vrule width\tableaurule
height#1\hss}}}
\def\sqr{\vbox{%
\phantomhrule\tableaustep

\hbox{\phantomvrule\tableaustep\kern\tableaustep\phantomvrule\tableaustep}%
\hbox{\vbox{\phantomhrule\tableauside}\kern-\tableaurule}}}
\def\squares#1{\hbox{\count0=#1\noindent\loop\sqr
\advance\count0 by-1 \ifnum\count0>0\repeat}}
\def\tableau#1{\vcenter{\offinterlineskip
\tableaustep=\tableauside\advance\tableaustep by-\tableaurule
\kern\normallineskip\hbox
{\kern\normallineskip\vbox
{\gettableau#1 0 }%
\kern\normallineskip\kern\tableaurule}%
\kern\normallineskip\kern\tableaurule}}
\def\gettableau#1 {\ifnum#1=0\let\next=\null\else
\squares{#1}\let\next=\gettableau\fi\next}

\newtheorem{prop}{Proposition}[section]
\newtheorem{lemma}[prop]{Lemma}
\newtheorem{definition}{Definition}[section]

\newtheorem{corollary}{Corollary}[section]

\newcommand{\bref}[1]{\textbf{\ref{#1}}}

\newcommand{\RR}{\mathbb{R}}
\newcommand{\ZZ}{\mathbb{Z}}
\newcommand{\NN}{\mathbb{N}}

\def\cF{\mathcal{F}}

\def\cL{\mathcal{L}}
\def\cM{\mathcal{M}}
\def\cN{\mathcal{N}}
\def\cO{\mathcal{O}}
\def\cP{\mathcal{P}}
\def\cQ{\mathcal{Q}}

\def\cS{\mathcal{S}}

\def\cW{\mathcal{W}}

\numberwithin{equation}{section} \makeatletter
\@addtoreset{equation}{section}

\def\be{\begin{equation}}
\def\ee{\end{equation}}
\def\ba{\begin{array}}
\def\ea{\end{array}}

\def\dps{\displaystyle}

\def\ba{\begin{array}}
\def\ea{\end{array}}

\def\dps{\displaystyle}

\def\tdelta{\tilde \delta}

\newcommand{\td}{{h}}

\newcommand{\wb}{\bar{w}}

\def\C2{\text{C}_2}

\def\sl2{sl(2,\mathbb{R})}

\def\pr{\text{\calligra Prin}\;}
\def\reg{\text{\calligra Reg}\;}
\def\g2{\textup{P}}
\def\f1{\textup{B}}
\def\h2{\mathcal{K}}
\def\wb{\mathcal{B}}
\def\ws{\mathcal{T}} 
\def\tih{\tilde h}
\def\tiw{\tilde w}
\def\cb{\mathsf{f}}
\def\QQ{\mathbb{Q}^+_0} 
\def\rn{{_{\mathbb{Q}}}}

\definecolor{forestgreen(web)}{rgb}{0.13, 0.55, 0.13}

\usepackage{jheppub}
\makeatletter
\def\@fpheader{\vspace{-.1cm}}
\makeatother

\title{\centering{A note on the  large-$c$ conformal block asymptotics and $\alpha$-heavy operators}}

\author[a,b]{Konstantin\ Alkalaev}   
\author[c]{and Pavel\ Litvinov}

\affiliation[a]{I.E. Tamm Department of Theoretical Physics, \\P.N. Lebedev Physical
Institute, 119991 Moscow, Russia}
\affiliation[b]{Institute for Theoretical and Mathematical Physics,\\
Lomonosov Moscow State University,
119991 Moscow, Russia}
\affiliation[c]{Department of Physics of Complex Systems, \\ Weizmann Institute of Science, Rehovot 7610001, Israel}

\abstract{We consider $\alpha$-heavy conformal  operators in CFT$_2$ which dimensions grow as $h = O(c^\alpha)$ with $\alpha$ being non-negative rational number and conjecture  that the large-$c$ asymptotics of the respective 4-point Virasoro conformal block is exponentiated similar to the standard case of $\alpha=1$. It is shown that  the leading exponent is given by a Puiseux polynomial which is a  linear combination of power functions in the central charge with fractional powers  decreasing  from $\alpha$ to $0$ according to some pattern. Our analysis is limited by considering the first six explicit coefficients of the Virasoro block function in the coordinate. For simplicity,  external primary operators are chosen to be of equal conformal dimensions that, therefore, includes  the case of the vacuum  conformal block. The consideration is also extended to  the  4-point  ${\cal W}_3$  conformal block of  four semi-degenerate operators, in which case the exponentiation hypothesis  works the same way. Here, only the first three block coefficients can be treated  analytically.
}

\begin{document}

\maketitle
\flushbottom

\section{Introduction}

The large-$c$ regime in CFT$_2$  plays an important role  in the AdS$_3$/CFT$_2$ correspondence due to the Brown-Henneaux relation \cite{Brown:1986nw} which in this case says that  quantum gravity in the bulk can be considered semi-classically. On the boundary,  this leads to various essential simplifications of  conformal correlation functions expanded into conformal blocks. Namely, assuming that conformal dimensions  grow at large $c$ either as $c^0$ or $c^1$, which corresponds to  light and heavy operators, the conformal blocks are drastically simplified. In the former case, the original Virasoro block at leading order in $c$ reduces to the global block which is completely defined by $\sl2$ subalgebra of Virasoro algebra \cite{Ferrara:1974ny,Belavin:1984vu}. This is the simplest approximation because the (multi-point) global blocks are hypergeometric-type functions \cite{Ferrara:1974ny,Belavin:1984vu,Fateev:2011qa,Alkalaev:2015fbw,Rosenhaus:2018zqn,Fortin:2019zkm,Fortin:2020zxw,Fortin:2023xqq,Alkalaev:2023evp}.\footnote{The subleading corrections in $1/c$ to global blocks were studied in e.g. \cite{Hikida:2018dxe,Bombini:2018jrg}.} In the later case, the conformal blocks are exponentiated with the exponent being a linear function of $c$  \cite{Zamolodchikov1986,Hadasz:2005gk,Litvinov:2013sxa,Besken:2019jyw}.

The large-$c$ conformal blocks in CFT$_2$  can be  considered from many perspectives, most prominently in the context of the AdS/CFT correspondence. E.g.  they were shown to calculate Witten geodesic diagrams and lengths of geodesic networks anchored on the conformal boundary of AdS$_3$ space with defects \cite{Hartman:2013mia,Fitzpatrick:2014vua,Hijano:2015rla,Fitzpatrick:2015zha,Alkalaev:2015wia,Hijano:2015qja,Hijano:2015zsa,Alkalaev:2015lca,Alkalaev:2015fbw,Banerjee:2016qca,Anous:2016kss,Alkalaev:2016rjl,Chen:2016cms,Chen:2017yze,Belavin:2017atm,Kraus:2017ezw,Alkalaev:2017bzx,Gobeil:2018fzy,Hijano:2018nhq,Alkalaev:2018qaz,Chen:2018qzm,Alkalaev:2018nik,Parikh:2019ygo,Anous:2019yku,Alkalaev:2019zhs,Chen:2019hdv,Jepsen:2019svc,Alkalaev:2020kxz,Anous:2020vtw,RamosCabezas:2020mew,Pavlov:2021lca}. Also, one can show that  gravitational Wilson networks with boundary attachments in $sl(N, \RR)\oplus sl(N, \RR)$ Chern-Simons higher-spin  gravity calculate  global conformal blocks \cite{deBoer:2013vca,Ammon:2013hba,deBoer:2014sna,Hegde:2015dqh,Bhatta:2016hpz,Besken:2016ooo,Besken:2017fsj,Hikida:2017ehf,Hikida:2018eih,Hikida:2018dxe,Besken:2018zro,Bhatta:2018gjb,DHoker:2019clx,Castro:2018srf,Kraus:2018zrn,Hulik:2018dpl,Hung:2018mcn,Castro:2020smu,Alkalaev:2020yvq,Belavin:2022bib,Belavin:2023orw,Alkalaev:2023axo}.

\subsection{$Vir$ classical conformal blocks}

Let us discuss the exponentiated Virasoro conformal blocks in more detail. We will focus on the case of the correlation  function $\langle \cO_1\cO_2\cO_3\cO_4 \rangle $ of four identical primary operators with conformal dimensions $h$ in plane CFT$_2$. Then, isolating the conformal block contribution of an intermediate dimension $\tilde h$ and stripping off the leg factor one is left with the  bare conformal block which we denote $\cF(h, \tilde h, c|z)$, see fig. \bref{fig:vir}. 

\begin{figure}[h]
    \centering
    \begin{tikzpicture}[scale = 0.4]
\draw[black, thick] (-2,0) -- (2,0);
\draw[black, thick] (-4,2) -- (-2,0);
\draw[black, thick] (-4,-2) -- (-2,0);
\draw[black, thick] (4,2) -- (2,0);
\draw[black, thick] (4,-2) -- (2,0);
\draw (-5,2.5) node{$\cO_h$};
\draw (-5,-2.5) node{$\cO_h$};
\draw (5,2.5) node{$\cO_h$};
\draw (5,-2.5) node{$\cO_h$};
\draw (0,0.9) node{$[\cO_{\tilde h}]$};

\end{tikzpicture}
    \caption{4-point Virasoro block diagram for four identical operators $\cO_h$ exchanging by a conformal family $[ \cO_{\tilde h}]$. }
    \label{fig:vir}
\end{figure}
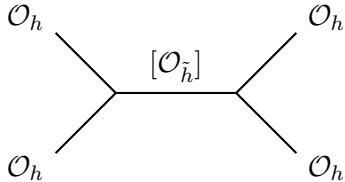

Assuming that the conformal dimensions grow as $h, \tilde h \sim c$  one finds the  large-$c$ asymptotics of the Virasoro conformal block  \cite{Zamolodchikov1986}: 
\be
\label{stand_block}
\cF(h, \tilde h, c|z) \; \simeq \;  e^{c f_1(z)+f_0(z)} \left[1+ \frac{a_1(z)}{c}+ \frac{a_2(z)}{c^2}+ O(c^{-3})\right],
\ee
where $f_1(z)$ is known as the classical conformal block, $f_0(z)$ is usually not indicated in the exponent in which case the asymptotic series in the square brackets is redefined and begins from some zeroth-order function $a_0(z)$. We restore $f_0$ for future convenience, in particular,  the $O(1/c)$ series in the brackets then starts with 1. The exponents  $f_{0,1}(z)$ and the expansion coefficients $a_n(z)$ additionally depend on the constant ratios $h/c, \tilde h/c$. The issue of  convergence of  \eqref{stand_block} along with summation  over all conformal blocks in the correlation function -- that turns  the large-$c$ expansion of conformal correlators into trans-series -- this requires careful analysis which leads to profound non-perturbative physical implementations on the both sides of the AdS$_3$/CFT$_2$ correspondence  \cite{Fitzpatrick:2016ive,Fitzpatrick:2016mjq,Chen:2017yze,Benjamin:2023uib}. 

In this technical (though still voluminous) note, we revisit the large-$c$ asymptotics of the  conformal blocks and study kinematical regimes when  operators have  dimensions which scale  with the central charge in a non-standard manner  as\footnote{For our notation and conventions see Appendix \bref{app:notation}.} 
\be
\label{a_dim}
h  = O(c^{\alpha})\,,
\quad
\text{where}\;\;\;  \alpha \in \QQ\,.
\ee
Such a regime of operator dimensions has been  discussed previously in \cite{Fitzpatrick:2014vua} for $\alpha = 1/2$ for the 4-point vacuum  block in Virasoro CFT$_2$ and in \cite{Anous:2020vtw} for $\alpha=1/2, 2/3$ for the six-point vacuum  block in the snowflake channel (where these operators were called ''hefty" and ''heftier" ones).\footnote{The vacuum (identity) block means  an exchange by the identity operator conformal family.} It was shown there  that the respective large-$c$ conformal vacuum  blocks  are exponentiated. Here, we show that these observations can be systematically extended to any $\alpha \in [0, \infty)_\rn$ and non-vacuum  blocks. 

In the large-$c$ Virasoro CFT$_2$ one may consider four types of  operators: 
\be
\label{four}
\ba{l}
h \sim \,1: \;\;\text{light operators, $\alpha=0$}
\vspace{1mm}
\\
{1 \ll }\; h\ll c: \;\;\text{hefty operators, $0<\alpha<1$}
\vspace{1mm}
\\
h \sim \, c: \;\;\text{heavy operators, $\alpha=1$}  
\vspace{1mm}
\\
h \gg c: \;\;\text{huge operators, $\alpha>1$}
\ea
\ee
which can be  collectively referred to as $\alpha$-heavy operators and $\alpha\in \mathbb{Q}_0^+$ can be  called a heaviness parameter. Complementary, one can  introduce $h = O(c^{-\beta})$ with $\beta \in \mathbb{Q}^+$ which means that the conformal dimensions are vanishing at large-$c$ and the respective conformal block $\cF(h, \tilde h, c|z) = 1+  O(c^{-\beta})$. Calling this type of  operators ''weightless" they can be depicted together with other types \eqref{four}  as in Fig. \bref{fig:alpha-line}.
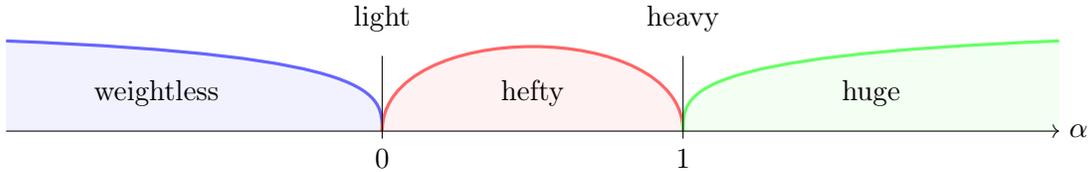
\begin{figure}
    \centering
    \begin{tikzpicture}
        \fill[ fill=blue!5]
            (-3,0) -- (-3,1.2)
            .. controls (2,1) and (2,0.5) .. (2,0) -- cycle;
        \draw[color = blue!60, very thick]
            (-3,1.2)
            .. controls (2,1) and (2,0.5)  .. (2,0);
        \fill[ fill=red!5]
            (2,0) .. controls (2,1.5) and (6,1.5) .. (6,0) -- cycle;
        \draw[color = red!60, very thick]  
            (2,0) .. controls (2,1.5) and (6,1.5) .. (6,0);
        \fill[ fill=green!5]
            (6,0) .. controls (6,0.5) and (6,1)  .. (11,1.2) -- (11,0) -- cycle;
        \draw[color = green!60, very thick] (6,0) .. controls  (6,0.5) and (6,1) .. (11,1.2);
        \node at (-1, 0.5) {weightless};
        \node at (4,0.5) {hefty};
        \node at (8.5,0.5) {huge};
        
        \draw[->] (-3,0) -- (11,0) node[right] {$\alpha$};
        
        \draw (2,1) -- (2,-0.1) node[anchor=north] {0};
        \draw (6,1) -- (6,-0.1) node[anchor=north] {1};
        
        \node at (2, 1.5) {light};
        \node at (6, 1.5) {heavy};

    \end{tikzpicture}
    \caption{Types of heaviness. }
    \label{fig:alpha-line}
\end{figure}

We conjecture  that  the large-$c$ asymptotics of the bare Virasoro conformal block for $\alpha$-heavy operators with rational $\alpha = p/q \in \QQ$ reads as 
\be
\label{alpha_block}
\cF(h, \tilde h, c|z) \; \simeq \;  e^{\cS^{(\alpha)}(c|z)} \left[1+ \frac{b_1(z)}{c^{1/q}}+ \frac{b_2(z)}{c^{2/q}}+ O(c^{-(3/q)})\right],
\ee
where the exponent is  {\it a Puiseux polynomial}  which is a linear combination of power functions in the central charge with decreasing non-negative  fractional degrees, 
\be
\label{cl_bl_S}
\cS^{(\alpha)}(c| z) = c^{\alpha}f^{(\alpha)}_{\alpha}(z) + \, ... \,+ c^{0}f^{(\alpha)}_0(z) = \sum_{n=0}^p c^{(p-n)/q} f^{(\alpha)}_{(p-n)/q}(z)\,.
\ee
In general, it is this function that we call the {\it ($\alpha$-heavy)  classical conformal block}. At $\alpha=1$ we keep the old terminology and call the coefficient $f^{(1)}_{1}(z)$ of $\cS^{(\alpha)}(c| z)$ as the {\it Zamolodchikov classical conformal  block}. At $\alpha = 0$ the asymptotic representation is superfluous since the 4-point Virasoro conformal block $\cF(h, \tilde h, c|z) $ is known to be the Gauss hypergeometric function in which case it is called the {\it global conformal block}.          

The large-$c$ representation \eqref{alpha_block}, \eqref{cl_bl_S}  can be extended to any real positive $\alpha$. In fact, our method makes no essential distinction between (ir)rational $\alpha$.   Note that for $\alpha\in \mathbb{R}_+ $ the infinity $c=\infty$ is an algebraic/logarithmic branch point. This means that instead of the Laurent series the large-$c$ expansion of the (logarithmic) conformal block is given by the Puiseux series which is the generalized Laurent series with fractional or real powers. The leading expansion coefficient in  \eqref{cl_bl_S}  can be defined as $f_{\alpha}(z) = \lim\limits_{c\to\infty} c^{-\alpha}\log \cF(\td, h_i , c\,| z)$. To find subleading terms one needs a  Cauchy integral formula modified by including   contours  which avoid  the branch cuts (for irrational $\alpha$).

\begin{figure}
\centering
\includegraphics[width=0.45\linewidth]{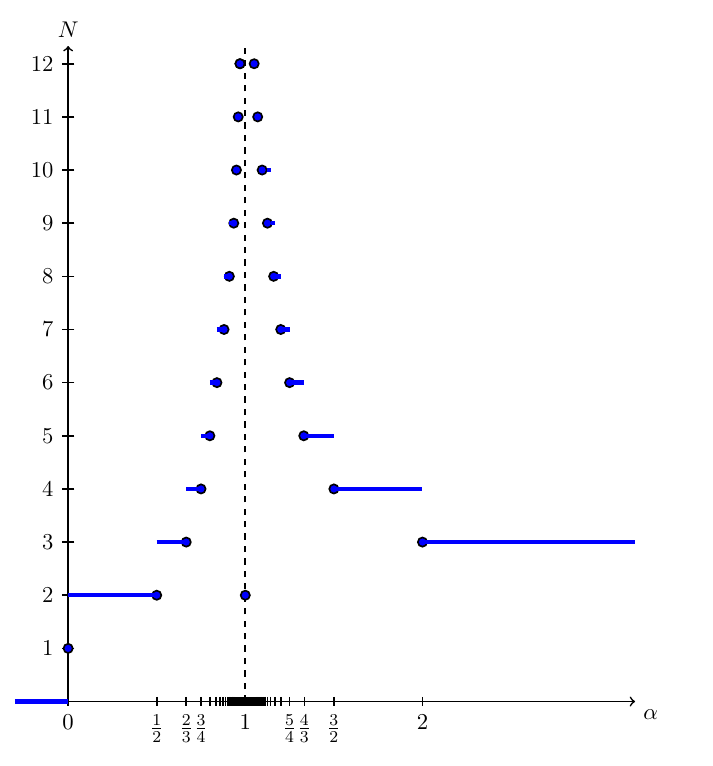} 
\caption{A counting  function $N = N(\alpha)$: a number of terms $N$ in the classical block as a function of the heaviness parameter $\alpha$. It is piecewise-defined, the bold dots denote values of $N(\alpha)$ at points of discontinuities. } 
\label{fig_graphic}
\end{figure}

The powers in \eqref{cl_bl_S} decrease with  step $1/q$ but the corresponding expansion coefficients $f^{(\alpha)}_{(p-n)/q}$ are not necessarily all non-zero. The number of terms in $\cS^{(\alpha)}(c| z)$ crucially depends on particular $\alpha = p/q$.  For $\alpha = 1$ which corresponds  to the standard heavy operators with $h/c, \tilde h/c = const$ the exponent  contains two terms, cf. \eqref{stand_block}. In fact, the problem is not only to show that the $\alpha$-heavy conformal block exponentiates but also to reveal the structure of the leading exponent $\cS^{(\alpha)}(c| z)$ depending on  $\alpha$. 

We demonstrate  that the number of terms in $\cS^{(\alpha)}(c| z)$ grows up from two to infinity when $\alpha$ increases from $0$ to $1^-$, then, at $\alpha=1$ this number is sharply reduced to two,  then again it sharply blows up to infinity at  $\alpha \to 1^+$, after which gradually decreases and for $\alpha\geq 2$ the number of terms in $\cS^{(\alpha)}(c| z)$ is stabilized  at three. Moreover, it is shown that splitting the half-line $\alpha\in [0, \infty)_\rn$ into an infinite union of particular half-open intervals the number of terms does not change within a given interval. The change in the number of terms occurs only when $\alpha$ goes from a given interval to  an adjacent one, see fig. \bref{fig_graphic}.

The evidence in support of  \eqref{alpha_block} is based on both  numerical computation of the conformal blocks for $\alpha \in \QQ$ along with analytical treatment of the lower level block coefficients.\footnote{\label{foo1} On the standard laptop during 12 hours our {\it Mathematica} code computes (for equal external dimensions):  non-vacuum Virasoro conformal block coefficients up to $O(z^{6})$; vacuum Virasoro block coefficients up to $O(z^{22})$;  classical (non-)vacuum conformal block coefficients up to $O(z^{16})$. The calculation  time for higher-order coefficients is much longer.} In particular, using {\it Mathematica} code we reproduce (up to $O(z^{16})$) the exponentiated vacuum  block analytically found in \cite{Fitzpatrick:2014vua} which  in our nomenclature corresponds to $\alpha=1/2$:  in this case the exponent has two terms, $\cS^{(1/2)}(c| z) = c^{1/2} f^{(1/2)}_{1/2}(z)+ f^{(1/2)}_0(z)$, where $f^{(1/2)}_{1/2} =0$ and $f^{(1/2)}_0$ is given by the hypergeometric series. In particular, in this example one can observe a lowering  the order of the pole in $c$. It is  a general phenomenon  which  we call a transmutation of singularities since  for some  specific  parameters $\{\alpha, h_i, \tilde h\}$ a singular point $c= \infty$  changes its type.

One of the major  consequences of  introducing $\alpha$-heavy operators is that the perturbative classical block at $\alpha=1$, which is obtained by making constant ratios $h/c, \tilde h/c$  (large) small, can be interpreted in terms of the classical blocks $\cS^{(\alpha)}(c| z)$ at $\alpha \neq 1$.  More precisely, the $\alpha$-heavy operators  \eqref{four} with $\alpha \in (0,1)_\rn \cup (1, \infty)_\rn$ define classical blocks which coefficients $f^{(\alpha)}_{m}(z)$ \eqref{cl_bl_S}  specifically reconstruct  coefficients of the Zamolodchikov  classical conformal block which is expanded near $0$ or $\infty$ in the space of conformal dimensions in the large-$c$ regime. In this way,  perturbative expansions of the Zamolodchikov classical block can be systematically described in terms of  considering contributions of $\alpha$-heavy operators.   

\subsection{$\cW_3$ classical conformal blocks}

The large-$c$ $\cW_3$ conformal blocks are much less studied in the literature. Some partial results basically concerning vacuum  blocks can be found in \cite{deBoer:2014sna,Hegde:2015dqh,Besken:2016ooo,Coman:2017qgv,Hulik:2018dpl,Karlsson:2021mgg,Pavlov:2022irx}, where the classical  blocks were considered  as solutions to the associated monodromy problem and in \cite{Fateev:2011qa,Belavin:2023orw,Belavin:2024nnw}, where the global blocks were studied. In fact, the real issue is deeper since $\cW_3$ conformal blocks cannot be uniquely defined since the extended conformal symmetry does not uniquely fix 3-point functions of descendants  in terms of 3-point functions of primaries. The leftover freedom reduces to calculating matrix elements of  the $(W_{-1})^n$ descendant operator which  is not generally known (here, $W_{-1}$ is the Laurent mode of the spin-$3$ current $W(z)$ and $n\in \mathbb{N}$). This difficulty can be overcome by imposing an extra constraint $W_{-1} \sim L_{-1}$ which is in fact the null-condition defining a singular vector on the first level in the  Verma module \cite{Bowcock:1993wq}.\footnote{See e.g. a  discussion of the conformal block decomposition  in Refs. \cite{Mironov:2009dr,Kanno:2010kj,Belavin:2016qaa}.} A respective local primary operator $\cO_{h,w}$ is called semi-degenerate, its conformal dimension $h$ and spin-3 charge $w$ are polynomially related as
\be
\label{degW0}
\left[\frac{32}{22+5c}\left(h+\frac{1}{5}\right)-\frac{1}{5}\right]h^2 = w^2\,.  
\ee
It follows that considering a 4-point correlation function $\langle \cO_1\cO_2\cO_3\cO_4 \rangle$, where all operators are semi-degenerate, one can consistently define the conformal block decomposition. 
\begin{figure}[h]
\centering
\begin{tikzpicture}[scale = 0.4]
\draw[black, thick, double, double distance=0.5mm] (-2,0) -- (2,0);
\draw[black, thick] (-4,2) -- (-2,0);
\draw[black, thick] (-4,-2) -- (-2,0);
\draw[black, thick] (4,2) -- (2,0);
\draw[black, thick] (4,-2) -- (2,0);
\draw (-5,2.5) node{$\cO_{h,w}^*$};
\draw (-5,-2.5) node{$\cO_{h,w}$};
\draw (5,2.5) node{$\cO_{h,w}$};
\draw (5,-2.5) node{$\cO_{h,w}^*$};
\draw (0,0.9) node{$[\cO_{\tilde h, 0}]$};

\end{tikzpicture}
\caption{4-point $\cW_3$ conformal block diagram. All external operators are semi-degenerate. The intermediate conformal family  has a primary non-degenerate operator of zeroth spin-3 charge. The double line notation is explained in Section  \bref{sec:mat_block}.}
\label{fig:comb_W3}
\end{figure}
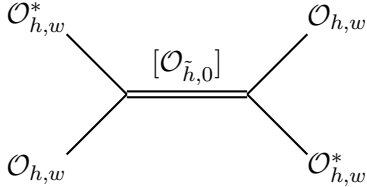

To keep things simple, in this paper we consider a particular correlation function 
\be
\label{4ptW}
\big\langle \cO(z_1)\cO^*(z_2)\cO(z_3)\cO^*(z_4) \big\rangle\,,
\ee
where $\cO$ and $\cO^*$ are dual semi-degenerate operators which means their weights are $(h,w)$ and $(h,-w)$ \eqref{degW0}. 
The exchange channel is fixed by considering the fusion rules arising in the OPEs $\cO\cO^*$ and $\cO^*\cO$. One can  show that the intermediate operator can be chosen  as a non-degenerate operator $\tilde \cO_{\tih,\tiw}(z)$ with (independent) weights  $\tih$ and $\tiw = 0$, i.e. 
\be
[\cO_{h,w}] [\cO_{h,-w}] = [\tilde \cO_{\tilde h,0}]\,.
\ee 

The  respective  (bare) block denoted by $\wb(h,w, \tilde h,c|z)$ is a power series in $z$ with expansion coefficients  depending on three independent  parameters: conformal dimension $h$  (for semi-degenerate operators the weights  are related,  $w= w(h)$),  conformal dimension $\tilde h$;  the central charge $c$, see fig. \bref{fig:comb_W3}.

The issue of the large-$c$ asymptotics is more involved because  $\cW_3$ conformal  operators have two  weights $h$ and $w$ which can be assigned different behaviour at large $c$. In general, one  can claim that for large $c$ the weights are given by 
\be
\label{ab_dim}
h = O(c^\alpha)\,,
\quad
w = O(c^\beta)\,,
\quad
\text{where}\;\;\;  \alpha, \beta \in \QQ\,,
\ee
cf. \eqref{a_dim}.  It follows that a given $\cW_3$ conformal operator  is $\alpha\beta$-heavy which means that  depending on particular values of $\alpha$ and $\beta$ there are  8 types of heaviness  obtained by applying  the list \eqref{four} for two independent weights $h$ and $w$:
\be
\label{ab_dim_h}
\alpha\beta\text{-heaviness}:\quad (\mathbb{A}, \mathbb{B})\,,
\qquad \text{where} \quad \{\mathbb{A}, \mathbb{B} = \text{light, heavy, hefty, huge}\}\;.
\ee
However, in this paper the spin-3 charges  are either functions of the conformal dimensions, i.e. $w = w(h)$, or just equal zero, $w=0$. Thus, all weights are characterized  by a single scaling parameter $\alpha\in \QQ$ so that  we can  keep calling the respective operators $\alpha$-heavy. 

It turns out that  for such $\cW_3$ conformal blocks  both the calculation of the expansion coefficients and the large-$c$ analysis go along the same lines as in the Virasoro case. However, our  {\it Mathematica} code allows finding   only lower-order block coefficients (up to $O(z^4)$)  in a reasonable time (see footnote \bref{foo1}). Nonetheless, they also demonstrate the same exponentiation pattern. Namely, for $\alpha$-heavy operators with rational $\alpha = p/q \in \QQ$ the $\cW_3$ conformal block for large $c$ is asymptotic to 
\be
\label{expW}
\wb(h,w, \tih, c|z)  \; \simeq \;  e^{\ws^{(\alpha)}(c|z)} \left[1+ \frac{d_1(z)}{c^{1/q}}+ \frac{d_2(z)}{c^{2/q}}+ O(c^{-(3/q)})\right],
\ee
where the leading exponent $\ws^{(\alpha)}(c|z)$ which is now {\it a  classical $\cW_3$ conformal block} 
\be
\label{cl_bl_T}
\ws^{(\alpha)}(c| z) = c^{\alpha}\,\cb^{(\alpha)}_{\alpha}(z) + \, ... \,+ c^{0}\,\cb^{(\alpha)}_0(z) = \sum_{n=0}^p c^{(p-n)/q}\, \cb^{(\alpha)}_{(p-n)/q}(z)
\ee
remains to be a Puiseux polynomial function of the central charge with coefficients $\cb^{(\alpha)}_{(p-n)/q}(z)$ depending on classical dimensions $h/c^{\alpha}$ and $\tih/c^{\alpha}$ (recall that $w$ is dependent on $h$). We show that the structure of $\ws^{(\alpha)}(c| z)$ (i.e. the number of terms and their growth with $\alpha$) obeys the same patterns which take place in the Virasoro case, see fig. \bref{fig_graphic}. Of course, these findings apply only to the $\cW_3$ conformal blocks with particular external/intermediate operators as fixed above  and for more general settings the structure of $\ws^{(\alpha)}(c| z)$ can change.

The paper is organized as follows. In Section \bref{sec:vir} we review the Virasoro blocks in the  large-$c$ regime and formulate  the procedure to calculate the asymptotics for $\alpha$-heavy operators. Later, in Section \bref{sec:w3} we extend this discussion to the $\cW_3$ conformal blocks. Section \bref{sec:pert} discusses a relation between perturbative Zamolodchikov classical conformal blocks and  classical conformal blocks for $\alpha$-heavy operators.  Our findings and future developments are discussed in the concluding Section \bref{sec:concl}. Notation and conventions are given in Appendix \bref{app:notation}. Appendix \bref{app:vir} contains analytical treatment of the lowest order coefficients of the Virasoro logarithmic blocks.  We discuss in detail how to expand Puiseux rational functions of the central charge near infinity and prove a number of Propositions collected  in Section \bref{sec:vir}. Also, we  develop the Newton  polygon technique which considerably simplifies and structures the study of the large-$c$ asymptotics. In Appendix \bref{app:W3} we apply this technique to the study of the $\cW_3$ logarithmic block coefficients. Appendix \bref{app:proof_prop} contains the proof of the  Proposition from Section \bref{sec:pert}. In Appendices  \bref{app:W3-deg} and \bref{app:B1} we discuss the fusion rules for semi-degenerate $\cW_3$ primary operators and calculate the first block coefficient, respectively. Appendix \bref{sec:numerics} contains a few  examples of the exponentiated Virasoro and $\cW_3$ conformal blocks in the lowest orders for various $\alpha$ which are intended to illustrate  Propositions in this paper by explicit expressions. 

\section{Virasoro conformal blocks}
\label{sec:vir}

The (bare) conformal block is given as a power series in coordinate $z$ in the domain $|z|<1$ as 
\be
\label{block_def}
\cF(\td, h_i , c\,| z) = 1+ F_1 z+F_2 z^2 +F_3 z^3+... \;=\; \sum_{n=0}^\cN F_n z^n+ O(z^{\cN+1})\,, 
\ee     
where the coefficients  $F_n = F_n(\td, h_i , c)$ are  rational functions which can be calculated in a given order \cite{Belavin:1984vu}. In practice, we take $z$ to be small and calculate the block function at finite order. Also, all external dimensions are chosen to be equal, $h=h_i$. Then, using {\it Mathematica} we can explicitly  calculate the first six coefficients  $F_{n}(h,\tilde h,c)$, $n=0,1,...,5$ which we later analyze analytically: 
\be
\label{General block}
\cF (h, \tilde h, c, z) = 1 + \frac{\tilde h}{2} z + \frac{c \tilde h (\tilde h + 1)^2 + 8 \big( {\tilde h}^4 + {\tilde h}^3 + {\tilde h}^2 (h-1) + h^2 + \tilde h h (2h-1) \big) }{8\tilde h(8\tilde h - 5) + c (8\tilde h + 4)} z^2 + 
\ee
$$
+ \frac{(\tilde h + 2) \Big( c \tilde{h} (\tilde{h}+1)(\tilde{h}+2)+2\tilde{h}^2 \big( \tilde{h} (4\tilde{h}+9)-7 \big) + 24 \tilde{h} h (\tilde{h}-1) + 24 h^2 (2\tilde{h}+1) \Big)}{48\tilde h(8\tilde h - 5) + 6 c (8\tilde h + 4)} z^3 + O(z^4)\,. 
$$

For fixed dimensions $h, \tilde h$ the Virasoro conformal block $\cF (h, \tilde h, c, z)$ has   a certain large-$c$ asymptotics which leading term is the global conformal block given by the hypergeometric function. On the other hand, if conformal dimensions are particular functions of the charge, then $\cF (h, \tilde h, c, z) \to \cF (c, z)$  can have different behaviour in $c$ and, hence, different large-$c$ asymptotics. A textbook example is given by the conformal block at the  Kac dimensions $h_{r,s} = h_{r,s}(c)$ which are simple poles of the block coefficients \cite{Zamolodchikov:recursion,Zamolodchikov1987}.\footnote{For equal external dimensions the poles arising in \eqref{General block} have all Kac dimensions except for $h_{1,1} =0$. Thus, the zeroth intermediate dimension is not a pole that makes  the vacuum  conformal block $\cF (h, c, z)$ possible. The large-$c$ behaviour of the conformal block with intermediate Kac dimensions will be considered in Section \bref{sec:sing}.}

\subsection{Large-$c$ conformal dimensions}
\label{sec:large-c}

Let   conformal dimensions  $\Delta  \equiv \{h, \tilde h\}$ and the central charge  $c$ be large. To consider the respective asymptotics of the conformal block one supposes that conformal dimensions are functions of the central charge which at $c\to\infty$  behave   as $\Delta(c) = O(c^\alpha)$ for some $\alpha\in\QQ$.\footnote{Equivalently, $\Delta/c^\alpha= fixed$ at $c\to\infty$ (recall that the big $O$ notation stands for a set of functions, not a single function).} Explicitly, parameterizing $\alpha = p/q$ this  means that $\Delta(c)$ can be expanded    near $c=\infty$ into the Puiseux series (a modified Laurent series containing fractional exponents) as 
\be
\label{delta_puis}
\Delta(c)  = \sum_{n=-p}^\infty \delta_{(n)} c^{-n/q}  \equiv  \pr\, \Delta(c) + \reg\, \Delta(c)\,,
\ee
with  constant  expansion coefficients $\delta_{(n)} \in  \mathbb{R}$. For future convenience, we define the principal part as containing a possible $c^0$ term:  
\begin{definition}
\label{def_prin}
The principal part $\pr$ is a Puiseux polynomial with non-negative powers of $c$. The regular part $\reg$ is a Puiseux series with negative powers of $c$. 
\end{definition}

It is assumed that $c=\infty$ is an isolated singularity so that the principal part has finitely many terms in which case it is an algebraic branch point of order $q$. E.g. for $\alpha = 1$ one can consider $\Delta(c) = \delta_{(-3)} c^1 + \delta_{(-1)} c^{1/3} + \delta_{(0)} c^0+ \reg\, \Delta(c)$. We will be studying  the simplest case when the principal parts contain just one term, i.e. the conformal dimensions are power-low functions of the central charge:   
\be
\label{h_th_a}
h(c) = \delta c^\alpha\,, \qquad \tilde h(c)  =\tdelta c^\alpha\,, \qquad \alpha \in \QQ
\;.
\ee
The coefficients $\delta, \tdelta$ will be called classical conformal dimensions. More complicated principal parts (as in the above example) as well as regular parts  will change the large-$c$ exponential behaviour of the Virasoro block.

\subsection{Large-$c$ Virasoro conformal block} 

Choosing the conformal dimensions as \eqref{h_th_a}  one finds that the  large-$c$ asymptotic expansion of the $n$-th level coefficient \eqref{block_def}  behaves  as $F_n  = O(c^{\alpha n})$ (this is correct at least up to $O(z^6)$):
\be
F_n(h, \tilde h, c) = \sum_{m=-np}^\infty F_{n|m}(\delta, \tilde \delta)\, c^{-m/q}\,.
\ee
Then, the Virasoro  conformal block is a double power series which can be  reorganized as the Puiseux series in  $c$ as
\be
\label{asymp}
\ba{c}
\dps
\cF(h, \tilde h, c\,| z) = \sum_{m=-\cN p}^\infty G_m(\delta, \tilde \delta|z)\, c^{-m/q} \equiv \pr \cF(\delta, \tilde \delta, c\,| z) + \reg\cF(\delta, \tilde \delta, c\,| z)\,,
\ea
\ee     
where the expansion coefficients $G_m(\delta, \tilde \delta|z)$  are  power series in $z$ calculated up to $O(z^{\cN+1})$.  Sending $\cN\to \infty$ yields the full-fledged  large-$c$ asymptotic expansion of the Virasoro conformal block. The infinity $c=\infty$ being generally a branch point of finite order for  $\cF\simeq \sum_{n=0}^\cN F_n z^n$  evolves into an essential singularity  since the principal part in  \eqref{asymp} acquires infinitely many terms.

In order to conveniently characterize an essential  singularity one may try to find a functional map of the original function which, instead, will have a pole or branch point of finite order. A textbook example here is  $\sum_{n=0}^\infty z^{-n}/n! = e^{1/z}$, where an essential singularity $z=0$ is packed into a simple pole of the exponent. Having this in mind one introduces the logarithmic conformal block  
\be
\label{def}
f(h, \tilde h, c\,| z) \coloneqq \log[\cF(h, \tilde h , c\,| z)]\,
\ee         
and studies the principal part $\pr f(h, \tilde h, c\,| z)$ choosing  $h,\tilde h$ as  particular functions of the central charge with $O(c^\alpha)$ asymptotic behaviour. Equivalently, the original conformal block can be cast into the exponential form
\be
\label{def2}
\cF(c\,| z) = e^{\pr f(c| z)}\left[1 + O(1/c^\beta)\right]\,,
\qquad \exists\, \beta >0\,,
\ee   
which is valid near $c=\infty$. We stress that the form of $\pr f(c | z)$ -- the order of singularity and the number of terms   -- crucially  depends on choosing one or another $c$-dependence of conformal dimensions \eqref{delta_puis}. E.g. choosing $h, \tilde h = c$ and $h, \tilde h  =  c+c^{1/2}$  which both belong to the $O(c)$ class of functions will produce different exponents $\pr f(c| z)$ as well as  regular parts in the square brackets.

Finally, comparing two large-$c$ representations \eqref{asymp} and \eqref{def2} one notes that knowing $\pr f$ does not completely reproduce $\pr \cF$ since the exponential gives  arbitrarily  large degrees of $c$ which can be compensated by arbitrarily small degrees from the square brackets in \eqref{def2} to yield any finite positive degrees of $c$ in \eqref{asymp}. Moreover, all expansion coefficients in $c$ in the both representations are functions of classical dimensions that leads to the issue of controlling their growth/decrease compared to that of $O(c^\rho)$ and $O(1/c^\sigma)$ for $\rho,\sigma>0$. Leaving this issue aside we will be treating  the large-$c$ expansions as formal (Puiseux) power series. From now on, we define a large-$c$ asymptotics of the Virasoro conformal block as the exponential representation \eqref{def2}.

\subsection{Large-$c$ logarithmic conformal block}
\label{sec:large_c}

Substituting  the series \eqref{block_def} into the logarithmic conformal block  \eqref{def}  one  obtains the series  \cite{Zamolodchikov1987} 
\be
\label{asymp_f}
f(h, \tilde h, c\,| z) = \sum_{n=1}^\cN g_n(h, \tilde h, c) z^n + O(z^{\cN+1})\,, 
\quad \text{where} \quad g_n(h, \tilde h, c) = \frac{p_n(h, \tilde h, c)}{q_n(\tilde h, c)}
\ee 
are rational functions of the conformal dimensions and the charge  which  directly follow from the Mercator series  formula. Numerical calculation yields the logarithmic Virasoro block: 
\be
\label{log_block}
\ba{l}
\dps
f (h, \tilde{h} , c\,| z) =  \frac{\tilde h}{2} z + \frac{\tilde h (3 \tilde h +2)c +2\Big( {\tilde h}^2 (13 \tilde h -8) + 8 h \tilde h (\tilde h -1) + 8 h^2 (2 \tilde h + 1) \Big)}{8 \big( c + 2 c \tilde h + 2 \tilde h (8 \tilde h -5) \big) } z^2 
\ea
\ee
$$
\dps
\hspace{20mm}+ \frac{\tilde h (5 \tilde h +4)c + 2 {\tilde h}^2 (23 \tilde h -14) + 48 h \tilde h (\tilde h -1) + 48 h^2 (2 \tilde h +1)}{24 \big( c + 2 c \tilde h + 2 \tilde h (8 \tilde h -5) \big)} z^3 + O(z^4)\,,
$$
where we omitted the forth and fifth coefficients. Now,  choosing  conformal dimensions as  \eqref{h_th_a} one  obtains from \eqref{log_block} the $\alpha$-heavy  logarithmic conformal block,   
\be
\label{f_exp_alpha}
f^{(\alpha)}(\delta, \tilde{\delta},c\,| z)  =  \sum_{n=1}^\cN g^{(\alpha)}_n(h, \tilde h, c) z^n + O(z^{\cN+1}) 
\ee
\vspace{-2mm} 
$$
\ba{c}
\dps
=\frac{\tilde \delta c^{\alpha}}{2}  z 
+ \frac{(26 {\tilde \delta}^3 + 16 {\tilde \delta}^2 \delta + 32 \tilde \delta \delta^2)c^{3 \alpha} + 16({\delta^2 - \tilde \delta \delta - {\tilde \delta}^2)c^{2 \alpha} + 3 {\tilde \delta}^2 c^{2 \alpha + 1} + 2 \tilde \delta c^{\alpha + 1}}}{128 {\tilde \delta}^2 c^{2 \alpha}+16 {\tilde \delta}c^{\alpha+1} - 80 {\tilde \delta} c^\alpha + 8 c} z^2 
\vspace{3mm}
\\
\dps
+ \frac{2 \tilde \delta (23 {\tilde \delta}^2 + 24 \tilde \delta \delta + 48 \delta^2)c^{3 \alpha} - 4 (7 {\tilde \delta}^2 + 12 \tilde \delta \delta - 12 \delta^2)c^{2 \alpha} + 5 {\tilde \delta}^2 c^{2 \alpha + 1} + 4\tilde \delta c^{\alpha + 1}}{24 \big( c - 10 \tilde \delta c^{\alpha} + 2 \tilde \delta c^{\alpha+1} + 16 {\tilde \delta}^2 c^{2 \alpha} \big)} z^3 + O(z^4)\,.
\ea
$$
Expanding this expression  near $c=\infty$ one represents the $\alpha$-heavy logarithmic conformal block $f^{(\alpha)}$ as the Puiseux series  
\be
\label{alpha_lor}
f^{(\alpha)} = \sum_{n\in \mathbb{D}_\alpha} f_{n}^{(\alpha)}\, c^n   
\equiv \pr f^{(\alpha)}  + \reg f^{(\alpha)}\,,
\ee 
where $\mathbb{D}_\alpha\subset \mathbb{Q}$ is the domain in rational numbers which is characterized by choosing a particular $\alpha \in  \QQ$. We are going to explicitly describe $\mathbb{D}_\alpha$ and show that the principal part in \eqref{alpha_lor} leads to the exponential representation \eqref{alpha_block}-\eqref{cl_bl_S}:
\be
\label{exp_prin}
\log[\cF(h, \tilde h , c\,| z)] \;\simeq \;  \pr f^{(\alpha)}(\delta, \tilde\delta, c|z) \equiv \cS^{(\alpha)}(\delta, \tilde\delta, c|z)\,,
\ee   
where $\simeq$ means that we omitted $\reg$-terms. Note that once the principal part has a term of the highest degree $\alpha = p/q$, where $p,q$ are not necessarily coprime,  then all other terms in $\pr f^{(\alpha)}$  will have degrees $m/q$, where $m \in \mathbb{Z}$, $m<p$. This is a simple consequence of the fact that any (finite) set of rational numbers can be represented as fractions with the same denominator. By introducing $c^{1/q}$ as a new variable a Puiseux series becomes a Laurent series in an $q$-th root of $c$. If the central charge is complex,\footnote{Complex central charges arise e.g. in the context of the dS/CFT  correspondence, see \cite{Witten:1989ip,Balasubramanian:2002zh}.} $c\in \mathbb{C}$, one finds that the classical block given by a Puiseux polynomial in fact defines $q$ different functions because a complex number has $q$ {different} $q$-th roots.

The first coefficient in \eqref{log_block} is  linear in $\tilde h$  so that $g_1^{(\alpha)}$ in \eqref{f_exp_alpha} already demonstrates that the highest degree term is $\sim c^\alpha$. Moreover, we see that representing $\tilde h$ as a general Puiseux series \eqref{delta_puis} yields for the logarithmic conformal block a pattern of $c$-dependence:
\be
\label{g_1_P}
g_1^{(\alpha)} = \half \left(\pr\, \tilde h(c) + \reg\, \tilde h(c)\right) 
= \sum_{n=0}^p \frac{\delta_{(-n)}}{2} c^{n/q} + O(c^{-1/q}) \, .
\ee   
In particular, the form of the exponent $\cS(\delta, \tilde\delta, c|z)$ \eqref{exp_prin} in the linear order in $z$ is completely defined by the first coefficient and cannot be changed by higher-order terms $O(z^2)$ \eqref{asymp_f}.  In the case of arbitrary   external conformal dimensions the first coefficient  of the Virasoro conformal block is given by   $F_1 = (\tilde h - h_1+h_2)(\tilde h +h_3-h_4)/2\tilde h$ so that  substituting the conformal dimensions as \eqref{delta_puis}  one  obtains  the same  type Puiseux series \eqref{g_1_P} (with different, however, series coefficients).

The second coefficient $g_2^{(\alpha)}$ in \eqref{f_exp_alpha} can be represented as
\be
\label{f_2ex}
g^{(\alpha)}_2 = \frac{A c^{3\alpha} + B c^{2\alpha+1} + C c^{2\alpha}+ D c^{\alpha+1} }{M c^{2\alpha} +N c^{\alpha+1}+K c^{\alpha}+L c^1}\,.  
\ee  
Here, the  principal part $\pr g_2^{(\alpha)}$ is more intricate, though, it is immediately seen that the highest-degree term here is again $\sim c^{\alpha}$. The third coefficient $g_3^{(\alpha)}$ takes the same form as \eqref{f_2ex} but with different coefficients $A,B,C,...\,$. The coefficients $g_4^{(\alpha)}$ and $g_5^{(\alpha)}$ are again rational functions with the same highest-degree term $c^\alpha$  but in this case the (de)nominators contain 27 terms each, see \eqref{g4_num}--\eqref{g4_den}.    

\subsection{Structure of $\cS^{(\alpha)} = \pr f^{(\alpha)}$}
\label{sec:structure}

This section describes the principal part of the large-$c$ logarithmic conformal blocks for  conformal dimensions given by \eqref{h_th_a}. Let us give a general view  of  the problem we solve. Generally, the logarithmic block coefficients are rational functions of the form  
\be
\label{FcPQ}
g(c) = \frac{\cP(c)}{\cQ(c)} = \frac{\dps\sum_{k,m \in D_{\cP}} p_{k,m}\, c^{k+m\alpha}}{\dps\sum_{s,t \in D_{\cQ}} q_{s,t}\, c^{s+t\alpha}}\;,  
\ee 
where  $\cP(c)$ and $\cQ(c)$ are some Puiseux polynomials, i.e. $\alpha \in \mathbb{Q}^+_0$  and $D_{\cP}$,  $D_{\cQ}$ are {\it the Newton polygons} \cite{chebotarev,gindikin} which are particular finite domains on a two-dimensional lattice  $\ZZ\oplus \ZZ$, coefficients $p_{k,m}, q_{s,t} \in \mathbb{C}$. Then, one  expands this function near $c=\infty$ into a Puiseux series as 
\be
\label{puiF}
g(c) = \pr g(c) + \reg g(c)= \sum_{ a\in D_{\hspace{-0.8mm}\pr}} p_a \, c^a + \reg g(c)\,,
\ee 
where the principal part is characterized by a domain $D_{\hspace{-0.8mm}\pr} \in \RR$ which again can conveniently be  described as a domain $\subset \ZZ\oplus \ZZ$  by parameterizing  $a = l+n  \alpha$.  {\it One seeks  to describe a domain $D_{\hspace{-0.8mm}\pr}$  and  a set of expansion coefficients $\{p_a, a \in D_{\hspace{-0.8mm}\pr}\}$ as functions of  $\alpha$}. The only thing that can be immediately seen  about $D_{\hspace{-0.8mm}\pr}$ is the order of the pole in the principal part. By contrast, identifying degrees of other terms in the principle part and, moreover, of those in the regular part requires an extended work.  

It turns out that the domain of $\alpha$ is naturally split into four regions, where the conformal blocks have essentially different properties:
\be
\label{domains}
\text{{\bf (I)}\; $\alpha = 0$; \;\;\;\;{\bf (II)}\; $\alpha\in (0,1)_\rn$;\;\;\;\; {\bf (III)}\; $\alpha = 1$;\;\;\;\; {\bf (IV)}\; $\alpha\in (1,\infty)_\rn$}\,,
\ee 
cf. \eqref{four}. These regions are shown in fig. \bref{fig:plot}, where  each $\alpha\in(0,1)_\rn$ and $\alpha\in(1,\infty)_\rn$ can be placed, respectively, into one of half-open intervals 
\be
\label{intervals}
\Big(\frac{M-1}{M}, \frac{M}{M+1} \Big]_\rn
\quad \text{and} \quad 
\Big[\frac{N+1}{N}, \frac{N}{N-1} \Big)_\rn
\ee
for some $M,N\in\mathbb{N}$. Note that the endpoints of  both intervals tend to 1 at $M,N \to \infty$ which means that their widths are shrinking with large $M,N$ and the intervals  themselves  condense near 1 (see fig. \bref{fig:plot}).

\begin{figure}
\centering
\begin{tikzpicture}[thick,scale=3]
\draw[thick,->] (0,0) -- (3.2,0) node[anchor=north west] {$\alpha$};
\foreach \x/\xtext in {0/0, 0.5/\frac{1}{2}, 0.666666/\frac{2}{3}, 0.75/\frac{3}{4}, 0.8/, 0.833333/, 0.857142/, 0.875/, 0.888888/, 0.9/, 0.909090/, 0.916666/, 0.923076/, 0.928571/, 0.933333/, 0.9375/, 1/1, 1.0625/, 1.066666/, 1.071428/, 1.076923/,  1.083333/, 1.090909/, 1.1/, 1.11111/,  1.125/, 1.142857/, 1.166666/, 1.2/, 1.25/\frac{5}{4}, 1.333333/\frac{4}{3}, 1.5/\frac{3}{2}, 2/2}{
   \draw (\x cm,0.75pt) -- (\x cm,-0.75pt) node[anchor=north] {$\xtext$};
   }
   \fill [black] (0.9375, -0.75pt) rectangle (1.0625, 0.75pt);
\end{tikzpicture}
\caption{Splitting the half-line $\alpha \in \mathbb{Q}^+_0$ into half-open intervals.}
\label{fig:plot}
\end{figure}
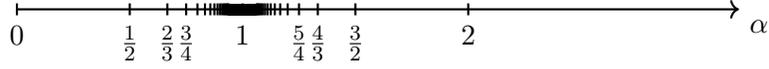

\begin{definition}
\label{def:reg}
The endpoints of the half-open intervals \eqref{intervals} denoted as  
\be
\tilde \alpha_N = \frac{N+1}{N}\;,
\qquad
\overset{\approx}{\alpha}_{M} =\frac{M}{M+1}\; 
\ee
are called exceptional. All other points $\alpha\in  \mathbb{Q}_0^+\backslash \{0,1\}$ are called regular. 
\end{definition}

In Appendix \bref{app:vir} we study  the Puiseux series \eqref{puiF} for  the first five Virasoro conformal block  coefficients and prove a number of Propositions  given below which structure the large-$c$ Puiseux asymptotics of the logarithmic conformal block $f^{(\alpha)}(h, \tilde h, c,z)$ up to $O(z^6)$ and up to $O(z^{16})$ in the case of  classical vacuum conformal blocks.\footnote{ Appendix \bref{app:vir}  also studies regular parts, but here, in view  of \eqref{def2},  we will be discussing  just principal parts. It turns out that the propositions are also valid for any $\alpha\in \RR_{\geq 0}$.}

\begin{prop}[proper fractions]
\label{prop1}
Let\; $\alpha\in (0,1)_\rn$ and $M\in \mathbb{N}$ be such that $\alpha \in \big(\frac{M-1}{M}, \frac{M}{M+1} \big]_\rn$. Then:
\begin{enumerate}
\item The principal part $\pr f^{(\alpha)}$ \eqref{exp_prin} has  coefficients:  
\be
\label{f_sequence1}
f^{(\alpha)}_\alpha\,,\;\;  f^{(\alpha)}_{2\alpha-1}\,,\;\;  f^{(\alpha)}_{3\alpha-2}\,,\;\;...\,, \;\;f^{(\alpha)}_{M\alpha-(M-1)}\,,\;\; f_0^{(\alpha)}\,;
\ee
equivalently, the classical conformal block \eqref{cl_bl_S} is the following Puiseux polynomial:   
\be
\label{lead_exp1}
\cS^{(\alpha)}(\delta, \tilde\delta, c|z) = \sum_{m=1}^{M} f_{m \alpha -m+1}^{(\alpha)}\, c^{m \alpha -m+1}+f_0^{(\alpha)}\,.
\ee
\item The  coefficients in $\pr f^{(\alpha)}$ at different $\alpha$ can be presented in tabular form 
\be
\label{table1}
    \begin{aligned}
    &M=1: \qquad f_{\alpha_1} ^{(\alpha_1)},\;\; &f_0^{(\alpha_1)}\,;\\
    &M=2: \qquad f_{\alpha_2} ^{(\alpha_2)},\;\;  f_{2\alpha_2 -1} ^{(\alpha_2)}, \;\;&f_0^{(\alpha_2)}\,;\\
    &M=3: \qquad f_{\alpha_3} ^{(\alpha_3)},\;\;  f_{2\alpha_3 -1} ^{(\alpha_3)},\;\; f_{3\alpha_3 -2} ^{(\alpha_3)}, &f_0^{(\alpha_3)}\,;\\
    &M=4: \qquad f_{\alpha_4} ^{(\alpha_4)},\;\;  f_{2\alpha_4 -1} ^{(\alpha_4)},\;\; f_{3\alpha_4 -2} ^{(\alpha_4)},\;\; f_{4\alpha_4 -3} ^{(\alpha_4)},\;\; &f_0^{(\alpha_4)}\,;\\
    &\hdots
    \end{aligned}
\ee
Here, each row lists the coefficient functions of $\pr f^{(\alpha_M)}$, where regular $\alpha_M \in \big(\frac{M-1}{M} , \frac{M}{M+1} \big)_\rn$ and $M=1,2,...\,$. Functions in a given column are equal to each other, i.e. 
\be
\label{columneq1}
\ba{l}
f_{\alpha_1} ^{(\alpha_1)} = f_{\alpha_2} ^{(\alpha_2)} = f_{\alpha_3} ^{(\alpha_3)} = \ldots\;,
\vspace{2mm}
\\
\dps
f_{2\alpha_1-1} ^{(\alpha_1)} = f_{2\alpha_2-1} ^{(\alpha_2)} = f_{2\alpha_3-1} ^{(\alpha_3)} = \ldots\;,
\vspace{2mm}
\\
\dps
\;\;\vdots
\vspace{2mm}
\\
\dps
f_{0}^{(\alpha_1)} = f_{0}^{(\alpha_2)} = f_{0} ^{(\alpha_3)} = \ldots\;.
\ea
\ee 
For all columns except the last one we can extend the interval $\big(\frac{M-1}{M} , \frac{M}{M+1} \big)_\rn$ to $\big(\frac{M-1}{M} , \frac{M}{M+1} \big]_\rn$. At endpoints which are exceptional points $\tilde \alpha_M = \frac{M}{M+1}$ the last column contains essentially different coefficient functions:
\be
\label{f0_exceptional1}
f_{0}^{(\tilde \alpha_1)} \neq  f_{0}^{(\tilde \alpha_2)} \neq  f_{0} ^{(\tilde\alpha_3)} \neq  ...\;. 
\ee
 
\end{enumerate}
\end{prop}

\begin{prop}[improper fractions]
\label{prop2}
Let $\alpha\in (1,\infty)_\rn$ and $N\in \mathbb{N}$ be such that $\alpha \in \big[\frac{N+1}{N} , \frac{N}{N-1} \big)_\rn$. Then:
\begin{enumerate}

\item The principal part $\pr f^{(\alpha)}$ \eqref{exp_prin}  has the coefficients  
\be
\label{f_sequence2}
f^{(\alpha)}_\alpha\,,\;\; f^{(\alpha)}_1\,,\;\; f^{(\alpha)}_{2-\alpha}\,,\;\;  f^{(\alpha)}_{3-2\alpha}\,,\;\;...\,\,,\;\; f^{(\alpha)}_{N-(N-1)\alpha}\,,\;\; f_0^{(\alpha)}\,;
\ee
equivalently, the classical conformal block \eqref{cl_bl_S} is the following Puiseux polynomial:   
\be
\label{lead_exp2}
\cS^{(\alpha)}(\delta, \tilde\delta, c|z) = \sum_{n=-1}^{N} f_{-n\alpha +n+1}^{(\alpha)}\, c^{-n\alpha +n+1}+f_0^{(\alpha)}\,.
\ee

\item The  coefficients in $\pr f^{(\alpha)}$ at different $\alpha$ can be presented in tabular  form
\be
\label{table2}
\begin{aligned}
    &\hdots\\
    &N=4: \qquad f_{\alpha_4}^{(\alpha_4)},\;\;  f_1 ^{(\alpha_4)},\;\; f_{2-\alpha_4} ^{(\alpha_4)},\;\; f_{3-2\alpha_4}^{(\alpha_4)},\;\; f_{4-3\alpha_4} ^{(\alpha_4)}, \;\; &f_0^{(\alpha_4)}\,;\\  
    &N=3: \qquad f_{\alpha_3}^{(\alpha_3)},\;\;  f_1 ^{(\alpha_3)},\;\;\;\; f_{2-\alpha_3} ^{(\alpha_3)},\;\; f_{3-2\alpha_3} ^{(\alpha_3)},\;\; &f_0^{(\alpha_3)}\,;\\
    &N=2: \qquad f_{\alpha_2}^{(\alpha_2)},\;\;  f_1 ^{(\alpha_2)},\;\; f_{2-\alpha_2} ^{(\alpha_2)},\;\; &f_0^{(\alpha_2)}\,;\\
    &N=1: \qquad f_{\alpha_1}^{(\alpha_1)},\;\;  f_1 ^{(\alpha_1)},\;\; &f_0^{(\alpha_1)}\,.
\end{aligned}
\ee
Here, each row lists the coefficient functions of $\pr f^{(\alpha_N)}$, where $\alpha_N \in \big(\frac{N+1}{N} , \frac{N}{N-1} \big)_\rn$ and $N=1,2,...\,$. Functions in a given column are equal to each other, i.e.
\be
\label{columneq2}
\ba{l}
\ldots  = f_{\alpha_3} ^{(\alpha_3)} = f_{\alpha_2} ^{(\alpha_2)} =  f_{\alpha_1} ^{(\alpha_1)}\;,
\vspace{2mm}
\\
\dps
\ldots  =  f_{1} ^{(\alpha_3)} = f_{1} ^{(\alpha_2)} = f_{1} ^{(\alpha_1)}\;,
\vspace{2mm}
\\
\dps
\hspace{35mm} \;\;\vdots
\vspace{2mm}
\\
\dps
\ldots  =  f_{0}^{(\alpha_3)} = f_{0}^{(\alpha_2)} = f_{0} ^{(\alpha_1)}\;.
\ea
\ee 
For all columns except the last one we can extend the interval $\big(\frac{N+1}{N} , \frac{N}{N-1} \big)_\rn$ to $\big[\frac{N+1}{N} , \frac{N}{N-1} \big)_\rn$. At endpoints which are exceptional points $\overset{\approx}{\alpha}_N = \frac{N+1}{N}$ the last column contains  essentially different coefficient functions:
\be
\label{f0_exceptional2}
\ldots  \neq  f_{0}^{(\overset{\approx}{\alpha}_3)} \neq f_{0}^{(\overset{\approx}{\alpha}_2)} \neq f_{0} ^{(\overset{\approx}{\alpha}_1)}\;.
\ee
\end{enumerate}
\end{prop}

The cases of $\alpha=0,1$ are known in the literature as global and Zamolodchikov classical conformal blocks (see our notation below \eqref{cl_bl_S}): 
\be
\ba{l}
\alpha=0:\qquad  \cS^{(0)}(\delta, \tilde\delta, c|z) = f_0^{(0)}\,,
\vspace{2mm} 
\\
\dps
\alpha=1:\qquad  \cS^{(1)}(\delta, \tilde\delta, c|z) = c f_1^{(1)} +f_0^{(1)}\,.
\ea
\ee
These are considered in Appendix \bref{app:alpha01}.

\begin{corollary}
\label{cor1}
There are only four different highest-order terms $f^{(\alpha)}_\alpha  c^\alpha$ in $\pr f^{(\alpha)}$  associated with four domains \eqref{domains}, cf. \eqref{four}. 
\end{corollary}

\begin{corollary}
\label{cor11} The number of terms in $\pr f^{(\alpha)}$ infinitely increases when $\alpha$ grows from $0$ to $1$, then at $\alpha=1$ the principal part sharply shrinks to just two terms, then, with  $\alpha$ growing  from $1$ to $\infty$, the number of terms in $\pr f^{(\alpha)}$  decreases from infinitely many  back to three terms. See fig. \bref{fig_graphic}. 
\end{corollary}

From above Propositions it follows that coefficients $f_0 ^{(\alpha)}$ are  the same for any {\it regular}  $\alpha$ (see Definition \bref{def:reg}): there are four different functions,  $f_0 ^{(\alpha <1)}$ and $f_0 ^{(\alpha > 1)}$, $f_0^{(1)}$ and $f_0^{(0)}$. On the other hand,  the zeroth-order  coefficients at {\it exceptional} points are different. 

\begin{definition}
\label{def:isol}
The two coefficients  $f_0^{(-)} \equiv f_0 ^{(\alpha < 1)}$ and $f_0^{(+)} \equiv f_0 ^{(\alpha >1)}$ existing for any regular $\alpha \in (0,1)_\rn \cup (1, \infty)_\rn$  are  called universal; the coefficients $f_0^{(\tilde \alpha_M)}$ and $f_{0}^{(\overset{\approx}{\alpha}_N)}$ are called isolated. 
\end{definition}

Both universal and isolated coefficients $f_0^{(\alpha)}$ can be collected as
\be
\ba{l}
\dps
\alpha \in (0,1)_\rn \;\,: \quad  \{ f_0 ^{(-)}\,,\; f_{0}^{(\tilde \alpha_M)}\;, \; M \in \NN  \}\,,
\vspace{2mm}
\\
\dps
\alpha \in (1, \infty)_\rn \,: \quad  \{ f_0 ^{(+)}\,,\; f_{0}^{(\overset{\approx}{\alpha}_N)}\;, \; N\in \NN\}\,.
\ea
\ee
In other words,  everywhere on the interval $\alpha \in (0,1)_\rn$ the zeroth-order term  $f_0^{(\alpha)}c^0$ in the Puiseux series remains the same except for  exceptional points, where its form is changed. The same holds true  on the interval
$\alpha \in (1,\infty)_\rn$.    

There are  linear relations allowing one to express particular non-zero order coefficients in the principal part in terms of  universal and isolated coefficients.

\begin{prop}
\label{universal01}
Consider a term   of minimal non-zero degree in $\pr f^{(\alpha)}$ (see the second-to-last terms in \eqref{f_sequence1} and  \eqref{f_sequence2}).  

\begin{enumerate}

\item Let $\alpha \in(0,1)_\rn$. For  regular $\alpha_M \in \left(\frac{M-1}{M} , \frac{M}{M+1} \right]_\rn$ and exceptional $\tilde \alpha_{M-1} = \frac{M-1}{M}$, $M=2,3,...\,$ one has  
\be
\label{restor}
f_{M \alpha_M-(M-1)} ^{(\alpha_M)} = f_0 ^{(\tilde \alpha_{M-1})} - f_0 ^{(-)}\;.
\ee

\item Let $\alpha \in (1, \infty)_\rn$. For  regular  $\alpha_N \in \left[\frac{N+1}{N} , \frac{N}{N-1} \right)_\rn$ and exceptional $\overset{\approx}{\alpha}_{N-1} = \frac{N}{N-1}$, $N=2,3,...\,$ one has  
\be
f_{N-(N-1)\alpha_N} ^{(\alpha_N)} = f_0^{(\overset{\approx}{\alpha}_{N-1})} - f_0 ^{(+)}\;.
\ee
\end{enumerate}
\end{prop}     

\noindent This proposition makes it possible to reconstruct the principal part from partial data. 

\begin{corollary}
\label{cor111}
The principal part $\pr f^{(\alpha)}$ for any $\alpha \in (0,1)_\rn \cup (1, \infty)_\rn$ can be found   by  having only $\pr f^{(\alpha)}$ on the $M,N=1$ open intervals  \eqref{table1} and \eqref{table2} as well as isolated coefficients $f_0 ^{(\tilde \alpha_M)}$ and $f_0^{(\overset{\approx}{\alpha}_{N})}$ for $\forall N,M \in \mathbb{N}$.
\end{corollary}

Let us see how this algorithm works for $\alpha \in (0,1)_\rn$ (the case $\alpha \in (1, \infty)_\rn$ can be considered analogously).   Suppose that from  Proposition  \bref{prop1} we know  the principal part $\pr f^{(\alpha_1)}$ on the first ($M=1$) open interval  $\alpha_1 \in(0,\half)_\rn$. It  consists of two terms which coefficients are  the highest-order  $f_{\alpha_1} ^{(\alpha_1)}$  and the zeroth-order universal $f_0 ^{(-)}$, see the first row in  \eqref{table1}:
\be
\pr f^{(\alpha_1)} =   f_{\alpha_1} ^{(\alpha_1)}c^{\alpha_1} + f_0 ^{(-)}\;.
\ee
These two functions are the same for any succeeding open interval ($M>1$). To have $\pr f^{(\alpha_1)}$ on the half-open interval $\alpha_1 \in (0,\frac12]_\rn$ one replaces the universal $f_0 ^{(-)}$ by the isolated $f_0 ^{\tilde \alpha_1}$, where $\tilde \alpha_1 = \frac12$ is the first exceptional point.   Consider now the second ($M=2$) open interval $\alpha_2 \in (\frac12,\frac23)_\rn$, where  the principal part gains one more coefficient, 
\be
\label{sec_int_f}
\pr f^{(\alpha_2)} =   f_{\alpha_2} ^{(\alpha_2)}\,c^{\alpha_2} + f_{2\alpha_2-1} ^{(\alpha_2)}\, c^{2\alpha_2-1} + f_0 ^{(-)}\;,
\ee
where the leading coefficient is equal to  that one from the first interval, $f_{\alpha_2} ^{(\alpha_2)} = f_{\alpha_1} ^{(\alpha_1)}$, see \eqref{columneq1}. Using Proposition \bref{universal01}  the subleading  coefficient $f_{2\alpha_2-1} ^{(\alpha_2)}$ can be found from \eqref{restor}  as 
\be
\label{restor2}
f_{2\alpha_2-1} ^{(\alpha_2)} = f_0 ^{(\tilde \alpha_{1})} - f_0 ^{(-)}\;.
\ee
Thus, knowing the isolated $f_0 ^{(\tilde \alpha_{1})}$ one can find the whole principal part on the second open interval \eqref{sec_int_f}. Going to the half-open interval $\alpha_2 \in (\frac12, \frac23]_\rn$ one again replaces the universal $f_0 ^{(-)}$ by the isolated $f_0 ^{(\tilde \alpha_2)}$, where $\tilde \alpha_2 = \frac23$ is the second exceptional point.  By repeating this procedure one obtains  $\pr f^{(\alpha_M)}$ on any interval labelled by $M = 2,3,...\,$.

\paragraph{Short summary.} Let us  summarize the basic properties of the leading exponent in the large-$c$ exponentiated Virasoro conformal block  \eqref{def2} for $\alpha$-heavy operators \eqref{h_th_a}. 
\begin{itemize}

\item The leading exponent is given by a Puiseux polynomial which is the principal part of the  logarithmic conformal block $\cS^{(\alpha)}(\delta, \tilde\delta, c|z)  = \pr f^{(\alpha)}(\delta, \tilde\delta, c|z)$. The principal part is unconventionally defined to contain the zeroth-order term which unexpectedly plays a crucial role in defining other terms in $\pr f^{(\alpha)}(\delta, \tilde\delta, c|z)$. (See Propositions \bref{prop1}, \bref{prop2}, and \bref{universal01},  Corollary \bref{cor111}.)

\item The half-line $\alpha \in [0, \infty)_\rn$ is naturally divided into four domains: $\alpha = 0 $, $\alpha\in (0,1)_\rn$, $\alpha  = 1$, $\alpha\in (1, \infty)_\rn$, where the leading exponent $\cS^{(\alpha)}(\delta, \tilde\delta, c|z)$ has crucially  different properties. The leading terms of $\cS^{(\alpha)}(\delta, \tilde\delta, c|z)$ in each of four domains are unique. (See Corollary \bref{cor1}.)   

\item Each open interval from the previous item can be divided into infinitely many subintervals within each of which the Puiseux polynomial $\cS^{(\alpha)}(\delta, \tilde\delta, c|z)$ has a fixed number of terms. The number of terms grows with a serial number of a given interval when moving $\alpha$ from $0$ to $1$ and decreases when  moving $\alpha$ from $1$ to $\infty$. (See \eqref{intervals} and fig. \bref{fig:plot} as well as Corollary \bref{cor11}.)

\item Every coefficient function which enters the principal part for some $\alpha_0$, except $f^{(\alpha_0)}_0$, is the same for all  $\alpha \in [\alpha_0,1)_\rn$, if $\alpha_0 \in(0,1)_\rn$, and for all $\alpha\in(1,\alpha_0]_\rn$, if $\alpha_0\in (1, \infty)_\rn$. 

\item There is a reconstruction method which allows building  $\cS^{(\alpha)}(\delta, \tilde\delta, c|z)$ for any value of $\alpha\in (0,1)_\rn \cup (1, \infty)_\rn$ from some initial data which are the leading exponent for the first intervals $\alpha\in (0, \frac12)_\rn$ and $\alpha \in (2, \infty)_\rn$ as well  as the zeroth-order terms for any $\alpha \in (0,1)_\rn\cup (1,\infty)_\rn$. (See Proposition \bref{universal01} and Corollary \bref{cor111}.) 

\item The above constructions for $\alpha \in \QQ$ are also valid for $\alpha\in \RR_0^+$. In this case the Puiseux polynomial contains irrational powers and, hence, $c=\infty$ becomes a logarithmic branch point. (See Appendix \bref{app:vir}.)

\end{itemize}

\subsection{Transmutation of singularities}
\label{sec:sing}

It is observed that for particular conformal dimensions the order of a branch point/pole of the logarithmic conformal block can be changed. We call this phenomenon {\it a transmutation of singular points} which manifests itself already for the second block coefficient. Recall that the $\alpha$-heavy  logarithmic conformal block is $O(c^\alpha)$.  Next we  consider two essentially different cases when the leading order of $c$ either increases or decreases. In the later case, the logarithmic block coefficients are $O(c^0)$ and the point $c=\infty$ becomes {\it a removable singularity}. In the former case, the logarithmic conformal block contains terms of  arbitrarily high order in powers of  $c$ and, therefore, becomes $O(c^\infty)$ which means that the singularity is {\it essential}.

In this section, this phenomenon is shown in different cases, without trying to develop a general theory. For further purposes, having in mind our definition \eqref{h_th_a}, we introduce a class of   conformal dimensions:
\begin{definition} Classical Kac dimensions,  
\be
\label{kac_cl}
\tilde h_{n} = \tilde \delta_{n} c =\frac{1 - n^2}{24} c\,,
\qquad 
\text{where}
\quad
\frac{\partial \tilde \delta_{n}}{\partial c}  = 0\,, \quad n \in \NN\,.
\ee  
\end{definition}
\noindent This name is due to the Kac dimensions belonging to the same $O(c)$ class, i.e. 
\be
\label{kac}
h_{r,s} = \frac{1 - r^2}{24} c +\frac{13 r^2-12 r s-1}{24}c^0 + O(c^{-1})\,, \qquad r,s = 1,2,...\,.
\ee

\subsubsection{Transmutation for heavy operators} 
\label{sec:trans_heavy}

Here, we consider the  $\alpha=1$ case. The logarithmic block   \eqref{f_exp_alpha}  can be represented as 
\be
\label{f_2ex2_1}
\ba{c}
\dps
f^{(1)}(\delta, \tilde \delta|z) = g_1^{(1)} z+ g_2^{(1)} z^2 + O(z^3) \equiv A_1 c\, z +\frac{A_{11} c^2 + A_{12} c}{A_{21} c +A_{22}}\,z^2 + O(z^3) 
\vspace{2mm}
\\
\dps
 =\frac{\tilde \delta}{2} c \, z + \frac{\tilde\delta(26 {\tilde\delta}^2 + 16 {\tilde \delta} \delta + 32\delta^2+ 3 {\tilde \delta})c^{2} + (16({\delta^2 - \tilde \delta \delta - {\tilde \delta}^2)  + 2\tilde\delta)c}}{16\tilde\delta(8\tilde \delta+1)c+(8- 80 {\tilde \delta})} z^2
+ O(z^3)\,.
\ea
\ee
Coefficients $A_1,A_{11}, ... $  depend on two classical dimensions  and the respective Puiseux series turns to the Laurent series: 
\be
\label{c1}
\ba{c}
\dps
f^{(1)}(\delta, \tilde \delta|z) = A_1 c\, z + \left[\frac{A_{11}}{A_{21}}c + \frac{A_{12} A_{21} -A_{11} A_{22}}{A_{21}^2}+O(c^{-1})\right]z^2 + O(z^3)
\vspace{2mm}
\\
\dps
= \frac{\tilde \delta}{2} c \, z 
+\left[\frac{32 \delta^2+16 \tilde\delta \delta+26 \tilde\delta^2+3 \tilde\delta}{16(8\tilde\delta+1)}c+\frac{(2\tilde\delta-24\delta+1)^2}{32 (8\tilde\delta+1)^2}+O(c^{-1})\right]z^2+ O(z^3)\,,
\ea
\ee
which means that the point  $c=\infty$ is an isolated singularity of the first order (a simple pole). For arbitrary conformal  dimensions the leading asymptotics here is always $c^1$. 

Examining  the representation \eqref{f_2ex2_1} one finds out that choosing the intermediate conformal dimension as the classical Kac dimension $\tilde\delta =  -1/8 \equiv \tilde \delta_{2}$ the leading asymptotics becomes  $ c^2$. Indeed, this is equivalent to imposing   $A_{21} = 0$ while other coefficients $A_1,A_{11},A_{12},A_{22}\neq 0$ that makes the denominator in \eqref{f_2ex2_1} to be of lower-order in $c$.  It immediately follows that the expansion \eqref{c1} is not applicable anymore because of the $A_{21}^{-1}$ pole.

On the other hand, choosing in \eqref{f_2ex2_1} another classical Kac dimension $\tdelta   = 0 \equiv  \tilde \delta_{1}$ the leading asymptotics remains  $c^1$ because in this case $A_1,A_{11}, A_{21}=0$, while  $A_{12},A_{22} \neq 0$. The expansion \eqref{c1} is still  not applicable: the condition $A_{21} =0$ generates a pole  yet the highest degree of $c$ is not changed.  We emphasize  that one substitutes  $\tdelta =0$ into the logarithmic block  $f^{(1)}(\delta, \tilde \delta|z)$ {\it before} expanding near $c=\infty$. Exchanging  these two steps yields  inconsistency. The expanded block $f^{(1)}(\delta, \tdelta|z)$ \eqref{c1} at $\tdelta=0$ takes the form 
\be
\label{noncom}
f^{(1)}(\delta, 0|z) = \left[2\delta^2c+\frac{(24\delta-1)^2}{32}+O(c^{-1})\right]z^2+ O(z^3)\,.
\ee 
However, this vacuum  block is false since the further substitution of  $\delta=0$ produces  a non-vanishing function $f^{(1)}(0, 0|z) = -1/32 \,z^2 + ...\,$ which by definition must be equal to $0$.\footnote{\label{F=1}The bare Virasoro conformal block \eqref{block_def} for zero  conformal dimensions equals $1$ for any $c$: $\cF(0, 0, c\,| z) = 1$.}   
 
To summarize, for particular intermediate classical  dimensions the second coefficient  changes from being rational to quadratic,  
\be
\label{secoef1}
\frac{A_{11} c^2 + A_{12} c}{A_{21} c +A_{22}} \;\;\to\;\; \frac{A_{11}}{A_{22}} c^2 + \frac{A_{12}}{A_{22}} c\,,
\ee
that leads to poles of different orders in the respective Laurent expansions. This is {\it a   transmutation of a singular point}. In particular, this  prevents getting  the vacuum  classical block by substituting $\tdelta=0$ into the classical block which was calculated for  $\tdelta \neq 0$ (see Section \bref{sec:identity} for further discussion). 

So far we have studied   the transmutation phenomenon for the first and second block coefficients $g_1^{(1)}$ and $g_2^{(1)}$. Considering other available higher-order block coefficients $g_n^{(1)}$, $n=1,...,5$ in \eqref{f_2ex2_1}  one finds poles in the classical Kac dimension $\tilde \delta_{4} = -5/8$. Our numerical studies demonstrate that other classical Kac dimensions on these levels  do not show up. One can explicitly see that for intermediate classical Kac dimensions $\tilde \delta_{2}$ and $\tilde \delta_{4}$ the large-$c$ logarithmic  conformal blocks  up to $O(z^6)$ are given by 
\be
\label{kac_blocks}
\ba{l}
\dps
f^{(1)}(\delta, \tilde \delta_{2}|z) = a_4(\delta|z) \,c^4 +  a_3 (\delta|z)\, c^3 +  a_2(\delta|z) \,c^2 +  a_1(\delta|z)\, c^1 + a_0(\delta|z) + O(c^{-1})\,,   
\vspace{2mm}
\\
\dps
f^{(1)}(\delta, \tilde \delta_{4}|z) =  b_2(\delta|z) \,c^2+  b_1(\delta|z) \, c^1 + b_0(\delta|z) \,  + O(c^{-1}) \,,  
\ea
\ee
where $a_i, b_j$ are some coefficients.  We observe here  that the leading terms  grow faster than $c^1$. Moreover, the more higher-order terms in $z$ we take into account, the  higher degree of the leading term in $c$ we obtain:\footnote{\label{merc}First, note that by means of the Mercator series formula the logarithmic conformal block inherits all poles of  the Virasoro conformal block. Second, the Kac determinant on the $\cN$-th level keeps the zeros from the previous levels. Then, substituting classical Kac dimensions into the block coefficients leads to that almost all terms in the denominators become zero which in turn leads to the pole's order growing linearly with $\cN$. On the other hand, higher level block coefficients contain $\tdelta_n$ with larger $n$. }  for the classical Kac dimension $\tdelta_n$ one finds the large-$c$ logarithmic conformal  block $O(c^{2(\cN-1)/n})$ in the $O(z^{\cN+1})$ approximation \eqref{block_def}.  In this respect, the large-$c$ expansion of the  logarithmic conformal block for intermediate classical Kac  dimensions is   similar to the large-$c$ expansion of  the original Virasoro conformal block \eqref{asymp} for arbitrary dimensions: at $\cN\to \infty$ the point  $c=\infty$ becomes {\it an essential singularity} of this particular  logarithmic conformal block, $f^{(1)} = O(c^\infty)$. Obviously, there is no exponentiation with a finite order pole in this case.            

Let us discuss now how the original Kac dimensions which are responsible for poles of the Virasoro conformal block show up within the large-$c$ asymptotic expansions. We stress that a classical Kac dimension \eqref{kac_cl}  is different from the original  Kac dimension \eqref{kac}. Instead, they both belong to the $O(c)$ class of conformal dimensions having  a leading term $\frac{1 - r^2}{24} c$. The classical Kac dimension is not a pole of the Virasoro conformal block which is defined by zero denominator $A_{21} c+A_{22} =0$ in \eqref{f_2ex2_1} (see footnote \bref{merc}).  Indeed, solving this equation one finds the Kac dimension $h_{2,1} = h_{2,1}(c)$ \eqref{kac} which can be represented as    
\be
\tilde h_{2,1} = \sum_{n=-1}^\cM  (\tilde h_{2,1})_n \,c^{-n} + O(1/c^{\cM+1})
\equiv  
\tdelta_2 c + \sum_{n=0}^\cM  (\tilde h_{2,1})_n \,c^{-n} + O(1/c^{\cM+1})\,,    
\ee      
see \eqref{kac_cl}. Now, discarding here the $O(1/c^{\cM+1})$ terms and using  such a truncated conformal dimension when calculating  the large-$c$ logarithmic conformal block  one finds that the second coefficient $g_2^{(1)}$ behaves as $ c^{\cM+3}$ (earlier we saw that using just the first term ($\cM=-1$) yields the leading asymptotic $ c^2$). Sending $\cM$ to infinity (thereby, reconstructing the original Kac dimension) we are gradually approaching the pole of the second coefficient which reveals  itself as 
\be
g^{(1)}_{2} =  \lim_{\cM\to \infty} O(c^{\cM+3}) = O(c^{\infty})\,.
\ee 
A number of terms in $\pr g_2^{(1)}$  increases  infinitely so that  $c=\infty $ becomes an essential singularity of the logarithmic block $f^{(1)}$ already in the second order in $z$ (contrary to the classical Kac dimensions where the essential singularity is produced when $\cN \to \infty$).

Finally, let us briefly consider the case of lowering the order of the pole. By imposing $A_{11} =0$ the second coefficient changes as (cf. \eqref{secoef1})  
\be
\frac{A_{11} c^2 + A_{12} c}{A_{21} c +A_{22}}
\;\; \to \;\; 
\frac{A_{12} c}{A_{21} c +A_{22}}\,,
\ee
which means that $g_2^{(1)} = O(c^0)$. The equation $A_{11}\equiv \tilde \delta\left(32  \delta^2+ 16 \tilde \delta \delta+26 \tilde \delta^2+3\tilde \delta\right) = 0$ has roots
\be
\tilde \delta = 0 \quad \text{and/or} \quad \delta = \frac{1}{8} \Big(-2 \tilde \delta \pm\sqrt{-6\tilde \delta (8 \tilde \delta+1)}\Big)\,,
\ee
the reality condition restricts the intermediate dimension to be $\tilde \delta \in [-1/8,0] \equiv [\tilde \delta_2, \tilde \delta_1]$, cf. \eqref{kac_cl}.\footnote{One can check higher-order  block  coefficients $g_n^{(1)}$, where $n=1,..., \cN$,  and find higher-order  polynomial equations on $\delta, \tilde \delta$ which lower the order of the pole. Their solutions contain pairs $(\delta, \tilde \delta)$, where $\tilde \delta$ is one of the classical Kac dimensions $\tilde \delta_{r,s}$ with $rs \leqslant \cN$.} We should emphasize here that contrary to the previous case, the lowering of the order happens only for some of higher-order coefficients and the large-$c$ asymptotics of the whole  logarithmic block remains the same as for general conformal dimensions, $f^{(1)} = O(c^1)$. Indeed, considering e.g. just first two terms in $f^{(1)} = g_1^{(1)} z + g_2^{(1)} z^2 + ...$ and imposing $A_{11} = 0$ does not generally lead to $g_1^{(1)} \equiv A_1 c =0$ that keeps the leading asymptotic equal to $c^1$ in this order.

\subsubsection{Transmutation for $\alpha$-heavy operators} 

The situation here is less informative since there are no Kac dimensions $O(c^\alpha)$ at $\alpha \neq 1$ which are poles of the original Virasoro block. For the $\alpha$-heavy   logarithmic conformal block this implies that choosing higher-order coefficients equal to zero in both the numerator  and denominator  does not give non-trivial solutions. Let us consider two examples: 
\be
\ba{l}
\label{f_2ex21/2}
\dps
\alpha = \frac12:\qquad f^{(\half)} = A_1 c^\half\, z +\frac{A_{11} c^{\frac32}+A_{12} c+A_{13} c^{\frac12}}{A_{21} c +A_{22} c^{\frac12}+A_{23}}\,z^2 + O(z^3)\,;
\vspace{3mm}
\\
\dps
\alpha = \frac32:\qquad f^{(\frac32)} = B_1 c^{\frac32}\, z +\frac{B_{11} c^{3}+B_{12} c^{\frac72}+B_{13} c^{2}+B_{14} c^{\frac32}}{B_{21} c^2 +B_{22} c^{\frac32}+B_{23} c^{\frac12}+B_{24}}\,z^2 + O(z^3)\,.
\ea
\ee
Explicit coefficients in each case  can be read off from \eqref{f_exp_alpha}. 

\begin{itemize}
\item $\alpha=\half$. The singularity here is the branch point of order $2$, i.e. $f^{(\half)} = O(c^{\half})$. It can transmute provided that at least $A_{21}=0$, $A_{11}\neq 0$, that would lead to $g_2^{(\half)} = O(c^1)$, or $A_1,A_{11}=0$, that would lead to $f^{(\half)}=O(c^0)$. However, one can see that the only solution is always given by $\tilde \delta = 0$, in which case the logarithmic vacuum conformal block is given by  
\be
f^{(\half)} = 2 \delta^2\, z^2 + O(z^3)\,,
\ee 
and, hence, $c=\infty$ is {\it a removable singularity} in this order.    

\item $\alpha=3/2$. Here, $g_2^{(3/2)}$ has more terms in the numerator and denominator and, hence, more  possibilities to change the order of singularity. However, in these cases one observes the same effect that the only value of the  intermediate dimension which triggers  the transmutation is given by  $\tilde \delta = 0$. Namely, in this case $B_1,B_{11},B_{12},B_{14},B_{21},B_{22},B_{23}=0$, while $B_{13} = 16 \delta^2$ and $B_{24} = 8$ that gives 
\be
f^{(3/2)} =  2 \delta^2 c^2\, z^2 +O(z^3)\,.
\ee 
The point $c=\infty$ is the second order pole here. One can show that taking account of  higher-order coefficients $g_n^{(3/2)}$ the order of $c$ increases with $n$  thereby making $c=\infty$ an essential singularity of the logarithmic block $f^{(3/2)}$.         
\end{itemize}

\noindent In the next section we analyze the   logarithmic vacuum conformal blocks in more detail and show that their large-$c$ behaviour crucially depends on either $\alpha$ is more or less than 1. Namely, $c=\infty$ for $\alpha \in  (0,1)_\rn$ is still  a branch point of $f^{(\alpha)}$ but of lower order, while for $\alpha \in (1, \infty)_\rn$ it is an essential singularity. The later case  is similar to that one with classical Kac dimensions examined  in Section \bref{sec:trans_heavy}. In both cases, one has a transmutation of singularities.

\subsection{Classical vacuum conformal blocks}
\label{sec:identity}

Substituting first $\tilde h = 0$ and $h = \delta c^{\alpha}$ in the Virasoro conformal block $(\ref{General block})$ one then computes the  logarithmic vacuum  conformal block
\be
\label{log_vac}
\ba{l}
\dps
\dot{f}^{(\alpha)}(\delta, c, z) = 2 \delta^2 c^{2\alpha -1} z^2 + 2 \delta^2 c^{2\alpha -1} z^3 + \frac{c^{2\alpha -2}\delta^2 \big( c(9c+40) + 4 \delta c^{\alpha +1} - 44 \delta^2 c^{2\alpha} \big)}{22+5c} z^4 + O(z^5)\,.
\ea
\ee
We use a dot to mark the vacuum  classical block since due to the singular transmutation this function cannot be obtained from non-vacuum  classical block by choosing $\tdelta=0$, see the discussion around eq. \eqref{noncom} in the previous section. Using our {\it Mathematica} code we can calculate the vacuum Virasoro conformal  blocks  up to $O(z^{22})$, while the  classical vacuum conformal block can be calculated up to $O(z^{16})$ (see the footnote \bref{foo1}). The behaviour of vacuum  blocks is essentially different in the four domains of  $\alpha$ \eqref{domains}, in particular, one directly see that for $\alpha \in \mathbb{Q}^-_0$ the  principal part $\pr f^{(\alpha)} = 0$. Other cases are analyzed below. 

\paragraph{${\bm\alpha\in(1, \infty)_\rn}$.} In this case, the singular transmutation converts   $c=\infty$ into an essential singularity.  Indeed, examining higher-order coefficients $g_n^{(\alpha)}$  in \eqref{log_vac} (up to $O(z^{16})$)  one finds a general pattern of increasing the order of $c$: 
\be
g_{n}^{(\alpha)} = O(c^{m \alpha - (m-1)})\,,
\qquad
\text{where} 
\quad 
m=n+ \frac{1 + (-1)^{n+1}}{2}\,, \quad n =2,3,4, ...\,,
\ee
which says that $\lim_{n\to\infty} g_{n}^{(\alpha)} = O(c^\infty)$.

\paragraph{${\bm\alpha \in(0,1)_\rn}$.} The same numerical analysis shows that the transmutation changes the leading asymptotics from $c^{\alpha}$ to $c^{2\alpha-1}$. The Proposition \bref{prop1} is still valid for the  classical vacuum conformal blocks. However, one finds out that both the highest-order and lowest-order terms in the principal part equal zero. Namely,  these are the leading  and  universal coefficients (see Definition \bref{def:isol}), 
\be
\label{fdot1}
\dot{f}_{\alpha}^{(\alpha)} = 0\,,
\qquad 
\dot{f}_0^{(-)} = 0\;.
\ee
In general, such constraints are not surprising because setting  $\tdelta =0$ makes many coefficients in both numerators and denominators of rational functions $g_{n}^{(\alpha)}$  equal to zero. In particular, the first relation is just the fact that the leading asymptotics becomes $O(c^{2\alpha-1})$, see \eqref{f_sequence1}. This allows one to  formulate the analog of Proposition \bref{universal01} for the  classical vacuum conformal blocks valid up to $O(z^{16})$:
\begin{prop} 
\label{prop_vacuum}
Let $\alpha\in (0,1)_\rn$. Consider a term of minimal non-zero degree in $\pr \dot{f}^{(\alpha)}$.  For regular $\alpha_M \in \left(\frac{M-1}{M} , \frac{M}{M+1} \right]_\rn$ and exceptional $\tilde \alpha_M = \frac{M-1}{M}$, $M=2,3,...\,$ one has 
\be
\label{rel_vac}
\dot{f}_{M \alpha_M-(M-1)}^{(\alpha_M)} = \dot{f}_0^{(\tilde \alpha_M)}\,.
\ee
\end{prop} 

Quite similar to the non-vacuum  classical conformal blocks  discussed in Corollary \bref{cor111} the relation \eqref{rel_vac} can be used to find the whole principal part $\pr \dot{f}^{(\alpha)}$ for any $\alpha \in (0,1)_\rn$ by knowing only isolated coefficients $\dot{f}_0 ^{(\tilde \alpha_M)}$ for $\forall M \in \mathbb{N}$. Here, we can repeat the reconstruction algorithm   from Section \bref{sec:structure} which for  $\pr \dot f^{(\alpha)}$  is even simpler.

Let us consider the first open interval $\alpha_1 \in (0, \frac12)$. On this interval the classical block contains two terms and from \eqref{fdot1}  one finds that they both equal zero,\footnote{\label{glbl} The respective large-$c$ Virasoro block is, therefore, equal to 1. This  is somewhat similar to the (bare) global conformal block  $G(h_i, \tilde h|z) = {}_2 F_1(\tilde h - h_1+h_2, \tilde h +h_3-h_4, 2\tilde h|z)$ \cite{Ferrara:1974ny}, which at pairwise equal external dimensions ($h_1=h_2$ and $h_3 = h_4$) and  zero intermediate dimension ($\tilde h = 0$) is also equal to 1.} 
\be
\pr \dot f^{(\alpha_1)} =   0\;.
\ee
Thus, the first non-zero classical block can appear at $\alpha \in [\frac12, \infty)_\rn$. Indeed, at the first exceptional point $\tilde \alpha_1  = \half$ one knows that \cite{Fitzpatrick:2014vua}:\footnote{In fact, the analytical  calculation of  \cite{Fitzpatrick:2014vua} demonstrates that for $\alpha = \half$ the transmutation of singularity  does hold in any order in $z$. In this case, the singularity $c=\infty$ is removable. }
\be
\label{vac1/2}
\pr \dot f^{(\tilde \alpha_1)} =  \dot f_0^{(\frac12)} = 2\delta^2 z^2~ _2F_1(2,2,4;z)\,.
\ee
One can directly check this result by expanding the hypergeometric series and comparing with \eqref{log_vac} up to $O(z^{16})$, see \eqref{iso1}. Note that here the classical block $\dot \cS^{(\tilde \alpha_1)}(\delta, c|z) = \pr \dot f^{(\tilde \alpha_1)} = O(c^0)$ contains just one term which is the isolated coefficient.

From  Proposition \bref{prop_vacuum} one finds on the second open interval $\alpha_2 \in (\frac12, \frac23)_\rn$ that 
\be
\pr \dot f^{(\alpha_2)} =  \dot f_{2\alpha_2-1}^{(\alpha_2)}\, c^{2\alpha_2-1} =  \dot f_0^{(\frac12)}\,c^{2\alpha_2-1}  = 2\delta^2 z^2~ _2F_1(2,2,4;z)\,c^{2\alpha_2-1}\,.
\ee
Recall that the universal coefficient is zero, see \eqref{fdot1}. Thus, one knows an analytic expression for the classical vacuum  block $\dot \cS^{(\alpha)}(\delta, c|z)$ at  $\alpha \in (0, 2/3)$. Developing this procedure further one can build  $\dot \cS^{(\alpha)}(\delta, c|z)$ for any $\alpha\in (0,1)_\rn$ provided one knows  isolated coefficients $f_0^{\tilde \alpha_n}$ at $n=2,3,...\,$. In this way, one can show (up to $O(z^{16})$) the main proposition of this section:

\begin{prop}[vacuum block reconstruction]
\label{prop_iden}
Classical vacuum conformal blocks for any $\alpha\in (0,1)_\rn$  are completely defined by isolated coefficients  $f_0^{(\tilde \alpha_M)}$, $M=1,2,...\,$. Namely, for $\alpha \in \big[\frac{M-1}{M}, \frac{M}{M+1} \big)_\rn$ one has  
\be
\label{isol_rec}
\dot\cS^{(\alpha)}(\delta, c|z) = \sum_{m=1}^{M-1} \dot f_{0}^{(\tilde\alpha_{m})}\, c^{(m+1) \alpha -m}\,.
\ee
\end{prop}
The first isolated coefficient  $f_0^{(\tilde \alpha_1)}$ here is explicitly known as the hypergeometric function \eqref{vac1/2}. One may wonder if  other isolated coefficients at higher  exceptional points $\tilde\alpha_{m} = \frac{m}{m+1} \in (0,1)_\rn$  have similar analytic expressions. More precisely, one may try to conjecture them by analyzing the numerical form of these  coefficients  which can be  calculated up to $O(z^{16})$, see e.g. \eqref{iso1}-\eqref{iso7}.  To this end, one introduces the function
\be
F_a(z) = {}_2F_1(a,a;2a|z)\,,
\ee  
which is in fact the $sl(2)$ vacuum global block (see the footnote \bref{glbl}). Then, the first two isolated coefficients \eqref{iso1} and \eqref{iso2}  can be represented as 
\be
\label{iso1223}
\dot f_0^{(\frac12)} =  2\delta^2 z^2 F_2(z)
\quad
\text{and}
\quad
\dot f_0^{(\frac23)} = \frac{4\delta^3 z^4}{5} \Big(6F_1(z)F_3(z)-5F_2(z)F_2(z)\Big)\,.
\ee
The available numerical expressions for the isolated coefficients at higher exceptional points do not contain enough terms to represent them in similar form. However, having in mind \eqref{iso1223} one may conjecture that  all of them are described by the following ansatz 
\be
\label{ansatz}
\dot f_0^{(\tilde\alpha_{n})} = \delta^{n+1} z^{m}\sum_{\substack{a_k \in \mathbb{Z} \\ a_1+...+a_k = m}} s_{a_1...a_k} F_{a_1}(z)F_{a_2}(z)\cdots F_{a_k}(z) \,,
\ee
where  coefficients $s_{a_1...a_k}$ are to be found and $m=n+1+(1+(-1)^n)/2$ for $n\in \NN$. It is interesting that the similar relations are known in the literature for  large-$c$ HHLL (heavy-heavy-light-light) vacuum  conformal block, see \cite{Kulaxizi:2018dxo,Fitzpatrick:2015qma,Karlsson:2021mgg}. In particular, relations \eqref{iso1223} arise in the lowest orders of the heavy-light perturbation theory. In our case the conformal block is of the HHHH type since  all operators are $\alpha$-heavy and there are differences arising in higher orders.   

Nonetheless, our numerical experiments show that apparently there are no any finite combinations of functions $F_{a}$ which can represent $f_0^{(\tilde\alpha_{n})}$ for any $n$ except for $n=1,2$. Infinite combinations \eqref{ansatz}  are still possible  because they are just the change of  basis of a function space. Presently, the only thing we can observe for sure is that 
\be
\dot f_0^{(\tilde\alpha_{n})} \sim \delta^{n+1} z^{m}\left[1 +\frac{m z}{2} + O(z^2)\right],
\qquad
m=n+1 + \frac{1+(-1)^n}{2}\,, \quad n\in \NN\,.
\ee
 
\paragraph{${\bm\alpha=1}$.} The classical vacuum block $\pr \dot f^{(1)}  = c\dot f_{1}^{(1)}+\dot f_{0}^{(1)}$  can be calculated up to $O(z^{16})$. The first term which is the Zamolodchikov vacuum conformal block reads
\be
\label{vacuum-classical}
\ba{l}
\dps \dot f_{1}^{(1)} = 2\delta^2 z^2 \bigg[1+z+z^2 \left(\frac{9}{10}+\frac{2 \delta }{5}-\frac{22 \delta^2}{5}\right)+z^3 \left(\frac{4}{5}+\frac{4 \delta }{5}-\frac{44 \delta ^2}{5}\right)+ 
\vspace{2mm} 
\\
\dps
  + z^4 \left(\frac{5}{7}+\frac{39 \delta }{35}-\frac{2503 \delta ^2}{210}-\frac{296 \delta ^3}{35}+\frac{6016 \delta^4}{105}\right)+z^5 \left(\frac{9}{14}+\frac{47 \delta }{35}-\frac{963 \delta ^2}{70}-\frac{888 \delta ^3}{35}+\frac{6016 \delta^4}{35}\right)+ 
\vspace{2mm} 
\\
\dps
+ z^6 \left(\frac{7}{12}+\frac{263 \delta }{175}-\frac{38359 \delta
   ^2}{2625}-\frac{41666 \delta ^3}{875}+\frac{814004 \delta ^4}{2625}+\frac{187616 \delta ^5}{875}-\frac{931408 \delta ^6}{875}\right)+
\vspace{2mm} 
\\
\dps
+z^7 \left(\frac{8}{15}+\frac{282 \delta }{175}-\frac{12937 \delta ^2}{875}-\frac{63064 \delta ^3}{875}+\frac{383472 \delta
   ^4}{875}+\frac{750464 \delta ^5}{875}-\frac{3725632 \delta ^6}{875}\right)
\ea
\ee
$$
\ba{l}
\dps
+ z^8\left(\frac{27}{55}+\frac{1943 \delta }{1155}-\frac{83773 \delta ^2}{5775}-\frac{2785771 \delta ^3}{28875}+\frac{15432733 \delta
   ^4}{28875}+\frac{58445216 \delta ^5}{28875} -\frac{276888008 \delta ^6}{28875} -\right. 
\vspace{2mm} 
\\
\dps
\left. -\frac{58081152 \delta ^7}{9625}+\frac{45651968 \delta^8}{1925}\right)+z^9 \left(\frac{5}{11}+\frac{1993 \delta }{1155}-\frac{80516 \delta ^2}{5775}-\frac{689743 \delta ^3}{5775}+\frac{3399349
   \delta ^4}{5775}+ \right.
\vspace{2mm} 
\\
\dps
\left.+\frac{3042464 \delta ^5}{825}-\frac{92469224 \delta ^6}{5775}-\frac{58081152 \delta ^7}{1925}+\frac{45651968 \delta^8}{385}\right) + O(z^{10}) \bigg]\,.
\ea
$$

\begin{prop}  The Zamolodchikov vacuum conformal block  can be represented as  
\label{prop:Zam} 
\be
\label{f10}
\dot f_{1}^{(1)} = \dot f_0^{(\frac12)} + \dot f_0^{(\frac23)}+ \dot f_0^{(\frac34)}+... 
\; = \; 
\sum_{M=1}^{\infty} \dot f_{0}^{(\tilde\alpha_{M})}\,,
\ee 
where $\dot f_{0}^{(\tilde\alpha_{M})}$ are isolated coefficients of the  classical vacuum conformal blocks  at exceptional points $\tilde\alpha_{M} = \frac{M}{M+1}$.  Equivalently, the relation \eqref{f10}  can be represented  in terms of regular  $\alpha$. Namely, for all $ \alpha_M\in \left(\frac{M-1}{M}, 1\right)_\rn$, $M \in \mathbb{N}$,  
\be
\label{f10-non-exep}
\dot f_{1}^{(1)} = \dot f^{(\alpha_2)}_{2\alpha_2-1} + \dot f^{(\alpha_3)}_{3\alpha_3-2} + f^{(\alpha_4)}_{4\alpha_4-3}+...\; = \;  \sum_{M=1}^{\infty} \dot f_{M\alpha_M- (M-1)}^{(\alpha_{M})}\,,
\ee
where $\dot f_{M\alpha_M - (M-1)}^{(\alpha_{M})}$ are coefficients from $\pr \dot f^{(\alpha_M)}$ \eqref{f_sequence1}.
\end{prop}

\noindent The relation \eqref{f10}  is directly seen by recalling the form of the  classical vacuum conformal blocks at exceptional points, see \eqref{iso1}-\eqref{iso7}. The relation \eqref{f10-non-exep} is obtained from \eqref{f10} by virtue of  Proposition \bref{prop_vacuum}, since in the vacuum case one has $ f_{0}^{(\tilde\alpha_{M})} =  f_{M\alpha_M- (M-1)}^{(\alpha_{M})}$ for any $\alpha_M \in \left(\frac{M-1}{M},1\right)_\rn$, $M \in \mathbb{N}$. Note that each $\dot f_{0}^{(\tilde\alpha_{M})}$ in \eqref{f10} is of order $\delta^{M+1}$. It means that assuming the perturbation theory in small $\delta$ the leading term in $\dot f_{1}^{(1)}$ is given by the $\alpha=\half$ vacuum block,\footnote{In fact, this relation was originally  observed in \cite{Fitzpatrick:2014vua}.} while the subleadings are given by the isolated coefficients of the  classical vacuum conformal blocks at higher exceptional points. We extend this observation in  Section \bref{sec:pert}.     

The second and last term in $\pr \dot f^{(1)}$ is quite complicated already in the second order:
\be
\dot{f}_{0}^{(1)} = \frac{2\delta^2 z^4}{25(1 - 22 \delta)^2}\Big[ 1+ 2 z + \frac{(545 + 
\delta (-23514 +\delta (235213 + 16 (36451 - 247612 \delta) \delta)))}{196(1 - 22 \delta)^2}z^2 
+O(z^3)\Big].
\ee 
Contrary to $\dot{f}_{1}^{(1)}$, there is no obvious structure.

\section{Perturbative $\Delta$-expansions and $1/\Delta$-expansions }
\label{sec:pert}

Making conformal dimensions $\Delta  = \{h_i, \tilde h\}$ small or large yields a different type of asymptotic expansions. In particular, the  classical (Zamolodchikov) conformal block $f^{(1)}_1 = f^{(1)}_1(\delta_i, \delta|z)$ can be  expanded in small or large classical dimensions $\delta_i, \delta$. In this section, we show that the corresponding Taylor or Laurent series in $\delta_i, \tdelta$ are related to the $\alpha$-heavy classical conformal blocks which are Puiseux polynomials in $c$. 

Let all external conformal dimensions be different, i.e $\cF = \cF (h_i, \tilde{h}, c|z)$, $i=1,..., 4$ and the logarithmic Virasoro conformal block is expanded as   
\be
\label{general-log-f}
f(h_i, \tilde{h}, c|z) = \sum_{n=1}^{\infty}  g_n(h_i, \tilde{h}, c) z^n =  \sum_{n=1}^\infty  \frac{p_n(h_i, \tilde{h}, c)}{q_n(\tilde{h},c)}\,z^n\,, 
\ee
where $p_n, q_n$ are polynomial functions of their arguments, cf. \eqref{asymp_f}.  Introducing $h_i =  \delta_i c$, $\tilde{h} =  \tilde{\delta} c$  and expanding \eqref{general-log-f} around $c = \infty$ one extracts the leading term $c^1 f^{(1)}_1(\delta_i, \tilde{\delta}|z)$ which coefficient  is the Zamolodchikov conformal  block.  Then, expanding $f^{(1)}_1$ around $\delta_i, \tilde{\delta} = 0$ one obtains the Taylor (perturbative) series, 
\be
\label{perturbative-f1-small}
f^{(1)}_1(\delta_i, \tilde{\delta}|z) =  \sum_{m} P_m(\delta_i, \tilde{\delta} |z)\,,
\ee
where $P_m(\delta_i, \tilde{\delta}|z)$ are homogeneous functions  of degree $m$ in variables $\delta_i, \tilde{\delta}$ (see e.g. \eqref{1-pert}).  In the same way, one  expands around   $\delta_i, \tilde{\delta} = \infty$ to  obtain the Zamolodchikov conformal block as the Laurent (perturbative) series 
\be
\label{perturbative-f1-large}
f^{(1)}_1(\delta_i, \tilde{\delta}|z) = \sum_{k} S_{k}(\delta_i, \tilde{\delta} |z)\,,
\ee
where $S_k(\delta_i, \tilde{\delta}|z)$ are  homogeneous functions of degree $-k$ in variables $\delta_i, \tilde{\delta}$ as well as  power series in $z$ (see e.g. \eqref{1-non-pert}).  In what follows we show that the summation domains in \eqref{perturbative-f1-small} and \eqref{perturbative-f1-large} are $m = 1,..., \infty$ and $k=-1, 0,1,..., \infty$.  The next Proposition explains how the above perturbative expansions are related to  principal parts $\pr f^{(\alpha)}$ at $\alpha \in (0,1)_\rn$ and $\alpha\in (1, \infty)_\rn$.

\begin{prop} 
\label{prop31}
\begin{enumerate}
\item  Let the Zamolodchikov classical block $f^{(1)}_1$ be expanded as \eqref{perturbative-f1-small}.  Fixing  some $m \in \NN$ and $\alpha_m\in \left(\frac{m-1}{m},1\right)_\rn$ one finds 
\be
\label{prop3.1_1}
P_m(\delta_i, \tilde{\delta}|z) = f^{(\alpha_m)}_{m\alpha_m-m+1}(\delta_i, \tilde{\delta}|z)\,,
\ee
where $f^{(\alpha_m)}_{m\alpha_m-m+1}(\delta_i, \tilde{\delta}|z)$ are coefficients from $\pr f^{(\alpha_m)}$  (see Proposition \bref{prop1}).

\item Let the Zamolodchikov classical block $f^{(1)}_1$ be expanded as \eqref{perturbative-f1-large}. Let $k \in  \{-1,0\} \cup \NN$. Fixing some $k \in  \{-1,0\}$ and $\alpha_k\in \left(1,\infty\right)_\rn$ or $k \in \NN$ and $\alpha_k\in \left(1, \frac{k+1}{k}\right)_\rn$ one finds  
\be
\label{prop3.1_2}
S_k(\delta_i, \tilde{\delta}|z) = f^{(\alpha_k)}_{-k\alpha_k+k+1}(\delta_i, \tilde{\delta}|z)\,,
\ee
where $ f^{(\alpha_k)}_{-k\alpha_k+k+1}(\delta_i, \tilde{\delta}|z)$ are coefficients from $\pr f^{(\alpha_k)}$ (see Proposition \bref{prop2}).
\end{enumerate}
\end{prop}

\noindent A proof is given in Appendix \bref{app:proof_prop}. We see that  coefficients from $\pr f^{(\alpha)}$  are in one-to-one correspondence with $P_m$ for $\alpha \in (0,1)_\rn$ and with $S_k$ for $\alpha\in (1, \infty)_\rn$. Also, the Proposition sets  summation domains in \eqref{perturbative-f1-small} and  \eqref{perturbative-f1-large} as $m\in \NN$ and $k \in  \{-1,0\} \cup \NN$, i.e. they are the same as in  \eqref{lead_exp1} and \eqref{lead_exp2}. 

\begin{corollary}
\label{corollary31} Let  $m\in \NN$ and $\alpha_m\in \left(\frac{m-1}{m}, 1\right)_\rn$.  Then, the Zamolodchikov conformal block expanded around $\delta_i, \tilde{\delta} = 0$ can be represented  as 
\be
\label{perturbative-f1-small2}
f^{(1)}_1(\delta_i, \tilde{\delta}|z) =  \sum_{m=1}^\infty f_{m\alpha_m - (m-1)}^{(\alpha_m)}(\delta_i, \tilde{\delta}|z) \,.
\ee
Let $k \in \{-1,0\} \cup \NN$:  $\alpha_k\in (1,\infty)_\rn$ if $k\in \{-1,0\}$; $\alpha\in \left(1, \frac{k+1}{k}\right)_\rn$ if $k\in \NN$.  Then, the Zamolodchikov classical  block expanded around $\delta_i, \tilde{\delta} = \infty$ can be represented as 
\be
\label{perturbative-f1-large2}
f^{(1)}_1(\delta_i, \tilde{\delta}|z) = \sum_{k=-1}^\infty f_{-k\alpha_k+k+1}^{(\alpha_k)}(\delta_i, \tilde{\delta}|z)\,.
\ee
\end{corollary}

Coming back to the case of equal external dimensions ($\delta_i \equiv  \delta$), one finds out that a given coefficient $P_m(\delta, \tilde{\delta}|z)$ does not contain negative degrees of $\tilde{\delta}$ while remaining a homogeneous function of degree $m$ in both arguments.\footnote{More generally, this is true for pairwise equal dimensions,  $h_1 = h_2$ and $h_3 = h_4$.} Therefore,  it is a homogeneous polynomial of degree $m$, i.e. 
\be
\label{P_m}
P_m(\delta, \tilde{\delta}|z) = \sum_{k+l = m} a_{k,l}(z)\, \delta^{k} \tilde{\delta}^{l}\,,
\qquad\;\;
k,l  \in \mathbb{N}_0\,,
\quad m-1\in \mathbb{N}\,.
\ee
Examples of such $P_m$ are given in Appendix \bref{app:vir_num}, where Proposition \bref{corollary31} is illustrated by eqs. \eqref{3/5-heavy} and \eqref{1-pert}. Finally, note that  Proposition \bref{prop:Zam} directly follows from Proposition \bref{corollary31}  when $\tilde{\delta} = 0$. In this case, $P_m(\delta|z) = a_{m,0}(z) \delta^m$, cf. \eqref{P_m}.

\section{$\cW_3$  conformal blocks}
\label{sec:w3}

Let us briefly recall basic facts regarding $\cW_3$ conformal symmetry and introduce the $\cW_3$ conformal block function (for review see e.g.  \cite{Bouwknegt:1992wg,10.1007/BFb0105278}).  The commutation relations of the $\cW_3$ algebra are \cite{Zamolodchikov:1985wn} 
\be
\left[L_n, L_m\right]=(n-m) L_{m+n}+\frac{c}{12}{n\left(n^2-1\right)\delta_{n+m,0}} \,,
\qquad
\left[L_n, W_m\right]=(2 n-m) W_{m+n}\,,
\ee
\be
\ba{l}
\dps
\frac{2}{9}\left[W_m, W_n\right]=c\,\frac{m(m^2-1)(m^2-4)}{360}\delta_{m+n,0}
\vspace{2mm}
\\
\dps
\hspace{3mm}+(m-n)\left\{ \frac{(m+n+3)(m+n+2)}{15}  - \frac{(m+2)(n+2)}{6} \right\}L_{m+n}
+ \frac{16(m-n)}{22+5c} \Lambda_{m+n}\,,
\ea
\ee
where $\Lambda_{n+m}$ is given by 
\be
\Lambda_m = \gamma_m L_m + \sum_{n\in \mathbb{Z}} :L_{m-n}L_n: \,, 
\ee  
where $\gamma_{2k} = (1-k)(1+k)/5$ and $\gamma_{2k+1} = (1-k)(2+k)/5$, $m,n,k \in \mathbb{Z}$, and $:\;:$ stands for the normal ordering. The Verma module is defined by a highest-weight vector $\ket{h,w}$ satisfying  the following defining  conditions 
\be
\label{W_Verma}
\ba{l}
\dps
L_{0}\ket{h,w} = h \ket{h,w},
\qquad
W_{0}\ket{h,w} = w \ket{h,w},
\vspace{2mm}
\\
\dps
 L_{n}\ket{h,w} = 0\,,
\qquad
W_{m}\ket{h,w} = 0\,,
\qquad
m,n \in \mathbb{Z}^+ \,,
\ea
\ee
with two (independent) parameters $h,w$ being a conformal dimension and a spin-3 charge. The module is spanned by the basis monomials ${Y} \equiv  L_{-\boldsymbol{\lambda}} W_{-\boldsymbol{\mu}}:= L_{-\lambda_1}... \,L_{-\lambda_k}W_{-\mu_1}... \, W_{-\mu_l}$, where $\lambda_i, \mu_i \in \mathbb{Z}^+$ . The sum $|Y| := \lambda_1 + ... + \lambda_k + \mu_1 + ... + \mu_l$ is a level: $L_0 Y \ket{h,w} = (|Y|+h)Y \ket{h,w}$. The inner product is defied by 
\be
\label{W3-conj}
L_n^\dagger = L_{-n}\,,
\qquad
W_n^\dagger = -W_{-n}\,,
\ee
which means that $h \in \RR$ and $w \in i\RR$.

The operator-state correspondence between  highest-weight vectors $\ket{h,w}$ and  primary operators $\cO_{h,w}(z)$ in $\cW_3$ CFT$_2$ can be introduced in the standard fashion by means of the   general vertex operator construction (see e.g. \cite{Goddard:196381,Gaberdiel:1999mc}), 
\be
\label{VOA}
\ba{l}
\dps
\ket{h,w} = \lim_{z, \bar z\to 0} \cO_{h,w}(z, \bar z)\ket{0}\,,
\vspace{2mm}
\\
\dps
\bra{h,w} = \lim_{z, \bar z\to 0} (-1/\bar z^2)^h (-1/z^2)^{\bar h}\bra{0} \cO^*_{h,w}(1/\bar{z}, 1/z)\,, 
\ea
\ee
where   $\cO_{h,w}^*(z, \bar z) \equiv \cO_{h,-w}(z, \bar z)$ is the charge conjugated  operator which therefore have an opposite spin-3 charge, $*^2 = 1$.  In particular, the 2-point correlation function  is related to the inner product as 
\be
\label{2pt}
1=\braket{h, w|h, w} = \lim_{z, \bar z\to \infty} (-z^{2})^{h} (-\bar{z}^{2})^{\bar h} \braket{\cO_{h,-w}(\bar z, z) \cO_{h,w}(0)}\,.
\ee
Therefore, the inner product is defined to have zero total spin-3 charge.\footnote{In what follows anti-holomorphic labels will be  suppressed for brevity.}  
 
As discussed in the Introduction, the leftover freedom in defining matrix elements and correlation functions  related to the presence of $(W_{-1})^n$ descendants can be fixed by imposing extra constraints which are  null-conditions in the respective Verma modules (for more details see  e.g. \cite{Fateev:1987vh,Belavin:2016wlo}).
  
\begin{definition}
\label{def_semi}
A semi-degenerate primary operator $\cO_{h,w}(z)$  satisfies the null-state condition on the 1st level \cite{Bowcock:1992gt}:  
\be
\label{semidegenerate1}
W_{-1} \cO_{h,w}(z)=\frac{3 w}{2 h} L_{-1} \cO_{h,w}(z)\,,
\ee
where $h$ and $w$ are polynomially related as 
\be
\label{degW}
\left[\frac{32}{22+5c}\left(h+\frac{1}{5}\right)-\frac{1}{5}\right]h^2 = w^2\,.  
\ee
\end{definition}

\subsection{Conformal block via matrix elements}
\label{sec:mat_block}

Consider the 4-point  function of arbitrary primary operators $\left\langle \cO_{a_1}(z_1)\cO_{a_2}(z_2)\cO_{a_3}(z_3)\cO_{a_4}(z_4) \right\rangle$. Collectively denoting the weights as $a =(h,w)$ one then considers the operator product expansion (OPE)   
\be
\label{OPE_W3}
\cO_{a_i}(z_i) \cO_{a_j}(z_j) = \sum_{a_k} f_{a_i a_j}^{a_k}\,\mathbb{C}_{a_ia_ja_k}(z_{ij}, \partial_j) \cO_{a_k}(z_j)\,,
\ee  
where $f_{a_i a_j}^{a_k}$ are structure constants and  OPE coefficients are packaged into differential operators $\mathbb{C}_{a_ia_ja_k}$ which are defined by conformal transformation properties of the left operator product. In general, the OPE coefficients depend on the fusion multiplicities related to the presence of  $W_{-1}$ and $W_{-2}$ descendants. Fixing  the $s$-channel one introduces the $\cW_3$ conformal blocks $\cF$ by acting twice with the OPEs, 
\be
\ba{l}
\dps
\left\langle \cO_{a_1}(z_1)\cO_{a_2}(z_2)\cO_{a_3}(z_3)\cO_{a_4}(z_4) \right\rangle = 
\sum_{s,p} f_{a_1a_2}^{s}f_{a_3a_4}^{p}\mathbb{C}_{a_1a_2 s}(z_{12}, \partial_2) \mathbb{C}_{a_3a_4 p}(z_{34}, \partial_4) 
\left\langle \cO_s(z_2)  \cO_p(z_4) \right\rangle
\vspace{2mm}
\\
\dps
= \sum_{p} f_{a_1a_2}^{p^*}f_{a_3a_4}^{p}\mathbb{C}_{a_1a_2 p^*}(z_{12}, \partial_2) \mathbb{C}_{a_3a_4p}(z_{34}, \partial_4) 
\left\langle \cO^*_p(z_2)  \cO_p(z_4) \right\rangle = \sum_{p}f_{a_1a_2}^{p^*}f_{a_3a_4}^{p}\, \cF(h_i,w_i, \tih_p, \tiw_p, c|z_i)\,,  
\ea
\ee 
where $p^* = (\tih_p,-\tiw_p)$ is charge conjugated with respect to $p = (\tih_p,\tiw_p)$. The form of the   2-point function \eqref{2pt}  in the last line suggests using  {\it the double line notation} to designate that the exchange channel is defined by two operators charge-conjugated to each other, $*^2 = 1$, see fig. \bref{fig:comb_W3}. 

Within the operator formulation the same conformal block expansion  can be obtained by inserting  a projector onto the Verma module $\mathcal{V}_{\tilde{h}, \tilde{w}}$ with weights $(\tilde h, \tilde w)$  between pairs $\cO_1\cO_2$ and $\cO_3\cO_4$. In this way, sending the points $(z_1,z_2,z_3,z_4)$  to $(\infty, 1,z,0)$ and introducing a bare conformal block $\wb$ by factoring out a leading power of $z$ in $\cF$, one  finds  \cite{Mironov:2009dr,Kanno:2010kj,Fateev:2011hq} 
\be
\label{blockW}
\ba{l}
\wb(h_i,w_i, \tih, \tiw,   c|z) = 
\vspace{3mm}
\\
\dps
\hspace{10mm}= \sum_{|{Y}|=\left|{Y}^{\prime}\right|}z^{|{Y}|}\; 
\Gamma^*(h_{1,2}, w_{1,2}; \tilde{h}, \tilde{w}|{Y})
\; Q_{\tilde{h},\tilde{w}}^{-1}\left({Y}, {Y}^{\prime}\right) \;
\Gamma(h_{3,4}, w_{3,4};\tilde{h}, \tilde{w}|Y')
\vspace{1mm}
\\
\dps
\hspace{20mm}\equiv 1+ B_1 z+B_2 z^2 +B_3 z^3+... \;= \; \sum_{n=0}^\cN B_n z^n+ O(z^{\cN+1})\,.
\ea
\ee
Here, $Q^{-1}_{\tilde{h}, \tilde{w}}\left({Y}, {Y}^{\prime}\right)$ is the inverse Gram  matrix of the Verma module $\mathcal{V}_{\tilde{h},\tilde{w}}$, 
\be
\label{gram}
Q_{\tilde{h}, \tilde{w}}\left({Y}, {Y}^{\prime}\right) = \bra{\tilde{h}, \tilde{w}}Y^\dagger Y'\ket{\tilde{h}, \tilde{w}}\,,
\ee
and the 3-point matrix elements   
\be
\label{3ptW}
\ba{l}
\dps
\Gamma^*(h_{1,2}, w_{1,2}; \tilde{h}, \tilde{w}|{Y}) = \big\langle h_1, -w_1|  \cO_{2}(1)  Y |\tilde{h}, \tilde{w}\big\rangle\,,
\vspace{2mm} 
\\
\dps
\;\Gamma(h_{3,4}, w_{3,4};\tilde{h}, \tilde{w}|Y)= \big\langle \tilde{h}, \tilde{w}|   Y^{\dagger}\cO_{3}(1)| h_4, w_4 \big\rangle\,, 
\ea
\ee
are associated to  the 3-point functions of two primary  and one secondary operators.  We notice that dual  matrix elements are not equal to each other, $\Gamma^* \neq  \Gamma$, which is different from Virasoro case, where $\Gamma^*=  \Gamma$. The expansion coefficients  $B_n = B_n(h_i, w_i, \tih, \tiw, c)$ are rational functions.\footnote{The  literature contains  sign conventions for conjugating $W_n$ operators that differ from \eqref{W3-conj}, e.g. see \cite{Kanno:2010kj} therein (19) and (30). However, it can be verified that the conformal block coefficients do not depend on the choice of a specific sign convention, see Appendix  \bref{app-signs}.}

If none of five involved representations are reducible than the conformal block depends on the free parameters which are particular 3-point matrix elements 
\be
\Gamma^*_n = \bra{h_1, -w_1} \cO_2(1) (W_{-1})^n \ket{\tih,\tilde{w}} 
\quad \text{and} \qquad 
\Gamma_n = \bra{\tih,\tilde{w}} (W_{1})^n \cO_3(1)  \ket{h_4,w_4}\,,
\ee 
where $n\in \mathbb{Z}^+_0$. It means that the $\cW_3$ conformal block can be  defined unambiguously only for particular conformal dimensions (in which case only the 3-point structure constants  $\Gamma^*_0$ and $\Gamma_0$ remain undetermined). In fact, taking one of three  operators in each 3-point function to be semi-degenerate is sufficient to lift the degeneracy  (see Appendix \bref{app:W3-deg}). It follows that the conformal block expansion is well-defined for two families of operators: (1) two of four external operators are semi-degenerate, the intermediate operator is arbitrary; (2) all four external operators are arbitrary, the intermediate operator is semi-degenerate. It is worth noting that in general the $\cW_3$ vacuum conformal block belongs to  the family (2) with pairwise equal external operators. Indeed, the vacuum block it is defined by choosing an identity exchange operator $(h, w=0)$ which is the simplest semi-degenerate operator \eqref{degW}. Specifying two of external operators to be semi-degenerate the vacuum block may also belong to  the family (1).   

In what follows we choose the operators as follows (see fig. \bref{fig:comb_W3}):
\be
\label{W3opers}
\ba{l}
\text{external operators:}\hspace{15mm} \cO_{1,3} = \cO_{h,w}\;\; \text{and} \;\; \cO_{2,4} = \cO^*_{h,w}
\vspace{2mm}
\\
\text{intermediate operator:}\hspace{29mm} \tilde \cO = \tilde \cO_{\tih,0} 
\ea
\ee
in which case the conformal block coefficients are drastically simplified (another simple choice of operators is discussed in Appendix \bref{app:B1}).  Note that the corresponding 4-point correlation function \eqref{4ptW} is invariant under  $w \to -w$ since $\cO^*_{h,w} = \cO_{h,-w}$. This symmetry is inherited for the   conformal blocks which, therefore, depend on $w^2$ that makes the block coefficients real functions (in other cases the $w$-dependence may lead to complex coefficients since $w\in \mathbb{C}$ for general $h$, cf. \eqref{degW}). 

By combining analytical and numerical calculations, we have succeeded to find explicitly  the  first three coefficients (up to $O(z^4)$ in \eqref{blockW}). The first coefficient is given by\footnote{The number of terms in the next coefficient $B_2$ is comparable to the number of terms  in the Virasoro block coefficient $F_4$. Our {\it Mathematica} code (on the standard laptop, dozens of hours) can help  to calculate the coefficient $B_3$. Higher-order coefficients are not yet available.}  
\be
B_1 =  \tilde{h}\frac{16(\tilde{h}+h) +2-c }{32\tilde{h} +2-c}\,,
\ee
see Appendix \bref{app:B1}. The pole $\tih = (c-2)/32$ here is the zero of the determinant of the Kac matrix of $\cW_3$ at level one \cite{Mizoguchi:1988vk} (see also e.g. \cite{Afkhami-Jeddi:2017idc,Carpi:2019szo}).   

\subsection{Large-$c$ asymptotic expansions}
\label{sec:structureW}

Similar to \eqref{def}  one  defines the $\cW_3$ logarithmic conformal block  as 
\be
\label{def-log-W3}
\cb(h,w, \tih,  c|z) \coloneqq \log[\wb(h,w, \tih, c|z)] 
\ee
and considers  its large-$c$ asymptotics by assuming that  dimensions and charges are given by  
\be
\label{h_th_aW}
\ba{l}
\dps
\; \alpha \in (0,1)_\rn:\quad\; h = \delta c^\alpha\,,
\qquad
w^2 =  - \frac{\delta^2}{5}  c^{2\alpha}\,,
\hspace{12mm}
\tih = \tdelta c^\alpha\,, 
\vspace{2mm}
\\
\dps
\alpha\in (1, \infty)_\rn:\quad\; h = \delta c^\alpha\,,
\qquad
w^2 = \frac{32 \delta^3}{5}  c^{3\alpha-1} \,,
\qquad
\tih = \tdelta c^\alpha\,,  
\vspace{2mm}
\\
\dps
\hspace{12mm}\alpha =1: \quad\; h = \delta c\,,
\qquad
w^2 = - \frac{\delta^2}{5}(1-32 \delta)  c^{2} \,,
\qquad
\tih = \tdelta c\,,
\ea
\ee
cf.   \eqref{degW}. Note that regions $\delta> 0$ at $\alpha\in (1, \infty)_\rn$ and $\delta< 1/32$ at $\alpha=1$ are manifestly non-unitary. This agrees with the unitary bound $h\geqslant (c-2)/32$ found in \cite{Afkhami-Jeddi:2017idc}.

Introducing the $\alpha$-heavy  logarithmic conformal block 
\be
\label{def-log-W3a}
\cb^{(\alpha)}(\delta, \tdelta |z) = \sum_{n=1}^\cN G^{(\alpha)}_n(h, \tilde h, c) z^n + O(z^{\cN+1})\,, 
\ee
and expanding near $c = \infty$ one finds  the large-$c$ asymptotics   
\be
\label{exp_prin2}
\wb(h,w, \tih, c|z)  \; \simeq \;  e^{\ws^{(\alpha)}(c|z)} \left[1 + O(1/c^\beta)\right]\,,
\qquad \exists\, \beta >0\,,
\ee
where the Puiseux polynomial  $\ws^{(\alpha)}(c|z) =  \pr \cb^{(\alpha)}(\delta, \tdelta |z)$ is now the classical $\cW_3$ conformal block. It turns out that the main results for  the principal part of the large-$c$ logarithmic $\cW_3$ conformal block are similar to  those obtained in the Virasoro case (see propositions \bref{prop1} and \bref{prop2}). Nevertheless, here we explicitly  formulate the analogous  proposition but in less detail. 

The  primary operators are listed in \eqref{W3opers}; the heaviness  parameter $\alpha\in \QQ$;  the conformal dimensions and spin-3 charges  are given by \eqref{h_th_aW}. The domains of $\alpha$ and their splitting in half-open intervals are the same as in Section \bref{sec:large_c}. The following Propositions are valid up to  $O(z^4)$ (see Appendix \bref{app:W3}). 

\begin{prop}
\label{prop1WX}

\begin{enumerate}
\item Let\; $\alpha\in (0,1)_\rn$ and $M\in \mathbb{N}$ be such that $\alpha \in \big(\frac{M-1}{M}, \frac{M}{M+1} \big]_\rn$. Then, the classical conformal block in  \eqref{exp_prin2} is the following Puiseux polynomial:   
\be
\label{lead_exp1W}
\ws^{(\alpha)}(\delta, \tilde\delta, c|z) = \sum_{m=1}^{M} \cb_{m \alpha -m+1}^{(\alpha)}\, c^{m \alpha -m+1}+\cb_0^{(\alpha)}\;.
\ee
The coefficient functions in  \eqref{lead_exp1W}, where regular $\alpha_M \in \big(\frac{M-1}{M} , \frac{M}{M+1} \big)_\rn$ and $M=1,2,...\,$ satisfy the relations
\be
\label{columneq1W}
\ba{l}
\cb_{\alpha_1} ^{(\alpha_1)} = \cb_{\alpha_2} ^{(\alpha_2)} = \cb_{\alpha_3} ^{(\alpha_3)} = \ldots\;,
\vspace{2mm}
\\
\dps
\cb_{2\alpha_1-1} ^{(\alpha_1)} = \cb_{2\alpha_2-1} ^{(\alpha_2)} = \cb_{2\alpha_3-1} ^{(\alpha_3)} = \ldots\;,
\vspace{2mm}
\\
\dps
\;\;\vdots
\vspace{2mm}
\\
\dps
\cb_{0}^{(\alpha_1)} = \cb_{0}^{(\alpha_2)} = \cb_{0} ^{(\alpha_3)} = \ldots\;.
\ea
\ee 
For all lines  except the last one we can extend the interval $\big(\frac{M-1}{M} , \frac{M}{M+1} \big)_\rn$ to $\big(\frac{M-1}{M} , \frac{M}{M+1} \big]_\rn$. At the right endpoints which are exceptional points $\tilde \alpha_M = \frac{M}{M+1}$ the last line contains essentially different coefficient functions:
\be
\label{f0_exceptional1W}
\cb_{0}^{(\tilde \alpha_1)} \neq  \cb_{0}^{(\tilde \alpha_2)} \neq  \cb_{0} ^{(\tilde\alpha_3)} \neq  ...\;. 
\ee
 
\item Let $\alpha\in(1, \infty)_\rn$ and $N\in \mathbb{N}$ be such that $\alpha \in \big[\frac{N+1}{N} , \frac{N}{N-1} \big)_\rn$. Then, the classical conformal block in  \eqref{exp_prin2} is the following Puiseux polynomial: 
\be
\label{lead_exp2W}
\ws^{(\alpha)}(\delta, \tilde\delta, c|z) = \sum_{n=-1}^{N-1} \cb_{-n\alpha +n+1}^{(\alpha)}\, c^{-n\alpha +n+1}+\cb_0^{(\alpha)}\,.
\ee
The  coefficients in \eqref{lead_exp2W}, where $\alpha_N \in \big(\frac{N+1}{N} , \frac{N}{N-1} \big)_\rn$ and $N=1,2,...\,$, satisfy the relations 
\be
\label{columneq2}
\ba{l}
\ldots  = \cb_{\alpha_3} ^{(\alpha_3)} = \cb_{\alpha_2} ^{(\alpha_2)} =  \cb_{\alpha_1} ^{(\alpha_1)}\;,
\vspace{2mm}
\\
\dps
\ldots  =  \cb_{1} ^{(\alpha_3)} = \cb_{1} ^{(\alpha_2)} = \cb_{1} ^{(\alpha_1)}\;,
\vspace{2mm}
\\
\dps
\hspace{35mm} \;\;\vdots
\vspace{2mm}
\\
\dps
\ldots  =  \cb_{0}^{(\alpha_3)} = \cb_{0}^{(\alpha_2)} = \cb_{0} ^{(\alpha_1)}\;.
\ea
\ee 
For all lines except the last one we can extend the interval $\big(\frac{N+1}{N} , \frac{N}{N-1} \big)_\rn$ to $\big[\frac{N+1}{N} , \frac{N}{N-1} \big)_\rn$. At the left endpoints which are exceptional points $\overset{\approx}{\alpha}_N = \frac{N+1}{N}$ the last line contains  essentially different coefficient functions:
\be
\label{f0_exceptional2}
\ldots  \neq  \cb_{0}^{(\overset{\approx}{\alpha}_3)} \neq \cb_{0}^{(\overset{\approx}{\alpha}_2)} \neq \cb_{0} ^{(\overset{\approx}{\alpha}_1)}\;.
\ee
\end{enumerate}
\end{prop}

\noindent  The cases of $\alpha=0,1$ are considered in Appendix \bref{app:w3_numerics}: 
\be
\ba{l}
\alpha=1:\qquad  \ws^{(1)}(\delta, \tilde\delta, c|z) = \cb_1^{(1)}c +\cb_0^{(1)}\,, \qquad \text{eqs. \eqref{f11W}-\eqref{f01W}}\,,
\vspace{2mm} 
\\
\dps
\alpha=0:\qquad  \ws^{(0)}(\delta, \tilde\delta, c|z) = \cb_0^{(0)}\,, \qquad  \text{eq. \eqref{f0W}}
\,.
\ea
\ee
The counting function for $\cW_3$ classical conformal blocks is the same as in the Virasoro case, see fig. \bref{fig_graphic}.

\section{Conclusion}
\label{sec:concl}

We have formulated the problem of description of the conformal blocks for $\alpha$-heavy operators, where $\alpha$-heaviness is defined by  general scaling behaviour of conformal dimensions and spin charges in   the large-$c$ regime  as $O(c^\alpha)$ with $\alpha\in \mathbb{Q}^+_0$. We have conjectured that the Virasoro and $\cW_3$ conformal blocks are exponentiated similar to the standard case $\alpha=1$ but the leading exponents are more involved -- they are represented by Puiseux polynomials in the central charge. Our consideration is both numerical and analytical: the conformal  block coefficients can be obtained explicitly up to high order by using   {\it Mathematica};  their  large-$c$ expansions can be analyzed analytically.  

Despite the relative  simplicity of the exponentiation procedure for any $\alpha \in \mathbb{Q}^+_0$ we underline that the real challenge is to describe the  process of restructuring  the classical conformal block (i.e. of the respective principal part $\pr f^{(\alpha)}$) at changing the  heaviness parameter  $\alpha$. On this way, we have observed a number of interesting related properties and methods: 

\begin{itemize}

\item $\alpha$-dependent counting of terms in the Puiseux polynomial $\pr f^{(\alpha)}$;

\item transmutation of singularities;

\item reconstruction of the  $\alpha$-heavy classical vacuum conformal blocks;

\item  asymptotic expansions and Newton polygons;

\item perturbation theory and $\alpha$-heavy operators.

\end{itemize}

Also, in order to formulate the conformal block expansion  in CFT$_2$ with $\cW_3$   symmetry we have reconsidered the fusion rules and found new constraints 3-point functions with two semi-degenerate operators. Note that we have  not analyzed the transmutation of singularities for the $\cW_3$ logarithmic conformal blocks because of our limited knowledge of the Gram matrix in the $\cW_3$ Verma modules and, respectively, the Kac dimensions/poles of the conformal blocks. Also, we have omitted from our analysis  perturbative classical $\cW_3$ conformal blocks. These  issues will be considered elsewhere.

Since the Virasoro and  $\cW_N$ conformal block functions are not presently known analytically for the general values of their arguments (conformal dimensions, higher-spin charges, the central charges, and coordinates) and can only be given  asymptotically, the $\alpha$-scalings (and their generalizations  \eqref{ab_dim_h})  naively just add another layer of complexity. However, this mechanism can be considered a useful tool for studying the asymptotic behaviour of the  conformal block coefficients in more general CFTs. E.g.  the counting functions, like that one shown on fig. \bref{fig_graphic}, with one or more peaks and infinite discontinuities may clarify the structure of  perturbation theory in  small/large conformal dimensions and higher-spin charges in the large-$c$ regime. Indeed, it would be interesting to study perturbative expansions of the  $\cW_N$ classical  blocks along the lines of Section \bref{sec:pert}. Such a  task is rather non-trivial since  in the general case there are two heaviness parameters $\alpha$ and $\beta$ (related to conformal dimensions and spin-3 charges) and one may wonder how  the associated $\alpha\beta$-heavy classical conformal blocks  control the perturbative expansion of the $\alpha=\beta=1$ $\cW_N$ classical conformal block.

At the present stage, our study of $\alpha$-heavy operators and their large-$c$ conformal blocks is rather formal and needs possible applications including the holographic interpretation. In particular, it is interesting to understand the gravitational role of subleading terms in $\alpha$-heavy classical conformal blocks given by Puiseux polynomials. At the same time, one may continue studying the $\alpha$-heavy classical conformal blocks  in other different directions.  

For example, they can be analyzed  within the free field realization of the Virasoro and $\cW_3$ algebra and the conformal block  in the spirit of \cite{Besken:2019jyw}. In addition, in the same vein, one can consider the case of $\cW_N$ conformal blocks  with  $N>3$, where we do not expect any fundamental obstacles except for the rapidly growing  technical difficulties associated   with calculating matrix elements and finding (inverse) Gram matrices. 

It would be interesting to see whether it is possible to find analytic expressions for $\alpha$-heavy vacuum conformal blocks at exceptional values of $\alpha$  using e.g. the direct matrix approach  of \cite{Fitzpatrick:2014vua}. According to  Proposition \bref{prop_iden} that will lead to the explicit analytic form of the exponentiated $\alpha$-heavy vacuum conformal  blocks for any $\alpha\in (0,1)_\rn$.  Also, one may  consider different scalings of operators, $h_i = O(c^{\alpha_i})$ that generalizes the HHLL conformal block, where  $h_H = O(c^1)$ and $h_L = O(c^0)$ \cite{Fitzpatrick:2015zha}. It is tempting to speculate that various relations between $\alpha$-heavy classical conformal blocks at different $\alpha\in \mathbb{Q}^+$ observed in this paper can be systematically obtained by means of considering CFTs in  different $\alpha$-parameterized  background geometries  \cite{Fitzpatrick:2015zha}.
 
\vspace{3mm}

\noindent \textbf{Supplemental material.} We provide  {\it Mathematica} notebooks with the arXiv submission for download. They include notebooks which calculate
\begin{itemize}
\item Gram matrices for Verma modules of $Vir$ and $\cW_3$ algebras;

\vspace{-2mm} 

\item the classical conformal blocks for $\alpha$-heavy operators for $Vir$ algebra;

\vspace{-2mm}

\item the classical conformal blocks for four semi-degenerate $\alpha$-heavy operators for $\cW_3$ algebra;

\vspace{-2mm}

\item the first twenty coefficients of the  Virasoro vacuum conformal block.
\end{itemize}

\vspace{3mm}

\noindent \textbf{Acknowledgements.} We are grateful to Danil Zherikhov for fruitful collaboration at early  stage of this project. We benefited from discussions with Vladimir  Belavin, Alexey Litvinov, Semyon  Mandrygin, Mikhail Pavlov. Our special thanks to  Alexey Litvinov for sharing with us his code for calculating  conformal blocks  which we have further modified, and especially to Aleksandr Kovalenko for his significant help in programming with {\it Mathematica}. Our work was supported by the Foundation for the Advancement of Theoretical Physics and Mathematics “BASIS”.

\appendix
\section{Notation for sets of  numbers} 
\label{app:notation}

We use the following (non)standard notation and conventions:
$$
\ba{ll}
\text{natural, integer, rational numbers:} & \NN, \ZZ, \mathbb{Q} 

\\

\text{positive (negative) integer, rational, real numbers:} & \ZZ^{\pm}, \mathbb{Q}^{\pm}, \mathbb{R}^{\pm} 

\\

\text{non-negative(positive) integer  numbers:} & \ZZ^\pm_0  = \{0\} \cup \ZZ^\pm

\\

\text{non-negative(positive) rational  numbers:} & \mathbb{Q}^\pm_0 = \{0\} \cup \mathbb{Q}^\pm 

\\

\text{non-negative(positive) real  numbers:} & \mathbb{R}^\pm_0 = \{0\} \cup \mathbb{R}^\pm 

\\

\text{(half)-open intervals on $\mathbb{Q}$:} & (a,b)_\rn : = (a,b) \cap \mathbb{Q}\,, \quad a,b \in \mathbb{Q}   

\\

 & \, [a,b)_\rn  : = [a,b) \cap \mathbb{Q}\,, \quad a,b \in \mathbb{Q}   

\\

& \, (a,b]_\rn  : = (a,b] \cap \mathbb{Q}\,, \quad a,b \in \mathbb{Q}   

\ea
$$

\section{The logarithmic $Vir$ block coefficients}
\label{app:vir}

In this Appendix we analyze the large-$c$ expansion of the lowest-order coefficients of the logarithmic block $g_{n}^{(\alpha)}$ at $n=2,3,4,5$ \eqref{f_exp_alpha}. Our analysis here is valid for real  $\alpha \in \RR_0^+$ which we restrict to rational $\alpha\in \mathbb{Q}^+_0$ in the main text. The  two coefficients, $g_2^{(\alpha)}$ and $g_4^{(\alpha)}$ are explicitly considered. Other two available coefficients $g_3^{(\alpha)}$ and $g_5^{(\alpha)}$ have the same dependence on $c$ (due to equal external dimensions) as respectively $g_2^{(\alpha)}$ and $g_4^{(\alpha)}$, and, therefore, they need not be analyzed. The final Propositions of Section  \bref{sec:structure} are valid up to $O(z^6)$.

\subsection{Second coefficient}
\label{app:second}

Let $\alpha \in \RR_0^+$. Consider  the second coefficient $g_2^{(\alpha)}$ \eqref{f_2ex} which, for convenience, we rename in this section as $\g2^{(\alpha)}$:
\be
\label{f_2ex2}
\g2^{(\alpha)} = \frac{A c^{3\alpha} + B c^{2\alpha+1} + C c^{2\alpha}+ D c^{\alpha+1} }{M c^{2\alpha} +N c^{\alpha+1}+K c^{\alpha}+L c^1}\,.  
\ee  
The powers of $c$ here are given by $k+m\alpha$, where $k,m  \in  \mathbb{Z}$. In other words, particular functions $c^{k+m\alpha}$ can be represented as nods on a two-dimensional lattice. In this way, we can concisely  describe both the numerator and denominator in \eqref{f_2ex2} as the Newton polygons shown on  fig. \bref{fig-num-g2} and fig. \bref{fig-den-g2}, where  the nods (bold black dots) represent non-zero terms; the nods are connected by edges (red lines) to form a Newton polygon. 

\begin{figure}[h!]
 \begin{minipage}{0.48\textwidth}
     \centering
     \begin{tikzpicture}[scale = 0.5]
        \draw[thin,dotted] (0,0) grid (2,3);
        \draw[->] (0,0) -- (2.4,0) node[right] {$k$};
        \draw[->] (0,0) -- (0,3.4) node[above] {$m$};
        \foreach \x/\xlabel in { 1/1, 2/2}
    \draw (\x cm,1pt ) -- (\x cm,-1pt ) node[anchor=north,fill=white] {\xlabel};
  \foreach \y/\ylabel in {1/1, 2/2, 3/3}
    \draw (1pt,\y cm) -- (-1pt ,\y cm) node[anchor=east, fill=white] {\ylabel};
    \draw[thick, red] (0,3) -- (1,2);
    \draw[thick, red] (0,3) -- (0,2);
    \draw[thick, red] (0,2) -- (1,1);
    \draw[thick, red] (1,1) -- (1,2);
    \draw[fill=black] (1,1) circle (0.12) node[above right] {};
    \draw[fill=black] (0,2) circle (0.12) node[above right] {};
    \draw[fill=black] (1,2) circle (0.12) node[above right] {};
    \draw[fill=black] (0,3) circle (0.12) node[above right] {};
    \end{tikzpicture}
    \caption{The numerator in $\g2^{(\alpha)}$. }
    \label{fig-num-g2}
 \end{minipage}\hfill
  \begin{minipage}{0.48\textwidth}
     \centering
     \begin{tikzpicture}[scale = 0.5]
        \draw[thin,dotted] (0,0) grid (2,3);
        \draw[->] (0,0) -- (2.4,0) node[right] {$k$};
        \draw[->] (0,0) -- (0,3.4) node[above] {$m$};
        \foreach \x/\xlabel in { 1/1, 2/2}
    \draw (\x cm,1pt ) -- (\x cm,-1pt ) node[anchor=north,fill=white] {\xlabel};
  \foreach \y/\ylabel in {1/1, 2/2, 3/3}
    \draw (1pt,\y cm) -- (-1pt ,\y cm) node[anchor=east, fill=white] {\ylabel};
    \draw[thick, red] (0,2) -- (0,1);
    \draw[thick, red] (0,2) -- (1,1);
    \draw[thick, red] (0,1) -- (1,0);
    \draw[thick, red] (1,0) -- (1,1);
    \draw[fill=black] (1,0) circle (0.12) node[above right] {};
    \draw[fill=black] (0,1) circle (0.12) node[above right] {};
    \draw[fill=black] (1,1) circle (0.12) node[above right] {};
    \draw[fill=black] (0,2) circle (0.12) node[above right] {};
    \end{tikzpicture}
    \caption{The denominator in $\g2^{(\alpha)}$.}
    \label{fig-den-g2}
  \end{minipage}
\end{figure}

\subsubsection{${\alpha\in (1, \infty)}$} 
\label{app:second_1inf}

Let us parameterize real  $\alpha \in (1, \infty)$ as $\alpha = 1 + \varkappa $, where $\varkappa \in (0,\infty)$. The coefficient \eqref{f_2ex2} can be equally represented  as:
\be
\label{g2_1inf}
\g2^{(1+\varkappa)}  = \left(\frac{A}{M} c^{1+\varkappa} + \frac{B}{M} c^{1} + \frac{C}{M} c^0+ \frac{D}{M} c^{-\varkappa}\right)\,\cfrac{1}{\dps 1+\frac{N}{M}c^{-\varkappa} + \frac{K}{M}c^{-1-\varkappa} + \frac{L}{M}c^{-1-2\varkappa}}\,.
\ee
The second factor has the form  $\frac{1}{1+X(c)}$, where $X(c) = \frac{N}{M}c^{-\varkappa} + \frac{K}{M}c^{-1-\varkappa} + \frac{L}{M}c^{-1-2\varkappa}$  contains only negative powers of $c$ so that  $\lim_{c\to \infty}  X(c) =  0$. Then, the small-$X(c)$ expansion of the second factor  yields  the  Puiseux  series. 

A convenient trick  is to represent the resulting series  as a double series in $x \equiv c$ and $y \equiv c^\varkappa$ variables treated as independent. Indeed, a general term of the Puiseux series  is   $c^{-n -m \varkappa}$ with integer $n,m$ running over some domain in  two-dimensional lattice $\mathbb{Z}\oplus \mathbb{Z}$. Such a term can be split as $x^{-m} \, y^{-n}$. In fact, this splitting allows one to read off the domain directly from \eqref{g2_1inf} that gives the following double expansion 
\be
\label{main_pui}
\g2^{(1+\varkappa)} = \sum_{n=-1}^{\infty} \sum_{m\geq n}\g2_{n,m}^{(1+\varkappa)}\, x^{-n} y^{-m} \;\equiv\; \sum_{n=-1}^{\infty} \sum_{m = n}^{\infty}\g2_{n,m}^{(1+\varkappa)}\, c^{-n- m\varkappa}\;.
\ee
Here, the expansion coefficients $\g2_{n,m}^{(1+\varkappa)}$ are  rational functions of $A,B,C,...\,$ which can be conveniently organized as the following set  
\be 
\label{array_1inf0}
\hspace{-2mm}\bold{P}^{(1+\varkappa)} =
\left\{ \g2_{n,m}^{(1+\varkappa)}~|~n+1, m+1\in\mathbb{Z}^+_0,\; m\geq n \right\}\;.
\ee 

Note that the splitting  trick is directly applicable  only when $\varkappa$ is irrational. Otherwise, variables $c$ and $c^{\varkappa}$ and their powers can be entangled like e.g. for a particular rational  $\varkappa=1$ (to some extent, this phenomenon resembles an (ir)rational winding of a torus). In the sequel, we carefully analyze these two different situations of rational and irrational $\varkappa$ keeping the splitting trick as a convenient and useful  parameterization of the Puiseux series.

The principal part of \eqref{main_pui} is given by   
\be
\label{prin_1inf}
\pr\g2^{(1+\varkappa)} =  \sum_{m=-1}^{N} \g2^{(1+\varkappa)}_{-1,m}\, c^{1-m\varkappa} + \g2_{0,0}^{(1+\varkappa)}\,,
\ee
with  $N$ being a maximum integer number depending on $\varkappa$ such that  
\be
\label{N_ineq}
N = N(\varkappa): \qquad 1-N \varkappa > 0\,.
\ee 
Thus, in the principal part there are $N+3$ terms  of degrees   $\{0, 1-m\varkappa ~| ~m=-1,0,1,...\,,N\}$: the smaller  $\varkappa$, the more  terms in the principal part. Recall that a principal part may contain $c^0$-term which in our notation is $\g2_{0,0}^{(1+\varkappa)}$, see Definition \bref{def_prin}. The regular part of $\g2^{(1+\varkappa)}$ is more complicated:
\be
\label{reg_1inf}
\reg \g2^{(1+\varkappa)} = \sum_{m=N+1}^\infty \g2^{(1+\varkappa)}_{-1,m}\,c^{1-m\varkappa} +\sum_{k=0}^{\infty}\sum_{\substack{l=k \\
l\geq 1}}^{\infty}\g2_{k,l}^{(1+\varkappa)}\,c^{-k-l\varkappa} \,.
\ee

Having in mind the form of principal and regular parts \eqref{prin_1inf} and \eqref{reg_1inf} the set of coefficients \eqref{array_1inf0} can be split  as
\be 
\label{array_1inf}
\hspace{-2mm}\bold{P}^{(1+\varkappa)} = \left\{ \g2^{(1+\varkappa)}_{-1,m} ~|~m\in\mathbb{Z}\,,\; m\geq -1\right\}\, \cup\, \left\{ \g2_{k,l}^{(1+\varkappa)}~|~k, l\in\mathbb{Z}^+_0,\; l\geq k \right\}\;,
\ee 
where $\g2_{-1,m}^{(1+\varkappa)}$ and $\g2^{(1+\varkappa)}_{k,l}$ are coefficients at $c^{1-m\varkappa}$ and  $c^{-k-l\varkappa}$, respectively.  Note that the first set contains coefficients from both the principal and  regular parts, the second set contains just one coefficient  from the principal part which is  $\g2_{0,0}^{(1+\varkappa)}$, while other coefficients are from  the regular part. 

For notational  convenience, in this section we assume that 
\be\;
k,l: ~~ k,\, l\in \mathbb{Z}^+_0,\; l\geq k 
\qquad \text{and} \qquad  
n,m: ~~ n+1,\, m+1\in\mathbb{Z}^+_0,\; m\geq n\,.
\ee

\noindent The following proposition comes  directly from the use of the splitting trick.
\begin{prop}
    \label{array_universality_1inf}
    The set of coefficients $\bold{P}^{(1+\varkappa)}$ \eqref{array_1inf0} is independent of $\varkappa$: 
        \be 
            \bold{P}^{(1+\varkappa)} \equiv \bold{P}^{(1+\varkappa')}\,,
            \qquad \forall\varkappa, \varkappa' \in (0, \infty),
        \ee
        or, equivalently, 
\be
\label{B_coin}
\g2^{(1+\varkappa)}_{-1,m} = \g2^{(1+\varkappa')}_{-1,m}\;,
\qquad
\g2^{(1+\varkappa)}_{k,l}  = \g2^{(1+\varkappa')}_{k,l}\;.
\ee
\end{prop}

\begin{corollary} We can change the notation and suppress the superscript $\varkappa$:
\be
\bold{P}^{+} \equiv \bold{P}^{(1+\varkappa)}\;, 
\qquad
\g2_{n,m}^+ \equiv \g2^{(1+\varkappa)}_{n,m}\;.
\ee 
\end{corollary}

\vspace{5mm} 

Let us now represent the Puiseux series    \eqref{main_pui} as 
\be
\label{puis1}
\g2^{(1+\varkappa)} = \sum_{a\in A} \g2^{(1+\varkappa)}_{a}\, c^a\,,
\ee
where the summation domain is defined as
\be
\label{domainA0}
A^{(1+\varkappa)} = \big\{-n-m \varkappa ~|~n+1, m+1\in\mathbb{Z}^+_0,\; m\geq n   \big\}, 
\ee  
and the series coefficients $\g2^{(1+\varkappa)}_{a}$ are elements of $\bold{P}^{+}$ \eqref{array_1inf0} and their linear combinations (see below). Note that we keep  the index $(1+\varkappa)$ here since $\g2_a^{(1+\varkappa)}$ indeed depends on $\varkappa$, contrary to $\g2_{n,m}^+$. For a given $\varkappa$ each pair $(n,m)$ defines a particular power $a$. The summation domain can be represented  similar to \eqref{array_1inf}: 
\be
\label{domainA}
A^{(1+\varkappa)} = \big\{1-m \varkappa ~|~m+1\in \mathbb{Z}^+_0\big\}\cup \big\{-k -l\varkappa\,~| ~k, l\in\mathbb{Z}^+_0,\; l\geq k  \big\}\,, 
\ee  
that reflects the  splitting of $\g2^{(1+\varkappa)}$ into principal and regular parts.  

It turns out that the domain  $A^{(1+\varkappa)}$  contains coinciding  elements provided $\varkappa$ is rational. We will call this  phenomenon {\it a degeneracy}. Here is an example: it might happen that for a given index  $a$ one has $a=1-m\varkappa = -k-l\varkappa$. This last equation can have non-trivial solution relating  $\varkappa$ and $m,k,l$ so that the index $a$  is degenerate.  In this case, the respective contribution  $\g2_a^{(1+\varkappa)} c^a$ to the Puiseux series \eqref{puis1}  is given by a sum 
\be
\label{exB1}
\left(\g2_{-1,m}^{+} + \g2_{k,l}^{+}\right) c^a\,.
\ee
Also, it might happen that $a = -k - l\varkappa = -k' - l'\varkappa$ and then one has 
\be
\label{exB2}
\left(\g2_{k,l}^{+} + \g2_{k',l'}^{+}\right) c^a \,.
\ee

The above discussion shows that the set  $\bold{P}^{+}$ indeed encodes all  expansion coefficients  in the respective Puiseux series. However, a  rule for extracting them from $\bold{P}^{+}$ for a given $\varkappa$ crucially depends on whether $\varkappa$ is rational or not. Each of these cases is examined in detail below.

\paragraph{Irrational case.} Let us consider  irrational $\varkappa\in (0,\infty)\setminus \mathbb{Q}^+$. Then,  there is a one-to-one correspondence between elements of $\bold{P}^{+}$ and the  series  coefficients $\g2_a$ in \eqref{puis1}. This is  because in this case the equation  $a = -n-m\varkappa = -n'-m'\varkappa$ has no solutions except  the trivial one, $n = n', m =m'$. Thus, any element of $A^{(1+\varkappa)}$ \eqref{domainA0} is unique and the degeneracy is lifted. The next proposition is now obvious:
\begin{prop}
\label{array_coincidence_1inf}
For  $\varkappa\in (0,\infty)\setminus \mathbb{Q}^+$ and  $\forall a \in A^{(1+\varkappa)} ~ \exists!~ (m,n): ~ a = -n-m\varkappa $. Hence, there is a unique term $\g2_a^{(1+\varkappa)} c^a$ in the Puiseux series \eqref{puis1} of the form  
\be
\label{irrational_B_1inf}
\g2_{n,m}^{+}\, c^a\,.
\ee
\end{prop}

Due to  Proposition \bref{array_universality_1inf} the  Puiseux series for any two irrational $\varkappa, \varkappa'$ have the same series coefficients, $\g2_{a}^{(1+\varkappa)} = \g2_{a}^{(1+\varkappa')} \equiv \g2_{n,m}^{+}$ for $a = -n-m\varkappa$, but different powers of $c$:
\be
\label{main_pui12}
\ba{l}
\dps
\g2^{(1+\varkappa)} = \sum_{n=-1}^{\infty} \sum_{m\geq n}\g2_{n,m}^+\,c^{-n- m\varkappa}\;,
\\
\vspace{2mm}
\dps
\g2^{(1+\varkappa')} = \sum_{n=-1}^{\infty} \sum_{m\geq n}\g2_{n,m}^+\,c^{-n- m\varkappa'}\;.
\ea 
\ee
However, a separation between principal and regular parts can be different, see fig. \bref{fig:array1}.    

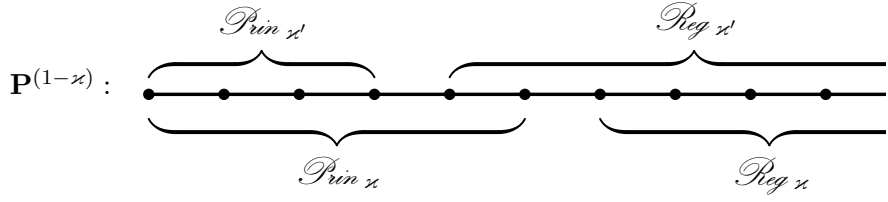
\begin{figure}
\hspace{15mm}\begin{tikzpicture}[thick, line width  = 1.3]
{$\bold{P}^{(1-\varkappa)}: \hspace{0.5 cm}$} \draw (0,0) -- (10,0);
\tkzDefPoint(0, 0){a1}
\tkzDefPoint(1, 0){a2}
\tkzDefPoint(2, 0){a3}
\tkzDefPoint(3, 0){a4}
\tkzDefPoint(4, 0){a5}
\tkzDefPoint(5, 0){a6}
\tkzDefPoint(6, 0){a7}
\tkzDefPoint(7, 0){a8}
\tkzDefPoint(8, 0){a9}
\tkzDefPoint(9, 0){a10}
\draw[decoration={calligraphic brace, raise=15pt, amplitude = 0.35 cm}, decorate, line width  = 1.5] (0, -0.3) -- node[above=25 pt] {$ \pr_{\varkappa'} $} (3, -0.3);
\draw[decoration={calligraphic brace, raise=15pt, amplitude = 0.35 cm}, decorate, line width  = 1.5] (4, -0.3) -- node[above = 25 pt] {$ \reg_{\varkappa'} $} (10.5, -0.3);
\draw[decoration={calligraphic brace, mirror, raise=15pt, amplitude = 0.35 cm},decorate, line width  = 1.5] (0, 0.2) -- node[below=25 pt] {$ \pr_{\varkappa} $} (5, 0.2);
\draw[decoration={calligraphic brace, mirror, raise=15pt, amplitude = 0.35 cm},decorate, line width  = 1.5] (6, 0.2) -- node[below=25 pt] {$ \reg_{\varkappa} $} (10.5, 0.2);
\fill [white] (9.8, 25 pt) rectangle (10.6, -25 pt);
\tkzDrawPoints[color=black,fill=black,size=3](a1, a2, a3, a4, a5, a6, a7, a8, a9, a10)
\end{tikzpicture}
\caption{The irrational case. Elements  $\g2_{-1,m}^{+} \in \bold{P}^{+}$ as nods of the one-dimensional  half-lattice: $m = -1,0,1,...,N(\varkappa'), ...\, ,N(\varkappa)$  enumerate coefficients  in the principal parts.  The coefficients with $m>N$ belong to the regular parts.}
\label{fig:array1}
\end{figure}

\paragraph{Rational case.} Let $\varkappa =p/q \in \QQ$, where $p,q$ are  coprime. Due to the degeneracy, $\exists\, a, a' \in A^{(1+\varkappa)}$ defined by $(n,m) \neq (n',m')$ such that $a = a'$. Let us describe all the cases when the degeneracy takes place. According to \eqref{domainA}  three types of constraints are possible.  
\begin{lemma}
\label{lemma:eqs} 
\label{constraint_solutions_1inf}
Let  $\varkappa = p/q\in \QQ$ with  $p, q$ being  coprime. 
\begin{enumerate}

\item The solution of the equation
\be
\label{coin_cond0_1inf}
1-m \, \frac{p}{q} = 1-m' \, \frac{p}{q}
\ee
is trivial, $m'=m$.

\item The solution of the equation 
\be
\label{coin_cond1_1inf}
1-m \, \frac{p}{q} = -k' - l'\, \frac{p}{q}, 
\ee
is given by  $k' = pt -1$,  $l' = m - qt$, where   $m+2,t\in \mathbb{N}, t \leq \left[ \frac{m+1}{p+q}\right]$. 

\item The solution of the equation
\be 
\label{coin_cond2_1inf}
-k - l\,\frac{p}{q} = -k' - l' \,\frac{p}{q}
\ee
is given by $k' = k - tp $ and $l' = l+tq$, where $k,l \in \mathbb{Z}^+_0, t\in\mathbb{Z},\; l\geq k,\; \left[\frac{k-l}{p+q} \right]\leq t  \leq \left[ \frac{k}{p} \right]$ .
\end{enumerate}
\end{lemma}

In simple words, the Lemma claims  that for a given $m,k,l$ there exist only finitely  many solutions $m',k',l'$. Consequently, only finitely many indices can be the same and the respective term $c^a$ in the Puiseux series has an overall coefficient which is a finite linear combination of elements of $\bold{P}^{(1+\varkappa)}$, cf. \eqref{exB1}--\eqref{exB2}. However, it is  also claimed that there are infinitely many $m,k,l$ which allow non-trivial solution. For example, for all $m\geq p+q$ there is a non-trivial solution of \eqref{coin_cond1_1inf}. Therefore, there are infinitely many degenerate indices in the summation domain $A^{(1+\frac{p}{q})}$ \eqref{domainA}.

To describe the coefficients with (non)degenerate indices $a\in A^{(1+\varkappa)}$  one introduces a set of pairs that defines  $a$: 
\be
\{(n,m)\}_a := \{(n,m)~|~a = -n-m\varkappa\}\,,
\ee 
and formulates the following 
\begin{prop}
\label{general_a_1inf}
Let $\varkappa = p/q $ with $p,q$ being  coprime and   $a\in A^{(1+\varkappa)}$ be an arbitrary index defined by pairs from $\{(n,m)\}_a$.  Then, the respective term in the  Puiseux series is given by 
\be
\label{B_sum_1inf}
\bigg[ \sum\limits_{\{(n,m)\}_a} \g2_{n,m}^{+}\bigg]\,c^a \;.
\ee
\end{prop}

\paragraph{Principal and regular parts.} In what follows we discuss  three statements about principal and regular parts of the Puiseux series: (1) positive degree terms in the principal part are not degenerate for any (ir)rational $\varkappa$; (2) a $c^0$-term can be degenerate for  rational $\varkappa$; (3) an infinite degeneracy in the regular part for rational $\varkappa$.

A convenient tool which allows one to structure the Puiseux series is to divide  the interval $(0,\infty) $ as 
\be
\label{dec_int}
(0,\infty) = \bigcup_{N=0}^{\infty}\Big[\frac{1}{N+1}, \frac{1}{N}\Big)\,,
\qquad
\text{(see fig. \bref{fig:plot})}.
\ee
This representation is instrumental in controlling a number of terms in the principal part. Indeed, for all $\varkappa \in \left[\frac{1}{N+1}, \frac{1}{N}\right)$ we have $N(\varkappa) = const = N$, so that the principal part has exactly $N+3$ terms of the from \eqref{prin_1inf} for all such $\varkappa$'s. Going to the adjacent interval increases a number of terms in the principal part by one.

(1) Let us focus on the principal part without $c^0$-term. Such terms have coefficients $\g2_a^{(1+\varkappa)}$ which are parameterized by positive $a\in A^{(1+\varkappa)}$. The condition $a>0$ holds  only for indices of the form $a = 1 - m \varkappa$ for some $m \geq -1$. In all other cases,  $a\leq 0$. One can see that for a given $\varkappa$ all positive indices $a>0$ are non-degenerate. Indeed, it directly follows from Lemma \bref{constraint_solutions_1inf}:  \eqref{coin_cond0_1inf} has only trivial solutions, while the right-hand side of \eqref{coin_cond1_1inf} is always   non-positive. Then $a = 1-m\varkappa$ for $\varkappa\in \Big[\frac{1}{N+1}, \frac{1}{N}\Big)$ and $-1\leq m \leq N$ is non-degenerate so

\be
\g2_a^{(1+\varkappa)} = \g2_{-1,m}.
\ee

(2) Now, we consider a $c^0$-term in the principal part, $\g2_{0}^{(1+\varkappa)}$. From Lemma \bref{constraint_solutions_1inf} we see that the only non-trivial solution of \eqref{coin_cond1_1inf}--\eqref{coin_cond2_1inf} for $k=l=0$ exists only if $\varkappa = \frac{1}{N}, N\in\mathbb{N}$ and it is a solution of \eqref{coin_cond1_1inf} of the form $m = N$. We call such $\varkappa = \frac{1}{N}$ \textit{exceptional}, while any other $\varkappa \in (0,\infty)\setminus \{\frac{1}{N}, N\in \mathbb{N} \}$  are \textit{regular} (see $\tilde \alpha_n$ in Definition \bref{def:reg}). Thus, $a = 0 $ is degenerate  only for exceptional $\varkappa$'s which are exactly boundaries of intervals in \eqref{dec_int}. The $c^0$-term is given by
\be
\begin{aligned}
    & \g2_0^{(1+\varkappa)} = \g2_{0,0}\,, \quad \text{regular } \varkappa\;; \\
    & \g2_0^{(1+\varkappa)} = \g2_{0,0} + \g2_{-1,N} \,, \quad \text{exceptional } \varkappa = \frac{1}{N}\,.
\end{aligned}
\ee
Note that   $\g2_{-1,N} c^{1-N\varkappa}$ is in the regular part for $\varkappa > \frac{1}{N}$  and in the principal  part for $\varkappa < \frac{1}{N}$ and exactly at $\varkappa = \frac{1}{N}$ the coefficient  $\g2_{-1,N}$ is a part of the $c^0$-term. It  demonstrates  the special role of $c^0$-term: it indicates whether the transition of terms of the form $\g2_{-1,m}$ from the regular part to the principal part happens or not. 

(3) We see that  all indices in the principal part are non-degenerate except $a = 0$ at  $\varkappa = \frac{1}{N}$ and, therefore, the principal part is quite simple to describe. In contrast, the regular part is much more complicated. Here,  for every rational $\varkappa$ there are infinitely many  different degenerate indices. It can be seen from Lemma \bref{constraint_solutions_1inf}. E.g. consider $\varkappa = \frac{p}{q}$ with coprime  $p,q$, then from the constraint \eqref{coin_cond1_1inf} one can see that for every $m \geq p+q-1$ the index $a = 1-m\frac{p}{q}$ is degenerate because there exist $k' = p -1$ and $l' = m-q$ such that $1-m\frac{p}{q} = -k'-l'\frac{p}{q}$. The same degeneracy takes place  for all indices of the form $a = -k-l\frac{p}{q}$, where $k \geq p$. In this case, from the solution of \eqref{coin_cond2_1inf} we see that there exist $k' = k-p$ and $l' = l+q$ such that $-k-l\frac{p}{q} = -k'-l'\frac{p}{q}$.

The above discussion can be summarized as follows.  
\begin{prop}
Let $N \in \mathbb{Z}^+_0$ and $\varkappa \in \big[\frac{1}{N+1}, \frac{1}{N} \big) \subset (0, \infty)$.
\begin{enumerate}

\item The principal part has exactly $N+3$ terms of the form:
\be
\label{prin_g2_1inf}
\ba{c}
\dps
\pr\g2^{(1+\varkappa)} = \g2_{1+\varkappa}^{(1+\varkappa)}\,c^{1+\varkappa}+\g2_{1}^{(1+\varkappa)} \,c+... +\g2_{1-N\varkappa}^{(1+\varkappa)}\,c^{1-N\varkappa}+\g2_0^{(1+\varkappa)} 
\vspace{2mm}
\\
\dps
 = \left(\sum_{m=-1}^{N} \g2_{1-m\varkappa}^{(1+\varkappa)}\, c^{1-m\varkappa}\right) +\g2_0^{(1+\varkappa)}\,.
\ea
\ee

\item Let $m = -1, 0, ... , N$. Then, the respective coefficients in $\pr\g2^{(1+\varkappa)}$ are independent of $\varkappa$, i.e. 
\be
\label{g2_prin_coeff_1inf}
\g2_{1-m\varkappa}^{(1+\varkappa)} = \g2_{-1,m}\,,
\ee
As a consequence, the coefficient  at $c^a, a = 1-m\varkappa$ is the same for all $\varkappa \in \big(0, \frac{1}{m} \big) \supset \big[ \frac{1}{N+1}, \frac{1}{N} \big)$.
     
\item  The $c^0$-term in $\pr\g2^{(1+\varkappa)}$ depends on $\varkappa$ as 
\be
\label{g2_f0_1inf}
\g2_0^{(1+\varkappa)} = \g2_{0,0} + \delta_{\varkappa, 1/N}\g2_{-1,N}\;,
\ee
where the Kronecker symbol is used in the last term. 
\end{enumerate}
\end{prop}
This Proposition about the second coefficient $g_2$ is in fact Proposition \bref{prop2} formulated  in Section  \bref{sec:structure} for the whole logarithmic block. Namely, coming back to $\alpha = 1+\varkappa$, 
\begin{itemize}
    \item equation \eqref{prin_g2_1inf} corresponds to equations \eqref{f_sequence2}--\eqref{lead_exp2}; 
    \item equality of functions in the columns of \eqref{table2} except for the rightmost column, i.e. equations \eqref{columneq2}, except for the last line, follow from  \eqref{g2_prin_coeff_1inf}; 
    \item  equality  of functions in the rightmost column of \eqref{table2}, i.e. the last line in \eqref{columneq2} along with different $f_0$ at exceptional points  \eqref{f0_exceptional2} directly follow from  \eqref{g2_f0_1inf}.
\end{itemize}

\subsubsection{${\alpha\in (0, 1)}$}
\label{app:second_01}

Let us parameterize  real $\alpha \in (0, 1)$ as $\alpha = 1 - \varkappa$, where  $\varkappa \in (0,1)$, and reorganize  the second coefficient  \eqref{f_2ex}  as
\be
\label{g2_01}
\g2^{(1-\varkappa)}  
= \left(\frac{B}{N} c^{1-\varkappa}+ \frac{A}{N} c^{1-2\varkappa} + \frac{D}{N} c^{0} + \frac{C}{N} c^{-\varkappa}\right)\cfrac{1}{\dps 1+\frac{M}{N}c^{-\varkappa} + \frac{K}{N}c^{-1} + \frac{L}{N}c^{-1+\varkappa}}\,.
\ee
Similar to Appendix  \bref{app:second_1inf} one represents the resulting expression  as a double series in variables $x \equiv c$ and $y \equiv c^{-\varkappa}$  treated as independent:
\be
\label{main_pui_01}
\g2^{(1-\varkappa)} = \sum_{n=-1}^{\infty} \sum_{m\geq -n}\g2_{n,m}^{(1-\varkappa)}\, x^{-n} y^{m} \;\equiv\; \sum_{n=-1}^{\infty} \sum_{m\geq -n}\g2_{n,m}^{(1-\varkappa)}\, c^{-n- m\varkappa}\;.
\ee
Here, $\g2_{n,m}^{(1-\varkappa)}$ are some expansion coefficients which are  rational functions of $A,B,C,...\,$ which can be conveniently organized as the following set  
\be 
\label{array_01}
\hspace{-2mm}\bold{P}^{(1-\varkappa)} =
\left\{ \g2_{n,m}^{(1-\varkappa)}~|~n, m\in\mathbb{Z},\; n \geq -1, \; m\geq -n \right\}\;.
\ee 
The principal part of \eqref{g2_01} is given by
\be
\label{prin_01}
\pr\g2^{(1-\varkappa)} =  \sum_{m=1}^{M} \g2^{(1-\varkappa)}_{-1,m}\, c^{1-m\varkappa} + \g2_{0,0}^{(1-\varkappa)}\,,
\end{equation}
with  $M$ being a maximum integer number depending on $\varkappa$ such that  
\be
\label{M_ineq}
M = M(\varkappa): \qquad 1-M \varkappa > 0\,.
\ee 
Thus, in the principal part there are $M+1$ terms  of degrees   $\{0, 1-m\varkappa ~| ~m=1,2,...\,,M\}$. The regular part of $\g2^{(1-\varkappa)}$ is given by 
\be
\label{reg_01}
\reg \g2^{(1-\varkappa)} = \sum_{m=M+1}^\infty \g2^{(1-\varkappa)}_{-1,m}\,c^{1-m\varkappa} +\sum_{l=1}^{\infty}\g2_{0,l}^{(1-\varkappa)}\,c^{-l\varkappa}+ \sum_{k=1}^{\infty}\sum_{\substack{l=-k \\
}}^{\infty}\g2_{k,l}^{(1-\varkappa)}\,c^{-k-l\varkappa} \,.
\ee

\noindent The following proposition comes  directly from using  the splitting trick in \eqref{main_pui_01}:
\begin{prop}
    \label{array_universality_01}
    The set of coefficients $\bold{P}^{(1-\varkappa)}$ \eqref{array_01} is independent of $\varkappa$: 
        \be 
            \bold{P}^{(1-\varkappa)} \equiv \bold{P}^{(1-\varkappa')}\,,
            \qquad \forall\varkappa, \varkappa' \in (0,1),
        \ee
        or, equivalently, 
\be
\label{B_coin_01}
\g2^{(1-\varkappa)}_{n,m} = \g2^{(1-\varkappa')}_{n,m}\;.
\ee
\end{prop}

\begin{corollary} We can change the notation and suppress the superscript $\varkappa$:
\be
\bold{P}^{-} \equiv \bold{P}^{(1-\varkappa)}\;, 
\qquad
\g2_{n,m}^- \equiv \g2^{(1-\varkappa)}_{n,m}\;.
\ee 
\end{corollary}

\vspace{5mm} 

Let us now represent the Puiseux series   \eqref{main_pui_01} as
\be
\label{puis1_01}
\g2^{(1-\varkappa)} = \sum_{a\in A} \g2^{(1-\varkappa)}_{a}\, c^a\,,
\ee
where the summation domain is defined as
\be
\label{domainA0_01}
A^{(1-\varkappa)} = \big\{-n-m \varkappa ~|~n, m\in\mathbb{Z},\; n \geq -1, \; m\geq -n   \big\}, 
\ee  
and the series coefficients $\g2^{(1-\varkappa)}_{a}$ are elements of $\bold{P}^{-}$ \eqref{array_01}. Here, for a given $\varkappa$ each pair $(n,m)$ defines a particular power $a$. The interval $(0,1)\ni \varkappa$ can be divided as 
\be
\label{dec_int}
(0,1) = \bigcup_{M=1}^{\infty}\Big[\frac{1}{M+1}, \frac{1}{M}\Big)\,,
\qquad
\text{(see fig. \bref{fig:plot})}\;.
\ee
For all $\varkappa \in \left[\frac{1}{M+1}, \frac{1}{M}\right)$ we have $M(\varkappa) = const = M$, so that the principal part has exactly $M+1$ terms of the from \eqref{prin_01}.

One can see that the reasoning from the previous Appendix \bref{app:second_1inf} directly applies to the present case and the main conclusions have not changed. Thus, we can immediately collect the key facts in the following proposition.
\begin{prop}
Let $M \in \mathbb{N}$ and $\varkappa \in \big[\frac{1}{M+1}, \frac{1}{M} \big) \subset (0, 1)$.
\begin{enumerate}

\item   The principal part has exactly $M+1$ terms of the form:
\be
\label{prin_g2_01}
\ba{c}
\dps
\pr\g2^{(1-\varkappa)} = \g2_{1-\varkappa}^{(1-\varkappa)}\,c^{1-\varkappa}+... +\g2_{1-M\varkappa}^{(1-\varkappa)}\,c^{1-M\varkappa}+\g2_0^{(1-\varkappa)} 
\vspace{2mm}
\\
\dps
 = \left(\sum_{m=1}^{M} \g2_{1-m\varkappa}^{(1-\varkappa)}\, c^{1-m\varkappa}\right) +\g2_0^{(1-\varkappa)}\,.
\ea
\ee

\item Let $m = 1, 2, ..., \, M$. Then, the respective coefficients in $\pr\g2^{(1-\varkappa)}$ are independent of $\varkappa$, i.e. 
\be
\label{g2_prin_coeff_01}
\g2_{1-m\varkappa}^{(1-\varkappa)} = \g2_{-1,m}^-\,.
\ee
As a consequence, the coefficient  at $c^a, a = 1-m\varkappa$ is the same for all $\varkappa \in \big(0, \frac{1}{m} \big) \supset \big[ \frac{1}{M+1}, \frac{1}{M} \big)$.
     
\item  The $c^0$-term  in $\pr\g2^{(1-\varkappa)}$ depends on $\varkappa$ as  
\be
\label{g2_f0_01}
\g2_0^{(1-\varkappa)} = \g2_{0,0} + \delta_{\varkappa, 1/M}\g2_{-1,M}\;,
\ee
where the Kronecker symbol is used in the last term. 
\end{enumerate}
\end{prop}
This Proposition is in fact the Proposition \bref{prop1} formulated  in Section  \bref{sec:structure} for the whole logarithmic block. Namely, coming back to $\alpha = 1-\varkappa$: 
\begin{itemize}
    \item equation \eqref{prin_g2_01} corresponds to equations \eqref{f_sequence1}--\eqref{lead_exp1}; 
    \item equality of functions in the columns of \eqref{table1} except for the rightmost column, i.e. equations \eqref{columneq1}, except for the last line, follow from  \eqref{g2_prin_coeff_01}; 
    \item  equality  of functions in the rightmost column of \eqref{table1}, i.e. the last line in \eqref{columneq1} along with different $f_0$ at exceptional points  \eqref{f0_exceptional1} directly follow from  \eqref{g2_f0_01}.
\end{itemize}

\subsubsection{$\alpha = 1,0$}
\label{app:alpha01}

\paragraph{${\bm {\alpha=1}}$.} The second coefficient  \eqref{f_2ex} of the logarithmic conformal block is drastically simplified:
\be
\label{a1}
\begin{aligned}
& \g2^{(1)}  = \left(\dfrac{A+B}{M+N}c + \dfrac{C+D}{M+N}\right) \dfrac{1}{1 + \frac{K+L}{M+N}c^{-1}}=  \sum_{k=-1}^{\infty} \g2^{(1)}_k \frac{1}{c^k}\;. \\
\end{aligned}
\ee
Thus, 
\be 
\ba{l}
\dps
\pr \g2^{(1)} = \dfrac{A+B}{M+N}c + \left(\dfrac{C+D}{M+N} - \dfrac{(A+B)(K+L)}{(M+N)^2} \right)\;, 
\vspace{2mm}
\\
\dps
\reg \g2^{(1)} =  \sum_{k=1}^{\infty} (-1)^k\left(\dfrac{K+L}{M+N}\right)^k \left(\dfrac{C+D}{M+N} - \dfrac{(A+B)(K+L)}{(M+N)^2} \right)\dfrac{1}{c^k}\;.
\ea
\ee
The leading $c^1$-term in $\pr \g2^{(1)}$ is  the Zamolodchikov's classical conformal block \cite{Zamolodchikov1987}. 

\paragraph{${\bm \alpha=0}$.} In this case, the second coefficient is even simpler, 
\be
\label{a0}
\begin{aligned}
 & \g2^{(0)} = \left(\dfrac{B+D}{N+L} + \dfrac{A+C}{N+L}\dfrac{1}{c}\right) \dfrac{1}{1 + \frac{M+K}{N+L}c^{-1}} = \sum_{k=0}^{\infty} \g2^{(0)}_k \frac{1}{c^k} = \\
& = \dfrac{B+D}{N+L} + \sum_{k=1}^{\infty} (-1)^{k-1}\left(\dfrac{M+K}{N+L}\right)^{k-1} \left(\dfrac{A+C}{N+L} - \dfrac{(B+D)(M+K)}{(N+L)^2} \right)\dfrac{1}{c^k}.
\end{aligned}
\ee
The principal part here contains the only term $\sim c^0$. Taking the limit $c=\infty$  yields the second coefficient of the global conformal block \cite{Ferrara:1974ny}.

\subsection{Newton polygons and the large-$c$ expansions} 
\label{app:forth}

In this Appendix we formulate the Newton polygon technique   which helps to obtain the Puiseux series for the $\alpha$-heavy  conformal block coefficients at large-$c$ (it is equally applicable in the $\cW_N$ case, see Appendix \bref{app:W3}). To this end, we explicitly analyze the forth coefficient of the $\alpha$-heavy logarithmic  Virasoro conformal block and formulate the general rules valid for any higher-order coefficients. 

The fourth coefficient in \eqref{f_exp_alpha} has the form $g_4 ^{(\alpha)}=\mathcal{P}^{(\alpha)}/\mathcal{Q}^{(\alpha)}$:
\be
\label{g4_num}
\mathcal{P}^{(\alpha)} = \sum_{m=1}^{3} \sum_{k=3-m}^{4} A_{k,m}c^{k+m\alpha} + \sum_{m=4}^{8}\sum_{k=0}^{8-m} A_{k,m}c^{k+m\alpha} 
\ee
 and
\be
\label{g4_den}
\mathcal{Q}^{(\alpha)} = \sum_{m=0}^{2} \sum_{k=2-m}^{4} a_{k,m}c^{k+m\alpha} + \sum_{m=3}^{7}\sum_{k=0}^{7-m} a_{k,m}c^{k+m\alpha}\,,
\ee
where $A_{k,m} = A_{k,m}(\delta, \tdelta)$ and $a_{k,m} = a_{k,m}(\delta, \tdelta)$ are particular functions which can be explicitly found by using {\it Mathematica}. The fifth coefficient in \eqref{f_exp_alpha} has   the same $c$-dependence except for concrete  $A_{k,m}$ and $a_{k,m}$. We are interested in the large-$c$ expansion of $g_4^{(\alpha)}$ at different $\alpha$. To this end, we show that the Puiseux series can be read off directly from the Newton polygons on fig. \bref{fig-num-g4} and \bref{fig-den-g4}.

\begin{figure}[h!]
 \begin{minipage}{0.48\textwidth}
     \centering
     \begin{tikzpicture}[scale = 0.5]
        \draw[thin,dotted] (0,0) grid (4,8);
        \draw[->] (0,0) -- (4.4,0) node[right] {$k$};
        \draw[->] (0,0) -- (0,8.4) node[above] {$m$};
        \foreach \x/\xlabel in { 1/1, 2/2, 3/3, 4/4}
    \draw (\x cm,1pt ) -- (\x cm,-1pt ) node[anchor=north,fill=white] {\xlabel};
  \foreach \y/\ylabel in {1/1, 2/2, 3/3, 4/4, 5/5, 6/6, 7/7, 8/8}
    \draw (1pt,\y cm) -- (-1pt ,\y cm) node[anchor=east, fill=white] {\ylabel};
    \draw[thick, red] (0,8) -- (4,4);
    \draw[thick, red] (0,8) -- (0,3);
    \draw[thick, red] (0,3) -- (2,1);
    \draw[thick, red] (2,1) -- (4,1);
    \draw[thick, red] (4,1) -- (4,4);
    \draw[fill=black] (2,1) circle (0.12) node[above right] {};
    \draw[fill=black] (3,1) circle (0.12) node[above right] {};
    \draw[fill=black] (4,1) circle (0.12) node[above right] {};
    \draw[fill=black] (1,2) circle (0.12) node[above right] {};
    \draw[fill=black] (2,2) circle (0.12) node[above right] {};
    \draw[fill=black] (3,2) circle (0.12) node[above right] {};
    \draw[fill=black] (4,2) circle (0.12) node[above right] {};
    \draw[fill=black] (0,3) circle (0.12) node[above right] {};
    \draw[fill=black] (1,3) circle (0.12) node[above right] {};
    \draw[fill=black] (2,3) circle (0.12) node[above right] {};
    \draw[fill=black] (3,3) circle (0.12) node[above right] {};
    \draw[fill=black] (4,3) circle (0.12) node[above right] {};
    \draw[fill=black] (0,4) circle (0.12) node[above right] {};
    \draw[fill=black] (1,4) circle (0.12) node[above right] {};
    \draw[fill=black] (2,4) circle (0.12) node[above right] {};
    \draw[fill=black] (3,4) circle (0.12) node[above right] {};
    \draw[fill=black] (4,4) circle (0.12) node[above right] {};
    \draw[fill=black] (0,5) circle (0.12) node[above right] {};
    \draw[fill=black] (1,5) circle (0.12) node[above right] {};
    \draw[fill=black] (2,5) circle (0.12) node[above right] {};
    \draw[fill=black] (3,5) circle (0.12) node[above right] {};
    \draw[fill=black] (0,6) circle (0.12) node[above right] {};
    \draw[fill=black] (1,6) circle (0.12) node[above right] {};
    \draw[fill=black] (2,6) circle (0.12) node[above right] {};
    \draw[fill=black] (0,7) circle (0.12) node[above right] {};
    \draw[fill=black] (1,7) circle (0.12) node[above right] {};
    \draw[fill=black] (0,8) circle (0.12) node[above right] {};
    \end{tikzpicture}
    \caption{$\mathcal{P}^{(\alpha)}$-polygon \eqref{g4_num}.}
    \label{fig-num-g4}
 \end{minipage}\hfill
  \begin{minipage}{0.48\textwidth}
     \centering
     \begin{tikzpicture}[scale = 0.5]
          \draw[thin,dotted] (0,0) grid (4,8);
        \draw[->] (0,0) -- (4.4,0) node[right] {$k$};
        \draw[->] (0,0) -- (0,8.4) node[above] {$m$};
        \foreach \x/\xlabel in { 1/1, 2/2, 3/3, 4/4}
    \draw (\x cm,1pt ) -- (\x cm,-1pt ) node[anchor=north,fill=white] {\xlabel};
  \foreach \y/\ylabel in {1/1, 2/2, 3/3, 4/4, 5/5, 6/6, 7/7, 8/8}
    \draw (1pt,\y cm) -- (-1pt ,\y cm) node[anchor=east, fill=white] {\ylabel};
    \draw[thick, red] (0,7) -- (4,3);
    \draw[thick, red] (0,7) -- (0,2);
    \draw[thick, red] (0,2) -- (2,0);
    \draw[thick, red] (2,0) -- (4,0);
    \draw[thick, red] (4,0) -- (4,3);
    \draw[fill=black] (2,0) circle (0.12) node[above right] {};
    \draw[fill=black] (3,0) circle (0.12) node[above right] {};
    \draw[fill=black] (4,0) circle (0.12) node[above right] {};
    \draw[fill=black] (1,1) circle (0.12) node[above right] {};
    \draw[fill=black] (2,1) circle (0.12) node[above right] {};
    \draw[fill=black] (3,1) circle (0.12) node[above right] {};
    \draw[fill=black] (4,1) circle (0.12) node[above right] {};
    \draw[fill=black] (0,2) circle (0.12) node[above right] {};
    \draw[fill=black] (1,2) circle (0.12) node[above right] {};
    \draw[fill=black] (2,2) circle (0.12) node[above right] {};
    \draw[fill=black] (3,2) circle (0.12) node[above right] {};
    \draw[fill=black] (4,2) circle (0.12) node[above right] {};
    \draw[fill=black] (0,3) circle (0.12) node[above right] {};
    \draw[fill=black] (1,3) circle (0.12) node[above right] {};
    \draw[fill=black] (2,3) circle (0.12) node[above right] {};
    \draw[fill=black] (3,3) circle (0.12) node[above right] {};
    \draw[fill=black] (4,3) circle (0.12) node[above right] {};
    \draw[fill=black] (0,4) circle (0.12) node[above right] {};
    \draw[fill=black] (1,4) circle (0.12) node[above right] {};
    \draw[fill=black] (2,4) circle (0.12) node[above right] {};
    \draw[fill=black] (3,4) circle (0.12) node[above right] {};
    \draw[fill=black] (0,5) circle (0.12) node[above right] {};
    \draw[fill=black] (1,5) circle (0.12) node[above right] {};
    \draw[fill=black] (2,5) circle (0.12) node[above right] {};
    \draw[fill=black] (0,6) circle (0.12) node[above right] {};
    \draw[fill=black] (1,6) circle (0.12) node[above right] {};
    \draw[fill=black] (0,7) circle (0.12) node[above right] {};
    \end{tikzpicture}
    \caption{$\mathcal{Q}^{(\alpha)}$-polygon \eqref{g4_den}.}
    \label{fig-den-g4}
  \end{minipage}
\end{figure}

First of all, let us discuss the  general case. Let  $g_n^{(\alpha)} = P^{(\alpha)}/Q^{(\alpha)}$, $n\in \NN$, $\alpha \in \RR^+_0$; the respective Newton polygons are referred to as  $P^{(\alpha)}$-polygons and $Q^{(\alpha)}$-polygons. The large-$c$ expansion of $g_n^{(\alpha)}$ is the Puiseux series obtained through  the following three steps:

\begin{enumerate}
\item[{\bf I.}] Isolate the highest order term $a_{k,m}c^{k+m\alpha} $ in the denominator  and redefine  $g_n^{(\alpha)} = \tilde{P}^{(\alpha)}/\tilde{Q}^{(\alpha)} = P^{(\alpha)}/Q^{(\alpha)}$,  where $\tilde{P}^{(\alpha)}= P^{(\alpha)}/(a_{k,m}c^{k+m\alpha})$ and $\tilde{Q}^{(\alpha)} = Q^{(\alpha)}/(a_{k,m}c^{k+m\alpha})$. Then, the denominator $\tilde{Q}^{(\alpha)}$ has the  form $\tilde{Q}^{(\alpha)} =1+ X(c)$, where $X(c)$ contains negative power terms only; 

\item[{\bf II.}] Expand $1/\tilde{Q}^{(\alpha)}\equiv 1/(1+ X(c))$ at $X(c) \to 0 $  (which is equivalent to sending $c\to \infty$);
    
\item[{\bf III.}] Multiply the resulting expansion by $\tilde{P}^{(\alpha)}$ and rearrange terms in order to obtain the Puiseux series. 
\end{enumerate}

\noindent Note that the exponentiation hypothesis is equivalent to that the highest order term in $\tilde{P}$ is $c^\alpha$ for all $g_n^{(\alpha)}$ and for all $\alpha\in[0,\infty)$. It turns out that every block coefficient considered in this paper has the {\it same Puiseux expansion}, i.e. powers of $c$ are the same but expansion coefficients are different. Below we describe specific properties of  $P^{(\alpha)}$-polygons and $Q^{(\alpha)}$-polygons  which  are sufficient to have the large-$c$ expansions of the same form. It turns out that these properties are  related to the upper edge and some of its neighbouring points only  while the rest of the  Newton polygon is  irrelevant for finding the form of the respective Puiseux series.

\begin{itemize}

\item[{\bf (a)}] $\exists L, R, D, U \in \mathbb{Z}$ such that $R-L \leq U -D$ and all points $(k,m)$ in a given Newton  polygon necessarily satisfy  inequalities $L \leq k \leq R$ and $D \leq m \leq U$.  One denotes $L, R$ for $P^{(\alpha)}$ and $Q^{(\alpha)}$ as $L_{P^{(\alpha)}}, R_{P^{(\alpha)}}$ and $L_{Q^{(\alpha)}}, R_{Q^{(\alpha)}}$, respectively, and $D, U$ for $P^{(\alpha)}$ and $Q^{(\alpha)}$ as $D_{P^{(\alpha)}}, U_{P^{(\alpha)}}$ and $D_{Q^{(\alpha)}}, U_{Q^{(\alpha)}}$, respectively.

\item[{\bf (b)}] All points  in a given  Newton polygon necessarily satisfy  inequality $m + k \leq L + U $. There are points which satisfy $m = -k + L + U$ for all $k \in \{L,L+1,...,R\}$. The set of these points forms the \textit{upper edge} of a polygon.

\item[{\bf (c)}] The following relations hold:  $ U_{P^{(\alpha)}} = U_{Q^{(\alpha)}}+ 1$; \, $L_{P^{(\alpha)}}  = L_{Q^{(\alpha)}}$;\,  $R_{P^{(\alpha)}}  = R_{Q^{(\alpha)}}$.

\item[{\bf (d)}]  There exists a point $\in Q^{(\alpha)}$-polygon right  below the leftmost point of the upper edge, i.e. the point $(L_{Q^{(\alpha)}}, U_{Q^{(\alpha)}}+L_{Q^{(\alpha)}}-1)$.

\item[{\bf (e)}] There exists a point $\in Q^{(\alpha)}$-polygon right below the rightmost point of the upper edge, i.e. the point $(R_{Q^{(\alpha)}}, U_{Q^{(\alpha)}}+L_{Q^{(\alpha)}}-R_{Q^{(\alpha)}}-1)$.
\end{itemize}

\noindent One can see that for $n=4$ the respective Newton polygons  satisfy all these properties. From fig. \bref{fig-num-g4} and \bref{fig-den-g4}  one finds 
\be
\ba{c}
\dps
L_{\mathcal{P}^{(\alpha)}}=0\,, 
\quad
R_{\mathcal{P}^{(\alpha)}}=4\, 
\quad
D_{\mathcal{P}^{(\alpha)}}=1\,,
\quad
U_{\mathcal{P}^{(\alpha)}} = 8\,; 
\vspace{2mm}
\\
\dps
L_{\mathcal{P}^{(\alpha)}}=0\,,
\quad 
R_{\mathcal{P}^{(\alpha)}}=4\,,  
\quad
D_{\mathcal{P}^{(\alpha)}}=0\,, 
\quad
U_{\mathcal{P}^{(\alpha)}} = 7\,. 
\ea
\ee

In what follows  we demonstrate that these properties are sufficient to obtain the Puiseux expansion of the forth coefficient. 
 
\subsubsection{${\alpha = 1+\varkappa\, \in (1, \infty)}$} 

Let consider the Puiseux series for the coefficient $g_4^{(\alpha)}$ \eqref{g4_num}-\eqref{g4_den} at $\alpha \in (1, \infty)$ by applying the procedure described in  {\bf I. -- III.}  and {\bf (a) -- (e)} from Section \bref{app:forth}. 

{\bf I.} Consider a term of maximal power in the denominator $\mathcal{Q}^{(\alpha)}$.  To this end, rearranging ${k+m\alpha} \equiv {(k+m)+m\varkappa}$  we are looking for a term of maximal power $(k+m) + m\varkappa$. Equivalently, we need a point $(k,m)$ on the respective Newton polygon which has maximal values of $k+m$ and $m$ provided $\varkappa > 0$. According to item {\bf (b)}, the maximal value $k+m$ can be  reached on the upper edge, $m = -k +U_{\cQ^{(\alpha)}} + L_{\cQ^{(\alpha)}} $. According to item {\bf (a)}, the maximal value of $m$ is equal to $U_{\cQ^{(\alpha)}}$ which is reached when $k = L_{\cQ^{(\alpha)}} = 0$ since $k\geq L_{\cQ^{(\alpha)}}$. Thus, this is the leftmost point of the upper edge $(L_{\cQ^{(\alpha)}}, U_{\cQ^{(\alpha)}}+L_{\cQ^{(\alpha)}}) = (0,7)$ which corresponds to the term $a_{0,7} c^{7+7\varkappa}$.  

    \begin{figure}[h!]
 \begin{minipage}{0.48\textwidth}
     \centering
     \begin{tikzpicture}[scale = 0.5]
        \draw[thin,dotted] (0,-7) grid (4,1);
        \draw[->] (0,0) -- (4.4,0) node[right] {$k$};
        \draw[->] (0,-7) -- (0,1.4) node[above] {$m$};
        \foreach \x/\xlabel in { 1/1, 2/2, 3/3, 4/4}
    \draw (\x cm,1pt ) -- (\x cm,-1pt ) node[anchor=north,fill=white] {\xlabel};
  \foreach \y/\ylabel in {-7/{-7\hphantom{-}},-6/{-6\hphantom{-}}, -5/{-5\hphantom{-}}, -4/{-4\hphantom{-}}, -3/{-3\hphantom{-}}, -2/{-2\hphantom{-}}, -1/{-1\hphantom{-}}, 1/1}
    \draw (1pt,\y cm) -- (-1pt ,\y cm) node[anchor=east, fill=white] {\ylabel};
    \draw[thick, red] (0,1) -- (4,-3);
    \draw[thick, red] (0,1) -- (0,-4);
    \draw[thick, red] (0,-4) -- (2,-6);
    \draw[thick, red] (2,-6) -- (4,-6);
    \draw[thick, red] (4,-6) -- (4,-3);
    \draw[fill=black] (2,-6) circle (0.12) node[above right] {};
    \draw[fill=black] (3,-6) circle (0.12) node[above right] {};
    \draw[fill=black] (4,-6) circle (0.12) node[above right] {};
    \draw[fill=black] (1,-5) circle (0.12) node[above right] {};
    \draw[fill=black] (2,-5) circle (0.12) node[above right] {};
    \draw[fill=black] (3,-5) circle (0.12) node[above right] {};
    \draw[fill=black] (4,-5) circle (0.12) node[above right] {};
    \draw[fill=black] (0,-4) circle (0.12) node[above right] {};
    \draw[fill=black] (1,-4) circle (0.12) node[above right] {};
    \draw[fill=black] (2,-4) circle (0.12) node[above right] {};
    \draw[fill=black] (3,-4) circle (0.12) node[above right] {};
    \draw[fill=black] (4,-4) circle (0.12) node[above right] {};
    \draw[fill=black] (0,-3) circle (0.12) node[above right] {};
    \draw[fill=black] (1,-3) circle (0.12) node[above right] {};
    \draw[fill=black] (2,-3) circle (0.12) node[above right] {};
    \draw[fill=black] (3,-3) circle (0.12) node[above right] {};
    \draw[fill=black] (4,-3) circle (0.12) node[above right] {};
    \draw[fill=black] (0,-2) circle (0.12) node[above right] {};
    \draw[fill=black] (1,-2) circle (0.12) node[above right] {};
    \draw[fill=black] (2,-2) circle (0.12) node[above right] {};
    \draw[fill=black] (3,-2) circle (0.12) node[above right] {};
    \draw[fill=black] (0,-1) circle (0.12) node[above right] {};
    \draw[fill=black] (1,-1) circle (0.12) node[above right] {};
    \draw[fill=black] (2,-1) circle (0.12) node[above right] {};
    \draw[fill=black] (0,0) circle (0.12) node[above right] {};
    \draw[fill=black] (1,0) circle (0.12) node[above right] {};
    \draw[fill=black] (0,1) circle (0.12) node[above right] {};
    \end{tikzpicture}
    \caption{ $\tilde{\mathcal{P}}^{(\alpha)}$-polygon, $\alpha\in (1, \infty)$.}
    \label{fig-g4-p-tilde}
 \end{minipage}\hfill
  \begin{minipage}{0.48\textwidth}
     \centering
     \begin{tikzpicture}[scale = 0.5]
       \draw[thin,dotted] (0,-7) grid (4,1);
        \draw[->] (0,0) -- (4.4,0) node[right] {$k$};
        \draw[->] (0,-7) -- (0,1.4) node[above] {$m$};
        \foreach \x/\xlabel in { 1/1, 2/2, 3/3, 4/4}
    \draw (\x cm,1pt ) -- (\x cm,-1pt ) node[anchor=north,fill=white] {\xlabel};
  \foreach \y/\ylabel in {-7/{-7\hphantom{-}}, -6/{-6\hphantom{-}}, -5/{-5\hphantom{-}}, -4/{-4\hphantom{-}}, -3/{-3\hphantom{-}}, -2/{-2\hphantom{-}}, -1/{-1\hphantom{-}}, 1/1}
    \draw (1pt,\y cm) -- (-1pt ,\y cm) node[anchor=east, fill=white] {\ylabel};
    \draw[thick, red] (0,0) -- (4,-4);
    \draw[thick, red] (0,0) -- (0,-5);
    \draw[thick, red] (0,-5) -- (2,-7);
    \draw[thick, red] (2,-7) -- (4,-7);
    \draw[thick, red] (4,-7) -- (4,-4);
   
    \draw[fill=black] (2,-7) circle (0.12) node[above right] {};
    \draw[fill=black] (3,-7) circle (0.12) node[above right] {};
    \draw[fill=black] (4,-7) circle (0.12) node[above right] {};
    \draw[fill=black] (1,-6) circle (0.12) node[above right] {};
    \draw[fill=black] (2,-6) circle (0.12) node[above right] {};
    \draw[fill=black] (3,-6) circle (0.12) node[above right] {};
    \draw[fill=black] (4,-6) circle (0.12) node[above right] {};
    \draw[fill=black] (0,-5) circle (0.12) node[above right] {};
    \draw[fill=black] (1,-5) circle (0.12) node[above right] {};
    \draw[fill=black] (2,-5) circle (0.12) node[above right] {};
    \draw[fill=black] (3,-5) circle (0.12) node[above right] {};
    \draw[fill=black] (4,-5) circle (0.12) node[above right] {};
    \draw[fill=black] (0,-4) circle (0.12) node[above right] {};
    \draw[fill=black] (1,-4) circle (0.12) node[above right] {};
    \draw[fill=black] (2,-4) circle (0.12) node[above right] {};
    \draw[fill=black] (3,-4) circle (0.12) node[above right] {};
    \draw[fill=black] (4,-4) circle (0.12) node[above right] {};
    \draw[fill=black] (0,-3) circle (0.12) node[above right] {};
    \draw[fill=black] (1,-3) circle (0.12) node[above right] {};
    \draw[fill=black] (2,-3) circle (0.12) node[above right] {};
    \draw[fill=black] (3,-3) circle (0.12) node[above right] {};
    \draw[fill=black] (0,-2) circle (0.12) node[above right] {};
    \draw[fill=black] (1,-2) circle (0.12) node[above right] {};
    \draw[fill=black] (2,-2) circle (0.12) node[above right] {};
    \draw[fill=green] (0,-1) circle (0.15) node[above right] {};
    \draw[fill=green] (1,-1) circle (0.15) node[above right] {};
    \draw[fill=black] (0,0) circle (0.12) node[above right] {};
    \end{tikzpicture}
    \caption{ $\tilde{\mathcal{Q}}^{(\alpha)}$-polygon, $\alpha\in (1, \infty)$.  }
    \label{fig-g4-q-tilde}
  \end{minipage}
\end{figure}

The Newton polygons associated to $\tilde{\mathcal{P}}$ and $\tilde{\mathcal{Q}}$ are shifted down along the $m$-axis as shown on  fig. \bref{fig-g4-p-tilde},  \bref{fig-g4-q-tilde}. In particular, we find that  $\tilde{\mathcal{Q}}^{(\alpha)} = 1 + X(c)$, where $X(c)$ contains only negative powers of $c$ and   $X(c) \to 0$ at large $c$ (i.e. $1$ is the point $(0,0)$ and terms in $X(c)$ lie below the $k$-axis, see fig. \bref{fig-g4-q-tilde}). Note that now due to shifting down by $U_{\mathcal{Q}^{(\alpha)}}+L_{\mathcal{Q}^{(\alpha)}} = 7$ we obtain: for the $\tilde{\cQ}^{(\alpha)}$-polygon $L_{\tilde{\mathcal{Q}}^{(\alpha)}} = 0$, $R_{\tilde{\mathcal{Q}}^{(\alpha)}} = 4$, $D_{\tilde{\mathcal{Q}}^{(\alpha)}} = - 7$, $U_{\tilde{\mathcal{Q}}^{(\alpha)}} = 0$  and since item {\bf (c)} holds, for  the $\tilde{\cP}^{(\alpha)}$-polygon $L_{\tilde{\mathcal{P}}^{(\alpha)}} = 0$, $R_{\tilde{\mathcal{P}}^{(\alpha)}} = 4$, $D_{\tilde{\mathcal{P}}^{(\alpha)}} = - 6$, $U_{\tilde{\mathcal{P}}^{(\alpha)}} = 1$.

{\bf II.} Let us consider the large-$c$ expansion of the denominator. 
\begin{lemma} The Puiseux series for $\alpha\in (1, \infty)$ is given by 
\be
\label{lemmaA8}
\frac{1}{\tilde{\mathcal{Q}}^{(\alpha)}} \equiv \frac{1}{1+X(c)} = \sum_{n=0}^{\infty}\sum_{m=n}^{\infty}(\tilde{\mathcal{Q}}^{(\alpha)})_{n,m} c^{-n-m\varkappa}\,.
\ee
\end{lemma}
\begin{proof} The $X(c)$ has the following terms: there are terms  $c^{-\varkappa}$ and $c^{-1-\varkappa}$ (they are marked by green dots on fig. \bref{fig-g4-q-tilde}) while others are  $c^{-r-s\varkappa}$ with $ 0 \leq r \leq  s$ and $2 \leq s \leq 7$ (they are below the $m=-1$ line). Their presence in $X(c)$ can be deduced  from the properties {\bf (a)},  {\bf (b)}, and {\bf (d)}. Indeed, from {\bf (b)} we see that there is a point on the upper edge such that $k+m = U_{\tilde{\mathcal{Q}}^{(\alpha)}} + L_{\tilde{\mathcal{Q}}^{(\alpha)}} = 0$ and $k = 1$, so $m = -1$, therefore, the corresponding term is $c^{1 - \alpha} \equiv c^{-\varkappa}$.  Then, from {\bf (d)} we see that there  is a term with coordinates $(L_{\tilde{\mathcal{Q}}^{(\alpha)}}, U_{\tilde{\mathcal{Q}}^{(\alpha)}}-1) = (0,-1)$, therefore, it is  $c^{-\alpha} \equiv c^{-1-\varkappa}$. All other terms have $D_{\tilde{\mathcal{Q}}^{(\alpha)}} \leq m \leq U_{\tilde{\mathcal{Q}}^{(\alpha)}}-2$, i.e. $-7\leq m \leq -2$ since a term with $m=U_{\tilde{\mathcal{Q}}^{(\alpha)}} = 0$ is $1$ and all terms with $m=U_{\tilde{\mathcal{Q}}^{(\alpha)}}-1 = -1$ are exactly $c^{-\varkappa}$ and $c^{-1-\varkappa}$.  Also, from {\bf (a)} and {\bf (b)} it follows that other terms in $\tilde{\mathcal{Q}}^{(\alpha)}$ are  $c^{(k+m)+m\varkappa}$, where $k,m$ necessarily satisfy the two  inequalities $ k+m\leq 0$ and $k\geq 0$ which both can be packed into  a single inequality $ m \leq k+m \leq 0$.  These terms can equivalently be  represented as $c^{-r-s\varkappa}$, where $ 0\leq r \leq s$ and $2 \leq s \leq 7$.  

Now, one expands the denominator, 
\be
\label{X-expansion}
\frac{1}{1+X(c)} = \sum_{m=0}^{\infty}(-1)^m X^m(c)\,.
\ee
Firstly, let us show that every term in $X^m(c)$ has the form $c^{-N-M\varkappa}$ for some $M,N: 0 \leq N \leq M$.  The multinomial expansion of $X^m(c)$ consists of the following terms: 
\be
\label{term_multinom}
\prod_{i=1}^{m} c^{-r_i - s_i \varkappa} = c^{-\sum_{i=1}^m r_i - \varkappa\sum_{i=1}^m s_i}
\ee
for $0\leq r_i \leq s_i$ and $2\leq s_i \leq 7$.  Since these inequalities hold for all $i$, then  $  \sum_{i=1}^m r_i \leq \sum_{i=1}^m s_i$. Denoting $ \sum_{i=1}^m r_i = N$ and $\sum_{i=1}^m s_i = M $ one finds  that the right-hand side of \eqref{term_multinom} is exactly $c^{-N-M\varkappa}$ for $0 \leq N \leq M$.  Secondly, for every $M \geq 0$ the expression  $X(c)^M \sim (c^{-1+\varkappa}+ ... + c^{-\varkappa} )^M$  contains the following monomials:  
\be
c^{N(-1+\varkappa)}c^{(M-N)(-\varkappa)} = c^{-N-M\varkappa} \quad  \text{for all}\quad  N = 0,1,...,M\,.
\ee 
Therefore, expanding  and rearranging the whole expression  $X^m(c)$ in \eqref{X-expansion} we find out  that all possible degrees of $c$ in \eqref{X-expansion} are generated by repeatedly multiplying $c^{-\varkappa}$ and $c^{-1-\varkappa}$ that results in $c^{-n-m\varkappa}$ for all  $n,m: 0\leq n \leq m$. In this way, one obtains the Puiseux series \eqref{lemmaA8}.
\end{proof}

{\bf III.} Finally, consider a product of  $1/\tilde{\mathcal{Q}}^{(\alpha)}$  and $\tilde{\mathcal{P}}^{(\alpha)}$. First of all, following item  {\bf I.} one can show that due to the properties {\bf (a)} and {\bf (b)} a point with maximal $(k+m)+m\varkappa, ~\forall \varkappa >0$ in the $\tilde{\mathcal{P}}^{(\alpha)}$-polygon  has coordinates  $(L_{\tilde{\mathcal{P}}^{(\alpha)}}, U_{\tilde{\mathcal{P}}^{(\alpha)}}) = (L_{\mathcal{P}^{(\alpha)}} - L_{\mathcal{Q}^{(\alpha)}},U_{\mathcal{P}^{(\alpha)}}- U_{\mathcal{Q}^{(\alpha)}})$ which equals $(0,1)$ due to {\bf (c)}. It corresponds to the term $c^{\alpha} \equiv c^{1+\varkappa}$, which, therefore, is the highest-order term in the numerator.  Subleading terms in the numerator have the form $c^{-r-s\varkappa}$ for $-1 \leq r \leq s$ and $0 \leq s \leq 6$. The reasoning is the same as for $\tilde{\cQ}^{(\alpha)}$:  every subleading term has a form $c^{(k+m) + m\varkappa}$ for $k+m\leq L_{\tilde{\mathcal{P}}^{(\alpha)}}+ U_{\tilde{\mathcal{P}}^{(\alpha)}} = 1$, $k\geq L_{\tilde{\mathcal{P}}^{(\alpha)}} = 0$ and $D_{\tilde{\mathcal{P}}^{(\alpha)}} \leq m\leq U_{\tilde{\mathcal{P}}^{(\alpha)}} -1$, i.e. $-6 \leq m\leq -1$  which by  denoting $k+m = -r$, $m = -s$ can be cast into the form $c^{-r-s\varkappa}$.  Using the expansion \eqref{lemmaA8} one finds   
\be
g_{4}^{(1+\varkappa)} = \tilde{\mathcal{P}}^{(\alpha)} \frac{1}{\tilde{\mathcal{Q}}^{(\alpha)}} \sim (c^{1+\varkappa}+...)\sum_{n=0}^{\infty}\sum_{m=n}^{\infty} (\tilde{\mathcal{Q}}^{(\alpha)})_{n,m} c^{-n-m\varkappa}\,,
\ee
where the ellipsis in $\tilde{\mathcal{P}}^{(\alpha)}$ denotes the subleading terms.   Multiplying $c^{1+\varkappa}$ by $c^{-n-m\varkappa}$ for $0\leq n \leq m$ one obtains $c^{-(n-1)-(m-1)\varkappa}$. Multiplying subleading term $c^{-r-s\varkappa}$ by $c^{-n-m\varkappa}$ one obtains $c^{-(n+r) - (m+s)\varkappa}$ so $-1 \leq n+r \leq m+s$ and $m+s \geq 0$. This shows that  multiplying $\tilde{\cP}^{(\alpha)}$ by the expansion \eqref{lemmaA8} and rearranging terms change only limits of the sum and the explicit form of the coefficients of the decomposition, while the general structure is not changed: it still consists of terms $c^{-n-m\varkappa}$ for $n\leq m$ but now $ n \geq -1$. Thus, the Puiseux series  $g_4^{(1+\varkappa)}$ can be cast into the form
\be
\label{g4_P1}
g_{4}^{(1+\varkappa)} = \sum_{n=-1}^{\infty}\sum_{m=n}^{\infty} (g_{4})_{n,m}^{(1+\varkappa)}c^{-n-m\varkappa}\;.
\ee

Let us stress again that in deriving this expansion we used only those properties of the Newton polygons in  variables $c$ and $c^{1+\varkappa}$ which were listed in the previous section. 
The same analysis can be made for the same Newton polygons but now in  variables $x$ and $ y $, where $x,y$ are treated  as independent (see the splitting trick in \eqref{main_pui}). In this way, one obtains the expansion of $g_4^{(\alpha)}$  around $x = \infty$ and $y = \infty$:
\be
\label{g4_xy}
g_{4}^{(1+\varkappa)} = \sum_{n=-1}^{\infty}\sum_{m=n}^{\infty} (g_{4})_{n,m}^{(1+\varkappa)}x^{-n}y^{-m}\;.
\ee

\subsubsection{${\alpha =1-\varkappa\, \in (0,1)}$}  
\label{app:A22}

Let consider the Puiseux series for the coefficient $g_4^{(\alpha)}$ \eqref{g4_num}-\eqref{g4_den} at $\alpha \in (0, 1)$ by applying the procedure described in  {\bf I. -- III.}  and {\bf (a) -- (e)} from Section \bref{app:forth}. 

{\bf I.} Consider a term of maximal power in the denominator $\mathcal{Q}^{(\alpha)}$.  To this end, rearranging ${k+m\alpha} \equiv {(k+m)-m\varkappa}$  we are looking for a term of maximal power $(k+m) - m\varkappa$. Equivalently, we need a point $(k,m)$ on the respective polygon which has a maximal value of $k+m$ and, simultaneously, a minimal value of $m$, provided $\varkappa\in (0,1)$. According to item {\bf (b)}, a maximal value $k+m$ can be  reached on the upper edge, $m = -k + L_{{\cQ}^{(\alpha)}}+U_{{\cQ}^{(\alpha)}}$. Then, a minimal $m$ on the upper edge is reached when $k $ is maximal, i.e. according to {\bf (a)}, when $k = R_{\tilde{\cQ}^{(\alpha)}}$. Thus, it is the rightmost point of the upper edge $(R_{\cQ^{(\alpha)}}, U_{\mathcal{Q}^{(\alpha)}}+L_{{\cQ}^{(\alpha)}} - R_{\cQ^{(\alpha)}}) = (4,3)$. This point corresponds to the term $a_{4,3} c^{7-3\varkappa}$.  

    \begin{figure}[h!]
 \begin{minipage}{0.48\textwidth}
     \centering
     \begin{tikzpicture}[scale = 0.5]
        \draw[thin,dotted] (0,-3) grid (-4,5);
        \draw[->] (-4,0) -- (1,0) node[right] {$k$};
        \draw[->] (0,-3) -- (0,5.8) node[above] {$m$};
        \foreach \x/\xlabel in {-4/{-4\hphantom{-}}, -3/{-3\hphantom{-}}, -2/{-2\hphantom{-}}, -1/{-1\hphantom{-}}}
    \draw (\x cm,1pt ) -- (\x cm,-1pt ) node[anchor=north,fill=white] {\xlabel};
  \foreach \y/\ylabel in { -3/{-3\hphantom{-}}, -2/{-2\hphantom{-}}, -1/{-1\hphantom{-}}, 1/1, 2/2, 3/3, 4/4, 5/5}
    \draw (1pt,\y cm) -- (-1pt ,\y cm) node[anchor=west, fill=white] {\ylabel};
    \draw[thick, red] (-4,5) -- (0,1);
    \draw[thick, red] (0,1) -- (0,-2);
    \draw[thick, red] (0,-2) -- (-2,-2);
    \draw[thick, red] (-2,-2) -- (-4,0);
    \draw[thick, red] (-4,0) -- (-4,5);
    \draw[fill=black] (-2,-2) circle (0.12) node[above right] {};
    \draw[fill=black] (-1,-2) circle (0.12) node[above right] {};
    \draw[fill=black] (0,-2) circle (0.12) node[above right] {};
    \draw[fill=black] (-3,-1) circle (0.12) node[above right] {};
    \draw[fill=black] (-2,-1) circle (0.12) node[above right] {};
    \draw[fill=black] (-1,-1) circle (0.12) node[above right] {};
    \draw[fill=black] (0,-1) circle (0.12) node[above right] {};
    \draw[fill=black] (-4,0) circle (0.12) node[above right] {};
    \draw[fill=black] (-3,0) circle (0.12) node[above right] {};
    \draw[fill=black] (-2,0) circle (0.12) node[above right] {};
    \draw[fill=black] (-1,0) circle (0.12) node[above right] {};
    \draw[fill=black] (0,0) circle (0.12) node[above right] {};
    \draw[fill=black] (-4,1) circle (0.12) node[above right] {};
    \draw[fill=black] (-3,1) circle (0.12) node[above right] {};
    \draw[fill=black] (-2,1) circle (0.12) node[above right] {};
    \draw[fill=black] (-1,1) circle (0.12) node[above right] {};
    \draw[fill=black] (0,1) circle (0.12) node[above right] {};
    \draw[fill=black] (-4,2) circle (0.12) node[above right] {};
    \draw[fill=black] (-3,2) circle (0.12) node[above right] {};
    \draw[fill=black] (-2,2) circle (0.12) node[above right] {};
    \draw[fill=black] (-1,2) circle (0.12) node[above right] {};
    \draw[fill=black] (-4,3) circle (0.12) node[above right] {};
    \draw[fill=black] (-3,3) circle (0.12) node[above right] {};
    \draw[fill=black] (-2,3) circle (0.12) node[above right] {};
    \draw[fill=black] (-4,4) circle (0.12) node[above right] {};
    \draw[fill=black] (-3,4) circle (0.12) node[above right] {};
    \draw[fill=black] (-4,5) circle (0.12) node[above right] {};
    \end{tikzpicture}
    \caption{ $\tilde{\mathcal{P}}^{(\alpha)}$-polygon, $\alpha\in(0,1)$. }
    \label{fig-g4-p-tilde-2}
 \end{minipage}\hfill
  \begin{minipage}{0.48\textwidth}
     \centering
     \begin{tikzpicture}[scale = 0.5]
        \draw[thin,dotted] (0,-3) grid (-4,5);
        \draw[->] (-4,0) -- (1,0) node[right] {$k$};
        \draw[->] (0,-3) -- (0,5.8) node[above] {$m$};
        \foreach \x/\xlabel in {-4/{-4\hphantom{-}}, -3/{-3\hphantom{-}}, -2/{-2\hphantom{-}}, -1/{-1\hphantom{-}}}
    \draw (\x cm,1pt ) -- (\x cm,-1pt ) node[anchor=north,fill=white] {\xlabel};
  \foreach \y/\ylabel in { -3/{-3\hphantom{-}}, -2/{-2\hphantom{-}}, -1/{-1\hphantom{-}}, 1/1, 2/2, 3/3, 4/4, 5/5}
    \draw (1pt,\y cm) -- (-1pt ,\y cm) node[anchor=west, fill=white] {\ylabel};
    \draw[thick, red] (-4,4) -- (0,0);
    \draw[thick, red] (0,0) -- (0,-3);
    \draw[thick, red] (0,-3) -- (-2,-3);
    \draw[thick, red] (-2,-3) -- (-4,-1);
    \draw[thick, red] (-4,-1) -- (-4,4);
    \draw[fill=black] (-2,-3) circle (0.12) node[above right] {};
    \draw[fill=black] (-1,-3) circle (0.12) node[above right] {};
    \draw[fill=black] (0,-3) circle (0.12) node[above right] {};
    \draw[fill=black] (-3,-2) circle (0.12) node[above right] {};
    \draw[fill=black] (-2,-2) circle (0.12) node[above right] {};
    \draw[fill=black] (-1,-2) circle (0.12) node[above right] {};
    \draw[fill=black] (0,-2) circle (0.12) node[above right] {};
    \draw[fill=black] (-4,-1) circle (0.12) node[above right] {};
    \draw[fill=black] (-3,-1) circle (0.12) node[above right] {};
    \draw[fill=black] (-2,-1) circle (0.12) node[above right] {};
    \draw[fill=black] (-1,-1) circle (0.12) node[above right] {};
    \draw[fill=green] (0,-1) circle (0.15) node[above right] {};
    \draw[fill=black] (-4,0) circle (0.12) node[above right] {};
    \draw[fill=black] (-3,0) circle (0.12) node[above right] {};
    \draw[fill=black] (-2,0) circle (0.12) node[above right] {};
    \draw[fill=black] (-1,0) circle (0.12) node[above right] {};
    \draw[fill=black] (0,0) circle (0.12) node[above right] {};
    \draw[fill=black] (-4,1) circle (0.12) node[above right] {};
    \draw[fill=black] (-3,1) circle (0.12) node[above right] {};
    \draw[fill=black] (-2,1) circle (0.12) node[above right] {};
    \draw[fill=green] (-1,1) circle (0.15) node[above right] {};
    \draw[fill=black] (-4,2) circle (0.12) node[above right] {};
    \draw[fill=black] (-3,2) circle (0.12) node[above right] {};
    \draw[fill=black] (-2,2) circle (0.12) node[above right] {};
    \draw[fill=black] (-4,3) circle (0.12) node[above right] {};
    \draw[fill=black] (-3,3) circle (0.12) node[above right] {};
    \draw[fill=black] (-4,4) circle (0.12) node[above right] {};
    \end{tikzpicture}
    \caption{ $\tilde{\mathcal{Q}}^{(\alpha)}$-polygon, $\alpha\in(0,1)$.  }
    \label{fig-g4-q-tilde-2}
  \end{minipage}
\end{figure}

The Newton polygons associated to $\tilde{\mathcal{P}}$ and $\tilde{\mathcal{Q}}$ are shifted down along the $m$-axis and shifted to the left along $k$-axis as shown on  fig. \bref{fig-g4-p-tilde-2},  \bref{fig-g4-q-tilde-2}. In particular, we find that  $\tilde{\mathcal{Q}}^{(\alpha)} = 1 + X(c)$, where $X(c)$ contains only negative powers of $c$ and   $X(c) \to 0$ at large $c$ (i.e. $1$ is the point $(0,0)$ and terms in $X(c)$ lie below the $k$-axis and to the left of the $y$ axis, see fig. \bref{fig-g4-q-tilde-2}). Note that now due to shifting down by $U_{\mathcal{Q}^{(\alpha)}}+L_{{\cQ}^{(\alpha)}} - R_{\cQ^{(\alpha)}} = 3$ and to the left by $R_{{\cQ}^{(\alpha)}} = 4$ we have for the $\tilde{\mathcal{Q}}^{(\alpha)}$-polygon $L_{\tilde{\mathcal{Q}}^{(\alpha)}} = -4$, $R_{\tilde{\mathcal{Q}}^{(\alpha)}} = 0$, $D_{\tilde{\mathcal{Q}}^{(\alpha)}} = -3$, $U_{\tilde{\mathcal{Q}}^{(\alpha)}} = 4$ and, since item {\bf (c)} holds, for the $\tilde{\cP}^{(\alpha)}$-polygon $L_{\tilde{\mathcal{P}}^{(\alpha)}} = -4$, $R_{\tilde{\mathcal{P}}^{(\alpha)}} = 0$, $D_{\tilde{\mathcal{P}}^{(\alpha)}} = -2$, $U_{\tilde{\mathcal{P}}^{(\alpha)}} = 5$.
   
{\bf II.} Let us consider the large-$c$ expansion of the denominator. 
\begin{lemma} The Puiseux series for $\alpha\in (0,1)$ is given by 
\be
\label{lemmaA9}
\frac{1}{\tilde{\mathcal{Q}}^{(\alpha)}} \equiv \frac{1}{1+X(c)} = \sum_{n=0}^{\infty}\sum_{m=-n}^{\infty}(\tilde{\mathcal{Q}}^{(\alpha)})_{n,m} c^{-n-m\varkappa}\,.
\ee
\end{lemma}
\begin{proof} The $X(c)$ has the following terms: there are terms  $c^{-\varkappa}$ and $c^{-1+\varkappa}$ (they are marked by the green dots on fig. \bref{fig-g4-q-tilde-2}) while others are  $c^{-r-s\varkappa}$ with coordinates $(r,s)$  which  necessarily satisfy $-r \leq  s \leq 4$, $0\leq r \leq 7$ and $(r,s) \neq (0,0), (0,-1), (1,-1)$. Their presence in $X(c)$ can be deduced  from the properties {\bf (a)},  {\bf (b)}, and {\bf (e)}. Indeed, from {\bf (b)} we see that there is a point on the upper edge such that $k+m = U_{\tilde{\mathcal{Q}}^{(\alpha)}} +L_{\tilde{\mathcal{Q}}^{(\alpha)}}  = 0$ and $k = -1$, so $m = 1$, therefore, the corresponding term is $c^{-\varkappa}$.  Then, from {\bf (e)} we see that there is a term with coordinates $(R_{\tilde{\cQ}^{(\alpha)}}, U_{\tilde{\mathcal{Q}}^{(\alpha)}}+L_{\tilde{\mathcal{Q}}^{(\alpha)}}-R_{\tilde{\cQ}^{(\alpha)}}-1) = (0,-1)$, therefore, it is  $c^{-1+\varkappa}$. From {\bf (a)}  it follows that all other terms in $\tilde{\mathcal{Q}}^{(\alpha)}$ are  $c^{(k+m)-m\varkappa}$, where $k,m$ necessarily satisfy the inequalities  $-4 = L_{\tilde{\mathcal{Q}}^{(\alpha)}} \leq k \leq R_{\tilde{\mathcal{Q}}^{(\alpha)}} = 0$, $ -3 = D_{\tilde{\mathcal{Q}}^{(\alpha)}} \leq m \leq U_{\tilde{\mathcal{Q}}^{(\alpha)}} = 4$ so that $-7 \leq k+m \leq m  \leq 4$. From {\bf (b)} it follows that  $ k+m\leq  0$.  Denoting $r \equiv -(k+m)$ and $ s\equiv m$ we obtain that these terms can be  represented as $c^{-r-s\varkappa}$, where $ -r \leq s \leq 4$ and $0 \leq r \leq 7$.  

Now, one expands the denominator, 
\be
\label{X-expansion-2}
\frac{1}{1+X(c)} = \sum_{m=0}^{\infty}(-1)^m X^m(c)\,.
\ee
Firstly, let us show that every term in $X^m(c)$ has a form $c^{-N-M\varkappa}$ for $-N \leq M, N\geq 0$. Since  every term in $X(c)$ has a form $c^{-r-s\varkappa}$, then  expanding $X^m(c)$ one obtains that 
\be
\label{term_multinom-2}
\prod_{i=1}^{m} c^{-r_i - s_i \varkappa} = c^{-\sum\limits_{i=1}^m r_i - \varkappa\sum\limits_{i=1}^m s_i}
\ee
 for $0\leq r_i \leq 7$ and $-r \leq s_i \leq 4$.  Since these inequalities hold for all $i$, then  $-\sum_{i=1}^m r_i \leq \sum_{i=1}^m s_i$. Denoting $\sum_{i=1}^m r_i = N$ and $\sum_{i=1}^m s_i = M $ one finds that the right-hand side of \eqref{term_multinom-2} becomes $c^{-N-M\varkappa}$ for  $-N \leq M$ and $N \geq 0$.  Secondly, let us show that every such term $c^{-N - M\varkappa}$ is present in the decomposition \eqref{X-expansion-2}. Indeed, 
\be
\label{c-multinom-2}
c^{-N-M\varkappa} = c^{N(-1+\varkappa) + (M+N)(-\varkappa)}\,.
\ee
Then it follows that $X^{2N+M}(c) = (c^{-\varkappa} +... + c^{-1+\varkappa})^{2N+M}$ contains  exactly the term in right-hand side of \eqref{c-multinom-2}. Therefore,  expanding and rearranging $X^m(c)$ in \eqref{X-expansion-2} we see that all possible degrees in $c$ in \eqref{X-expansion-2} are generated by multiplication of $c^{-\varkappa}$ and $c^{-1-\varkappa}$, i.e. they have form $c^{-n-m\varkappa}$ for all  $n,m : -n \leq m$ and $n\geq 0$. In this way, we obtain the Puiseux series \eqref{lemmaA9}. 
\end{proof}

{\bf III.} Finally, consider a product of  $1/\tilde{\mathcal{Q}}^{(\alpha)}$  and $\tilde{\mathcal{P}}^{(\alpha)}$. First of all, as in   item  {\bf I.} one can show that due to the properties {\bf (a)} and {\bf (b)} a point with maximal $(k+m)-m\varkappa, ~\forall \varkappa \in (0,1)$ in the $\tilde{\mathcal{P}}^{(\alpha)}$-polygon  has coordinates  $(R_{\tilde{\mathcal{P}}^{(\alpha)}}, U_{\tilde{\mathcal{P}}^{(\alpha)}} + L_{\tilde{\mathcal{P}}^{(\alpha)}} - R_{\tilde{\mathcal{P}}^{(\alpha)}}) = (0,1)$, due to {\bf (c)}. It corresponds to $c^{\alpha} \equiv c^{1-\varkappa}$, which, therefore, is the highest-order term in the numerator.  Subleading terms in the numerator have the form $c^{-r-s\varkappa}$ for $-r \leq s<5$ and $-1 \leq r \leq 6$. The reasoning is the same as for $\tilde{\cQ}^{(\alpha)}$:  every subleading term has a form $c^{(k+m) - m\varkappa}$ where $k,m$ necessarily satisfy $k+m\leq U_{\tilde{\mathcal{P}}^{(\alpha)}} + L_{\tilde{\mathcal{P}}^{(\alpha)}} = 1$, $-4 = L_{\tilde{\mathcal{P}}^{(\alpha)}} \leq  k\leq R_{\tilde{\mathcal{P}}^{(\alpha)}} = 0$ and $-2 = D_{\tilde{\mathcal{P}}^{(\alpha)}} \leq m\leq U_{\tilde{\mathcal{P}}^{(\alpha)}}  = 5$. Denoting $r \equiv -(k+m) $ and  $s \equiv m$ this term can be represented as  $c^{-r-s\varkappa}$ with $-r \leq s \leq 5$ and $-1\leq r \leq 6$.  Using the expansion \eqref{lemmaA9} one finds   
\be
g_{4}^{(1-\varkappa)} = \tilde{\mathcal{P}}^{(\alpha)} \frac{1}{\tilde{\mathcal{Q}}^{(\alpha)}} \sim (c^{1-\varkappa}+...)\sum_{n=0}^{\infty}\sum_{m=-n}^{\infty} (\tilde{\mathcal{Q}}^{(\alpha)})_{n,m} c^{-n-m\varkappa}\,,
\ee
where the ellipsis in $\tilde{\mathcal{P}}^{(\alpha)}$ denotes the subleading terms.   Multiplying $c^{1-\varkappa}$ by $c^{-n-m\varkappa}$ for $-n \leq m$, $n\geq 0$ one obtains $c^{-(n-1)-(m+1)\varkappa}$ so $-n+1 \leq m+1$ and $n -1\geq -1$. Multiplying subleading term $c^{-r-s\varkappa}$ by $c^{-n-m\varkappa}$ one obtains $c^{-(n+r) - (m+s)\varkappa}$ so  $ -(n+r) \leq m+s$ and $n+r \geq 0$. This shows that  multiplying $\tilde{\cP}^{(\alpha)}$ by the expansion \eqref{lemmaA9} and rearranging terms change  limits of the sum and expansion coefficients, while the general structure is not changed: it still consists of terms $c^{-n-m\varkappa}$ for $-n\leq m$ but now $ n \geq -1$. Thus, the Puiseux series  $g_4^{(1-\varkappa)}$ can be cast into the form
\be
\label{g4_P1-2}
g_{4}^{(1-\varkappa)} = \sum_{n=-1}^{\infty}\sum_{m=-n}^{\infty} (g_{4})_{n,m}^{(1-\varkappa)}c^{-n-m\varkappa}\;.
\ee

Finally, by analogy with  \eqref{g4_xy} one can treat $c \equiv x$ and $c^{-\varkappa} \equiv y$ as independent and expand $g_4^{(\alpha)}$  around $x = \infty$ and $y = 0$:
\be
\label{g4_xy-2}
g_{4}^{(1-\varkappa)} = \sum_{n=-1}^{\infty}\sum_{m=-n}^{\infty} (g_{4})_{n,m}^{(1-\varkappa)}x^{-n}y^{m}\;.
\ee

\subsubsection{${\alpha =1}$}
 
\begin{figure}[h!]
 \begin{minipage}{0.48\textwidth}
     \centering
     \begin{tikzpicture}[scale = 0.5]
        \draw[thin,dotted] (0,0) grid (4,8);
        \draw[->] (0,0) -- (4.4,0) node[right] {$k$};
        \draw[->] (0,0) -- (0,8.4) node[above] {$m$};
        \foreach \x/\xlabel in { 1/1, 2/2, 3/3, 4/4}
    \draw (\x cm,1pt ) -- (\x cm,-1pt ) node[anchor=north,fill=white] {\xlabel};
  \foreach \y/\ylabel in {1/1, 2/2, 3/3, 4/4, 5/5, 6/6, 7/7, 8/8}
    \draw (1pt,\y cm) -- (-1pt ,\y cm) node[anchor=east, fill=white] {\ylabel};
    \draw[thick, orange] (0,8) -- (4,4);
    \draw[thick, orange] (0,7) -- (4,3);
    \draw[thick, orange] (0,6) -- (4,2);
    \draw[thick, orange] (0,5) -- (4,1);
    \draw[thick, orange] (0,4) -- (3,1);
    \draw[thick, orange] (0,3) -- (2,1);
    \draw[fill=black] (2,1) circle (0.12) node[above right] {};
    \draw[fill=black] (3,1) circle (0.12) node[above right] {};
    \draw[fill=black] (4,1) circle (0.12) node[above right] {};
    \draw[fill=black] (1,2) circle (0.12) node[above right] {};
    \draw[fill=black] (2,2) circle (0.12) node[above right] {};
    \draw[fill=black] (3,2) circle (0.12) node[above right] {};
    \draw[fill=black] (4,2) circle (0.12) node[above right] {};
    \draw[fill=black] (0,3) circle (0.12) node[above right] {};
    \draw[fill=black] (1,3) circle (0.12) node[above right] {};
    \draw[fill=black] (2,3) circle (0.12) node[above right] {};
    \draw[fill=black] (3,3) circle (0.12) node[above right] {};
    \draw[fill=black] (4,3) circle (0.12) node[above right] {};
    \draw[fill=black] (0,4) circle (0.12) node[above right] {};
    \draw[fill=black] (1,4) circle (0.12) node[above right] {};
    \draw[fill=black] (2,4) circle (0.12) node[above right] {};
    \draw[fill=black] (3,4) circle (0.12) node[above right] {};
    \draw[fill=black] (4,4) circle (0.12) node[above right] {};
    \draw[fill=black] (0,5) circle (0.12) node[above right] {};
    \draw[fill=black] (1,5) circle (0.12) node[above right] {};
    \draw[fill=black] (2,5) circle (0.12) node[above right] {};
    \draw[fill=black] (3,5) circle (0.12) node[above right] {};
    \draw[fill=black] (0,6) circle (0.12) node[above right] {};
    \draw[fill=black] (1,6) circle (0.12) node[above right] {};
    \draw[fill=black] (2,6) circle (0.12) node[above right] {};
    \draw[fill=black] (0,7) circle (0.12) node[above right] {};
    \draw[fill=black] (1,7) circle (0.12) node[above right] {};
    \draw[fill=black] (0,8) circle (0.12) node[above right] {};
    \end{tikzpicture}
    \caption{$\cP^{(1)}$-polygon.}
    \label{fig-num-g4-1}
 \end{minipage}\hfill
  \begin{minipage}{0.48\textwidth}
     \centering
     \begin{tikzpicture}[scale = 0.5]
          \draw[thin,dotted] (0,0) grid (4,8);
        \draw[->] (0,0) -- (4.4,0) node[right] {$k$};
        \draw[->] (0,0) -- (0,8.4) node[above] {$m$};
        \foreach \x/\xlabel in { 1/1, 2/2, 3/3, 4/4}
    \draw (\x cm,1pt ) -- (\x cm,-1pt ) node[anchor=north,fill=white] {\xlabel};
  \foreach \y/\ylabel in {1/1, 2/2, 3/3, 4/4, 5/5, 6/6, 7/7, 8/8}
    \draw (1pt,\y cm) -- (-1pt ,\y cm) node[anchor=east, fill=white] {\ylabel};
    \draw[thick, orange] (0,7) -- (4,3);
    \draw[thick, orange] (0,6) -- (4,2);
    \draw[thick, orange] (0,5) -- (4,1);
    \draw[thick, orange] (0,4) -- (4,0);
    \draw[thick, orange] (0,3) -- (3,0);
    \draw[thick, orange] (0,2) -- (2,0);
    \draw[fill=black] (2,0) circle (0.12) node[above right] {};
    \draw[fill=black] (3,0) circle (0.12) node[above right] {};
    \draw[fill=black] (4,0) circle (0.12) node[above right] {};
    \draw[fill=black] (1,1) circle (0.12) node[above right] {};
    \draw[fill=black] (2,1) circle (0.12) node[above right] {};
    \draw[fill=black] (3,1) circle (0.12) node[above right] {};
    \draw[fill=black] (4,1) circle (0.12) node[above right] {};
    \draw[fill=black] (0,2) circle (0.12) node[above right] {};
    \draw[fill=black] (1,2) circle (0.12) node[above right] {};
    \draw[fill=black] (2,2) circle (0.12) node[above right] {};
    \draw[fill=black] (3,2) circle (0.12) node[above right] {};
    \draw[fill=black] (4,2) circle (0.12) node[above right] {};
    \draw[fill=black] (0,3) circle (0.12) node[above right] {};
    \draw[fill=black] (1,3) circle (0.12) node[above right] {};
    \draw[fill=black] (2,3) circle (0.12) node[above right] {};
    \draw[fill=black] (3,3) circle (0.12) node[above right] {};
    \draw[fill=black] (4,3) circle (0.12) node[above right] {};
    \draw[fill=black] (0,4) circle (0.12) node[above right] {};
    \draw[fill=black] (1,4) circle (0.12) node[above right] {};
    \draw[fill=black] (2,4) circle (0.12) node[above right] {};
    \draw[fill=black] (3,4) circle (0.12) node[above right] {};
    \draw[fill=black] (0,5) circle (0.12) node[above right] {};
    \draw[fill=black] (1,5) circle (0.12) node[above right] {};
    \draw[fill=black] (2,5) circle (0.12) node[above right] {};
    \draw[fill=black] (0,6) circle (0.12) node[above right] {};
    \draw[fill=black] (1,6) circle (0.12) node[above right] {};
    \draw[fill=black] (0,7) circle (0.12) node[above right] {};
    \end{tikzpicture}
    \caption{$\cQ^{(1)}$-polygon.}
    \label{fig-den-g4-1}
  \end{minipage}
\end{figure}
In this case,  every term in $\cP^{(1)}$ and  $\cQ^{(1)}$ has a form $c^{k+m}$, see fig. \bref{fig-num-g4-1}, \bref{fig-den-g4-1}: constant $k+m$ are shown as orange lines, i.e. the dots on a given orange line correspond to terms of equal power of $c$. Of course, both the numerator and denominator are crucially simplified at $\alpha=1$, but it is still convenient to represent them as Newton polygons which are now reorganized into lines of constant inclination each of which represents a single term. Now let us find the corresponding Puiseux series following  items {\bf I -- III} of section \bref{app:forth}.

{\bf I.} Consider a term of maximal power in the denominator $\mathcal{Q}^{(1)}$.  To this end, we are looking for a term with a maximal power $k+m$. Due to {\bf (b)}  it is the upper edge of the polygon: $k+m = U_{\mathcal{Q}^{(\alpha)}} +L_{\mathcal{Q}^{(\alpha)}} = 7$. So the highest order term in $\mathcal{Q}^{(1)}$ equals $c^7 \sum_{k = 0}^{4} a_{k,7-k}$.

Every term in $\mathcal{Q}^{(1)}$ has a form $c^{k+m}$ for pairs $(k,m)$ which necessarily satisfy $k+m \leq 7$, $0\leq k \leq 4$ and $0\leq m \leq 7$ so that $0\leq k+m \leq 7$. Thus, every term in the redefined  $\tilde{\cQ}^{(1)}$ has a form $c^{k+m}$ but $-7 \leq k+m \leq 0$.  Denoting $n \equiv -(k+m)$ one finds that every term $\tilde{\mathcal{Q}}^{(1)}$ becomes $c^{-n}$ for $0\leq n \leq 7$, and term with $k+m = 0$ is exactly $1$, so $\tilde{\cQ}^{(1)} = 1 +X(c)$, where $X(c)$ consists of terms $c^{-n}$, $1\leq n \leq 7$. Since {\bf (c)} holds,  then $\tilde{\cP}^{(1)}$  consists of terms $c^{-n}$, $-1 \leq n \leq 6$. Redefining the numerator and denominator leads to the $\tilde{\cP}^{(1)}$-polygon  and $\tilde{\cQ}^{(1)}$-polygon which are shifted down along the $m$-axis as shown on fig. \bref{fig-g4-p-tilde-1} and \bref{fig-g4-q-tilde-1}. 
 
    \begin{figure}[h!]
 \begin{minipage}{0.48\textwidth}
     \centering
     \begin{tikzpicture}[scale = 0.5]
        \draw[thin,dotted] (0,-7) grid (4,1);
        \draw[->] (0,0) -- (4.4,0) node[right] {$k$};
        \draw[->] (0,-7) -- (0,1.4) node[above] {$m$};
        \foreach \x/\xlabel in { 1/1, 2/2, 3/3, 4/4}
    \draw (\x cm,1pt ) -- (\x cm,-1pt ) node[anchor=north,fill=white] {\xlabel};
  \foreach \y/\ylabel in {-7/{-7\hphantom{-}},-6/{-6\hphantom{-}}, -5/{-5\hphantom{-}}, -4/{-4\hphantom{-}}, -3/{-3\hphantom{-}}, -2/{-2\hphantom{-}}, -1/{-1\hphantom{-}}, 1/1}
    \draw (1pt,\y cm) -- (-1pt ,\y cm) node[anchor=east, fill=white] {\ylabel};
    \draw[thick, orange] (0,1) -- (4,-3);
    \draw[thick, orange] (0,0) -- (4,-4);
    \draw[thick, orange] (0,-1) -- (4,-5);
    \draw[thick, orange] (0,-2) -- (4,-6);
    \draw[thick, orange] (0,-3) -- (3,-6);
    \draw[thick, orange] (0,-4) -- (2,-6);
    \draw[fill=black] (2,-6) circle (0.12) node[above right] {};
    \draw[fill=black] (3,-6) circle (0.12) node[above right] {};
    \draw[fill=black] (4,-6) circle (0.12) node[above right] {};
    \draw[fill=black] (1,-5) circle (0.12) node[above right] {};
    \draw[fill=black] (2,-5) circle (0.12) node[above right] {};
    \draw[fill=black] (3,-5) circle (0.12) node[above right] {};
    \draw[fill=black] (4,-5) circle (0.12) node[above right] {};
    \draw[fill=black] (0,-4) circle (0.12) node[above right] {};
    \draw[fill=black] (1,-4) circle (0.12) node[above right] {};
    \draw[fill=black] (2,-4) circle (0.12) node[above right] {};
    \draw[fill=black] (3,-4) circle (0.12) node[above right] {};
    \draw[fill=black] (4,-4) circle (0.12) node[above right] {};
    \draw[fill=black] (0,-3) circle (0.12) node[above right] {};
    \draw[fill=black] (1,-3) circle (0.12) node[above right] {};
    \draw[fill=black] (2,-3) circle (0.12) node[above right] {};
    \draw[fill=black] (3,-3) circle (0.12) node[above right] {};
    \draw[fill=black] (4,-3) circle (0.12) node[above right] {};
    \draw[fill=black] (0,-2) circle (0.12) node[above right] {};
    \draw[fill=black] (1,-2) circle (0.12) node[above right] {};
    \draw[fill=black] (2,-2) circle (0.12) node[above right] {};
    \draw[fill=black] (3,-2) circle (0.12) node[above right] {};
    \draw[fill=black] (0,-1) circle (0.12) node[above right] {};
    \draw[fill=black] (1,-1) circle (0.12) node[above right] {};
    \draw[fill=black] (2,-1) circle (0.12) node[above right] {};
    \draw[fill=black] (0,0) circle (0.12) node[above right] {};
    \draw[fill=black] (1,0) circle (0.12) node[above right] {};
    \draw[fill=black] (0,1) circle (0.12) node[above right] {};
    \end{tikzpicture}
    \caption{$\tilde{\mathcal{P}}^{(1)}$-polygon.}
    \label{fig-g4-p-tilde-1}
 \end{minipage}\hfill
  \begin{minipage}{0.48\textwidth}
     \centering
     \begin{tikzpicture}[scale = 0.5]
       \draw[thin,dotted] (0,-7) grid (4,1);
        \draw[->] (0,0) -- (4.4,0) node[right] {$k$};
        \draw[->] (0,-7) -- (0,1.4) node[above] {$m$};
        \foreach \x/\xlabel in { 1/1, 2/2, 3/3, 4/4}
    \draw (\x cm,1pt ) -- (\x cm,-1pt ) node[anchor=north,fill=white] {\xlabel};
  \foreach \y/\ylabel in {-7/{-7\hphantom{-}}, -6/{-6\hphantom{-}}, -5/{-5\hphantom{-}}, -4/{-4\hphantom{-}}, -3/{-3\hphantom{-}}, -2/{-2\hphantom{-}}, -1/{-1\hphantom{-}}, 1/1}
    \draw (1pt,\y cm) -- (-1pt ,\y cm) node[anchor=east, fill=white] {\ylabel};
    \draw[thick, orange] (0,0) -- (4,-4);
    \draw[thick, orange] (0,-1) -- (4,-5);
    \draw[thick, orange] (0,-2) -- (4,-6);
    \draw[thick, orange] (0,-3) -- (4,-7);
    \draw[thick, orange] (0,-4) -- (3,-7);
    \draw[thick, orange] (0,-5) -- (2,-7);
    \draw[fill=black] (2,-7) circle (0.12) node[above right] {};
    \draw[fill=black] (3,-7) circle (0.12) node[above right] {};
    \draw[fill=black] (4,-7) circle (0.12) node[above right] {};
    \draw[fill=black] (1,-6) circle (0.12) node[above right] {};
    \draw[fill=black] (2,-6) circle (0.12) node[above right] {};
    \draw[fill=black] (3,-6) circle (0.12) node[above right] {};
    \draw[fill=black] (4,-6) circle (0.12) node[above right] {};
    \draw[fill=black] (0,-5) circle (0.12) node[above right] {};
    \draw[fill=black] (1,-5) circle (0.12) node[above right] {};
    \draw[fill=black] (2,-5) circle (0.12) node[above right] {};
    \draw[fill=black] (3,-5) circle (0.12) node[above right] {};
    \draw[fill=black] (4,-5) circle (0.12) node[above right] {};
    \draw[fill=black] (0,-4) circle (0.12) node[above right] {};
    \draw[fill=black] (1,-4) circle (0.12) node[above right] {};
    \draw[fill=black] (2,-4) circle (0.12) node[above right] {};
    \draw[fill=black] (3,-4) circle (0.12) node[above right] {};
    \draw[fill=black] (4,-4) circle (0.12) node[above right] {};
    \draw[fill=black] (0,-3) circle (0.12) node[above right] {};
    \draw[fill=black] (1,-3) circle (0.12) node[above right] {};
    \draw[fill=black] (2,-3) circle (0.12) node[above right] {};
    \draw[fill=black] (3,-3) circle (0.12) node[above right] {};
    \draw[fill=black] (0,-2) circle (0.12) node[above right] {};
    \draw[fill=black] (1,-2) circle (0.12) node[above right] {};
    \draw[fill=black] (2,-2) circle (0.12) node[above right] {};
    \draw[fill=black] (0,-1) circle (0.12) node[above right] {};
    \draw[fill=black] (1,-1) circle (0.12) node[above right] {};
    \draw[fill=black] (0,0) circle (0.12) node[above right] {};
    \end{tikzpicture}
    \caption{ $\tilde{\mathcal{Q}}^{(1)}$-polygon.  }
    \label{fig-g4-q-tilde-1}
  \end{minipage}
\end{figure}

{\bf II.} The Laurent decomposition of   $\tilde{\cQ}^{(1)}$ is given by
\be
\frac{1}{\tilde{\cQ}^{(1)}} = \frac{1}{1 + X(c)} = \sum_{n=0}^\infty \left( \tilde{\cQ}^{(1)}\right)_{n} c^{-n} \,.
\ee
The line $k+m = -1$  is below the upper edge in the $\tilde{\cQ}^{(1)}$-polygon on fig. \bref{fig-g4-q-tilde-1}. From {\bf (d)} and {\bf (e)} it follows that the line  has at least two points. Thus, there are non-zero terms  $c^{-1}$. Then, decomposing $1/\tilde{\cQ}^{(1)}$ one obtains  
\be
\frac{1}{1 + X(c)} = \sum_{n = 0}^{\infty} (-1)^{n} X^n(c) \,,
\ee
so that each $X^n(c) = (c^{-1} + ...)^{n}$ contains terms $c^{-n}$ for all $n\geq 0$. 

{\bf III.} Consider a product of $1/\tilde{\cQ}^{(1)}$ and $\tilde{\cP}^{(1)}$. Due to {\bf (c)} the term of maximal power $k+m$ in $\tilde{\cP}^{(1)}$ is $c^1$. Subleading terms  $c^{-n}$ with $0\leq n\leq 6$. Then, 
\be
\label{g4-1}
g_4^{(1)} = \tilde{\cP}^{(1)}/\tilde{\cQ}^{(1)} = (c^{1} + .. ) \sum_{n = 0}^{\infty}  \left( \tilde{\cQ}^{(1)}\right)_{n} c^{-n} = \sum_{n = -1}^{\infty} \left(g_4^{(1)} \right)_n c^{-n} \;,
\ee 
and the principal part has two terms $c^1$ and $c^0$.

\subsubsection{${\alpha =0}$}
 
\begin{figure}[h!]
 \begin{minipage}{0.48\textwidth}
     \centering
     \begin{tikzpicture}[scale = 0.5]
        \draw[thin,dotted] (0,0) grid (4,8);
        \draw[->] (0,0) -- (4.4,0) node[right] {$k$};
        \draw[->] (0,0) -- (0,8.4) node[above] {$m$};
        \foreach \x/\xlabel in { 1/1, 2/2, 3/3, 4/4}
    \draw (\x cm,1pt ) -- (\x cm,-1pt ) node[anchor=north,fill=white] {\xlabel};
  \foreach \y/\ylabel in {1/1, 2/2, 3/3, 4/4, 5/5, 6/6, 7/7, 8/8}
    \draw (1pt,\y cm) -- (-1pt ,\y cm) node[anchor=east, fill=white] {\ylabel};
    \draw[thick, orange] (0,8) -- (0,3);
    \draw[thick, orange] (1,7) -- (1,2);
    \draw[thick, orange] (2,6) -- (2,1);
    \draw[thick, orange] (3,5) -- (3,1);
    \draw[thick, orange] (4,4) -- (4,1);
    \draw[fill=black] (2,1) circle (0.12) node[above right] {};
    \draw[fill=black] (3,1) circle (0.12) node[above right] {};
    \draw[fill=black] (4,1) circle (0.12) node[above right] {};
    \draw[fill=black] (1,2) circle (0.12) node[above right] {};
    \draw[fill=black] (2,2) circle (0.12) node[above right] {};
    \draw[fill=black] (3,2) circle (0.12) node[above right] {};
    \draw[fill=black] (4,2) circle (0.12) node[above right] {};
    \draw[fill=black] (0,3) circle (0.12) node[above right] {};
    \draw[fill=black] (1,3) circle (0.12) node[above right] {};
    \draw[fill=black] (2,3) circle (0.12) node[above right] {};
    \draw[fill=black] (3,3) circle (0.12) node[above right] {};
    \draw[fill=black] (4,3) circle (0.12) node[above right] {};
    \draw[fill=black] (0,4) circle (0.12) node[above right] {};
    \draw[fill=black] (1,4) circle (0.12) node[above right] {};
    \draw[fill=black] (2,4) circle (0.12) node[above right] {};
    \draw[fill=black] (3,4) circle (0.12) node[above right] {};
    \draw[fill=black] (4,4) circle (0.12) node[above right] {};
    \draw[fill=black] (0,5) circle (0.12) node[above right] {};
    \draw[fill=black] (1,5) circle (0.12) node[above right] {};
    \draw[fill=black] (2,5) circle (0.12) node[above right] {};
    \draw[fill=black] (3,5) circle (0.12) node[above right] {};
    \draw[fill=black] (0,6) circle (0.12) node[above right] {};
    \draw[fill=black] (1,6) circle (0.12) node[above right] {};
    \draw[fill=black] (2,6) circle (0.12) node[above right] {};
    \draw[fill=black] (0,7) circle (0.12) node[above right] {};
    \draw[fill=black] (1,7) circle (0.12) node[above right] {};
    \draw[fill=black] (0,8) circle (0.12) node[above right] {};
    \end{tikzpicture}
    \caption{$\cP^{(0)}$-polygon.}
    \label{fig-num-g4-0}
 \end{minipage}\hfill
  \begin{minipage}{0.48\textwidth}
     \centering
     \begin{tikzpicture}[scale = 0.5]
          \draw[thin,dotted] (0,0) grid (4,8);
        \draw[->] (0,0) -- (4.4,0) node[right] {$k$};
        \draw[->] (0,0) -- (0,8.4) node[above] {$m$};
        \foreach \x/\xlabel in { 1/1, 2/2, 3/3, 4/4}
    \draw (\x cm,1pt ) -- (\x cm,-1pt ) node[anchor=north,fill=white] {\xlabel};
  \foreach \y/\ylabel in {1/1, 2/2, 3/3, 4/4, 5/5, 6/6, 7/7, 8/8}
    \draw (1pt,\y cm) -- (-1pt ,\y cm) node[anchor=east, fill=white] {\ylabel};
    \draw[thick, orange] (0,7) -- (0,2);
    \draw[thick, orange] (1,6) -- (1,1);
    \draw[thick, orange] (2,5) -- (2,0);
    \draw[thick, orange] (3,4) -- (3,0);
    \draw[thick, orange] (4,3) -- (4,0);
    \draw[fill=black] (2,0) circle (0.12) node[above right] {};
    \draw[fill=black] (3,0) circle (0.12) node[above right] {};
    \draw[fill=black] (4,0) circle (0.12) node[above right] {};
    \draw[fill=black] (1,1) circle (0.12) node[above right] {};
    \draw[fill=black] (2,1) circle (0.12) node[above right] {};
    \draw[fill=black] (3,1) circle (0.12) node[above right] {};
    \draw[fill=black] (4,1) circle (0.12) node[above right] {};
    \draw[fill=black] (0,2) circle (0.12) node[above right] {};
    \draw[fill=black] (1,2) circle (0.12) node[above right] {};
    \draw[fill=black] (2,2) circle (0.12) node[above right] {};
    \draw[fill=black] (3,2) circle (0.12) node[above right] {};
    \draw[fill=black] (4,2) circle (0.12) node[above right] {};
    \draw[fill=black] (0,3) circle (0.12) node[above right] {};
    \draw[fill=black] (1,3) circle (0.12) node[above right] {};
    \draw[fill=black] (2,3) circle (0.12) node[above right] {};
    \draw[fill=black] (3,3) circle (0.12) node[above right] {};
    \draw[fill=black] (4,3) circle (0.12) node[above right] {};
    \draw[fill=black] (0,4) circle (0.12) node[above right] {};
    \draw[fill=black] (1,4) circle (0.12) node[above right] {};
    \draw[fill=black] (2,4) circle (0.12) node[above right] {};
    \draw[fill=black] (3,4) circle (0.12) node[above right] {};
    \draw[fill=black] (0,5) circle (0.12) node[above right] {};
    \draw[fill=black] (1,5) circle (0.12) node[above right] {};
    \draw[fill=black] (2,5) circle (0.12) node[above right] {};
    \draw[fill=black] (0,6) circle (0.12) node[above right] {};
    \draw[fill=black] (1,6) circle (0.12) node[above right] {};
    \draw[fill=black] (0,7) circle (0.12) node[above right] {};
    \end{tikzpicture}
    \caption{$\cQ^{(0)}$-polygon.}
    \label{fig-den-g4-0}
  \end{minipage}
\end{figure} 

In this case,  every term in $\cP^{(0)}$ and  $\cQ^{(0)}$ has a form $c^{k}$, see on fig. \bref{fig-num-g4-0}, \bref{fig-den-g4-0}: constant $k$ are shown as orange lines, i.e. the dots on a given orange line corresponds to terms of equal power of $c$. Similar to the case $\alpha=1$ the respective Newton  polygons are sliced into (vertical) lines representing the terms of equal powers.  Now let us find the corresponding Puiseux series following  items {\bf I -- III} of section \bref{app:forth}.

{\bf I.} Consider a term of maximal power in the denominator $\mathcal{Q}^{(0)}$.  To this end, we are looking for a term of a maximal power $k$. Due to {\bf (a)} these are  points with $k = R_{\cQ^{(0)}} = 4$. Thus, the highest order term in $\mathcal{Q}^{(1)}$ is given by $c^4 \sum_{m = 0}^{3} a_{4,m}$.

Every term in the redefined $\tilde{\cQ}^{(0)}$ has a form $c^{k}$ with $-4 \leq k \leq 0$.  Denoting $n \equiv -k$ every term $\tilde{\mathcal{Q}}^{(0)}$ becomes $c^{-n}$ for $0\leq n \leq 4$, and the $k = 0$ term is $1$. Then, $\tilde{\cQ}^{(0)} = 1 +X(c)$, where $X(c)$ consists of terms  $c^{-n}$, $1\leq n \leq 4$. Due to  {\bf (c)} the polynomial  $\tilde{\cP}^{(0)}$ consists of terms  $c^{-n}$, $0 \leq n \leq 4$.   Redefining the numerator and denominator leads to the $\tilde{\cP}^{(0)}$-polygon  and $\tilde{\cQ}^{(0)}$-polygon which are shifted leftwards and downwards as shown on fig. \bref{fig-g4-p-tilde-0} and \bref{fig-g4-q-tilde-0}.

{\bf II.} The Laurent decomposition of   $\tilde{\cQ}^{(0)}$ is given by
\be
\frac{1}{\tilde{\cQ}^{(0)}} = \frac{1}{1 + X(c)} = \sum_{n=0}^\infty \left( \tilde{\cQ}^{(0)}\right)_{n} c^{-n} \,.
\ee
The line $k = -1$ intersects  the upper edge  in the $\tilde{\cQ}^{(0)}$-polygon on fig. \bref{fig-g4-q-tilde-0}. From {\bf (b)} it follows that this line  has at least one point. Thus, there is a non-zero term  $c^{-1}$.  Then, decomposing $1/\tilde{\cQ}^{(\alpha)}$ one has 
\be
\frac{1}{1 + X(c)} = \sum_{n = 0}^{\infty} (-1)^{n} X^n(c) \,,
\ee
so every $X^n(c) = (c^{-1} + ...)^{n}$ contains terms $c^{-n}$ for all $n\geq 0$. 
  \begin{figure}[h!]
 \begin{minipage}{0.48\textwidth}
     \centering
     \begin{tikzpicture}[scale = 0.5]
        \draw[thin,dotted] (0,-3) grid (-4,5);
        \draw[->] (-4,0) -- (1,0) node[right] {$k$};
        \draw[->] (0,-3) -- (0,5.8) node[above] {$m$};
        \foreach \x/\xlabel in {-4/{-4\hphantom{-}}, -3/{-3\hphantom{-}}, -2/{-2\hphantom{-}}, -1/{-1\hphantom{-}}}
    \draw (\x cm,1pt ) -- (\x cm,-1pt ) node[anchor=north,fill=white] {\xlabel};
  \foreach \y/\ylabel in { -3/{-3\hphantom{-}}, -2/{-2\hphantom{-}}, -1/{-1\hphantom{-}}, 1/1, 2/2, 3/3, 4/4, 5/5}
    \draw (1pt,\y cm) -- (-1pt ,\y cm) node[anchor=west, fill=white] {\ylabel};
    \draw[thick, orange] (-4,5) -- (-4,0);
    \draw[thick, orange] (-3,4) -- (-3,-1);
    \draw[thick, orange] (-2,3) -- (-2,-2);
    \draw[thick, orange] (-1,2) -- (-1,-2);
    \draw[thick, orange] (0,1) -- (0,-2);
    \draw[fill=black] (-2,-2) circle (0.12) node[above right] {};
    \draw[fill=black] (-1,-2) circle (0.12) node[above right] {};
    \draw[fill=black] (0,-2) circle (0.12) node[above right] {};
    \draw[fill=black] (-3,-1) circle (0.12) node[above right] {};
    \draw[fill=black] (-2,-1) circle (0.12) node[above right] {};
    \draw[fill=black] (-1,-1) circle (0.12) node[above right] {};
    \draw[fill=black] (0,-1) circle (0.12) node[above right] {};
    \draw[fill=black] (-4,0) circle (0.12) node[above right] {};
    \draw[fill=black] (-3,0) circle (0.12) node[above right] {};
    \draw[fill=black] (-2,0) circle (0.12) node[above right] {};
    \draw[fill=black] (-1,0) circle (0.12) node[above right] {};
    \draw[fill=black] (0,0) circle (0.12) node[above right] {};
    \draw[fill=black] (-4,1) circle (0.12) node[above right] {};
    \draw[fill=black] (-3,1) circle (0.12) node[above right] {};
    \draw[fill=black] (-2,1) circle (0.12) node[above right] {};
    \draw[fill=black] (-1,1) circle (0.12) node[above right] {};
    \draw[fill=black] (0,1) circle (0.12) node[above right] {};
    \draw[fill=black] (-4,2) circle (0.12) node[above right] {};
    \draw[fill=black] (-3,2) circle (0.12) node[above right] {};
    \draw[fill=black] (-2,2) circle (0.12) node[above right] {};
    \draw[fill=black] (-1,2) circle (0.12) node[above right] {};
    \draw[fill=black] (-4,3) circle (0.12) node[above right] {};
    \draw[fill=black] (-3,3) circle (0.12) node[above right] {};
    \draw[fill=black] (-2,3) circle (0.12) node[above right] {};
    \draw[fill=black] (-4,4) circle (0.12) node[above right] {};
    \draw[fill=black] (-3,4) circle (0.12) node[above right] {};
    \draw[fill=black] (-4,5) circle (0.12) node[above right] {};
    \end{tikzpicture}
    \caption{$\tilde{\mathcal{P}}^{(0)}$-polygon. }
    \label{fig-g4-p-tilde-0}
 \end{minipage}\hfill
  \begin{minipage}{0.48\textwidth}
     \centering
     \begin{tikzpicture}[scale = 0.5]
        \draw[thin,dotted] (0,-3) grid (-4,5);
        \draw[->] (-4,0) -- (1,0) node[right] {$k$};
        \draw[->] (0,-3) -- (0,5.8) node[above] {$m$};
        \foreach \x/\xlabel in {-4/{-4\hphantom{-}}, -3/{-3\hphantom{-}}, -2/{-2\hphantom{-}}, -1/{-1\hphantom{-}}}
    \draw (\x cm,1pt ) -- (\x cm,-1pt ) node[anchor=north,fill=white] {\xlabel};
  \foreach \y/\ylabel in { -3/{-3\hphantom{-}}, -2/{-2\hphantom{-}}, -1/{-1\hphantom{-}}, 1/1, 2/2, 3/3, 4/4, 5/5}
    \draw (1pt,\y cm) -- (-1pt ,\y cm) node[anchor=west, fill=white] {\ylabel};
    \draw[thick, orange] (-4,4) -- (-4,-1);
    \draw[thick, orange] (-3,3) -- (-3,-2);
    \draw[thick, orange] (-2,2) -- (-2,-3);
    \draw[thick, orange] (-1,1) -- (-1,-3);
    \draw[thick, orange] (0,0) -- (0,-3);
    \draw[fill=black] (-2,-3) circle (0.12) node[above right] {};
    \draw[fill=black] (-1,-3) circle (0.12) node[above right] {};
    \draw[fill=black] (0,-3) circle (0.12) node[above right] {};
    \draw[fill=black] (-3,-2) circle (0.12) node[above right] {};
    \draw[fill=black] (-2,-2) circle (0.12) node[above right] {};
    \draw[fill=black] (-1,-2) circle (0.12) node[above right] {};
    \draw[fill=black] (0,-2) circle (0.12) node[above right] {};
    \draw[fill=black] (-4,-1) circle (0.12) node[above right] {};
    \draw[fill=black] (-3,-1) circle (0.12) node[above right] {};
    \draw[fill=black] (-2,-1) circle (0.12) node[above right] {};
    \draw[fill=black] (-1,-1) circle (0.12) node[above right] {};
    \draw[fill=black] (0,-1) circle (0.12) node[above right] {};
    \draw[fill=black] (-4,0) circle (0.12) node[above right] {};
    \draw[fill=black] (-3,0) circle (0.12) node[above right] {};
    \draw[fill=black] (-2,0) circle (0.12) node[above right] {};
    \draw[fill=black] (-1,0) circle (0.12) node[above right] {};
    \draw[fill=black] (0,0) circle (0.12) node[above right] {};
    \draw[fill=black] (-4,1) circle (0.12) node[above right] {};
    \draw[fill=black] (-3,1) circle (0.12) node[above right] {};
    \draw[fill=black] (-2,1) circle (0.12) node[above right] {};
    \draw[fill=black] (-1,1) circle (0.12) node[above right] {};
    \draw[fill=black] (-4,2) circle (0.12) node[above right] {};
    \draw[fill=black] (-3,2) circle (0.12) node[above right] {};
    \draw[fill=black] (-2,2) circle (0.12) node[above right] {};
    \draw[fill=black] (-4,3) circle (0.12) node[above right] {};
    \draw[fill=black] (-3,3) circle (0.12) node[above right] {};
    \draw[fill=black] (-4,4) circle (0.12) node[above right] {};
    \end{tikzpicture}
    \caption{$\tilde{\mathcal{Q}}^{(0)}$-polygon.  }
    \label{fig-g4-q-tilde-0}
  \end{minipage}
\end{figure}

{\bf III.} Consider a product of $1/\tilde{\cQ}^{(0)}$ and $\tilde{\cP}^{(0)}$. Due to {\bf (c)} the term of maximal power $k$ in $\tilde{\cP}^{(0)}$ is $c^0$. Subleading terms are $c^{-n}$ with $1\leq n\leq 4$. Then, 
\be
\label{g4-0}
g_4^{(0)} = \tilde{\cP}^{(0)}/\tilde{\cQ}^{(0)} = (c^{0} + .. ) \sum_{n = 0}^{\infty}  \left( \tilde{\cQ}^{(0)}\right)_{n} c^{-n} = \sum_{n = 0}^{\infty} \left(g_4^{(0)} \right)_n c^{-n} \,,
\ee 
and the principal part contains no poles.

\vspace{2mm} 

Further analysis of $g_4^{(\alpha)}$ is similar to that one for $g_2^{(\alpha)}$ because the structure of the  Puiseux series obtained here, see \eqref{g4_P1}, \eqref{g4_P1-2}, \eqref{g4-1}, \eqref{g4-0} for each domain of $\alpha$ is the same. Thus,  all the conclusions from Appendix   \bref{app:second}   apply here without any modifications. 

\section{Proof of Proposition \bref{corollary31}}
\label{app:proof_prop}

Consider an $n$-th order coefficient $g_n$  of the logarithmic conformal block \eqref{general-log-f} and suppose that $h_i = \gamma_i p$ and $\tilde{h} = \tilde{\gamma} p$, where $\gamma_i, \tilde{\gamma} = const$ and $p \in [0, \infty)$ is a common parameter which defines the way the conformal dimensions go to infinity. In general, to reach points at infinity in the space of conformal dimensions ($\Delta$-space) one  can choose different paths, $h_i = h_i(p)$ and $\tilde{h} = \tilde{h}(p)$, where $p$ is  a proper parameter along a given path.  Here, we fixed  a  linear dependence which will finally correspond to considering classical blocks. Other choices, e.g. $h_{1,2} = O(p^2)$, $h_{3,4}  =  O(p)$, $\tilde{h} = O(p)$, which define different rates of approaching the infinity, would correspond to heavy-light conformal blocks.  

Let us represent  $g_n$ in the following form:
\be
g_n(h_i, \tilde{h}, c) = \frac{p_n(h_i, \tilde{h}, c)}{q_n(\tilde{h},c)} =  \frac{\bar{p}_n(\gamma_i p/c, \tilde{\gamma} p/c, c)}{\bar{q}_n(\tilde{\gamma} p/c, c)} \equiv g_n(\gamma_i, \gamma \,|\, p,c) \,,
\ee
which is obviously achieved by multiplying and dividing each conformal dimension by $c$. In doing so we see that  the resulting functions $\bar{p}_n=p_n(c,p)$ and $\bar{q}_n=q_n(c,p)$ are  polynomials on the 2-plane $(p,c)$. This reduction of  the number of independent variables is due to  the fact that we have chosen a particular path in the $\Delta$-space.   Also, we assume that the Newton polygons for $p_n$, $q_n$ have the properties $\bold{(a)}-\bold{(e)}$ formulated  in Appendix \bref{app:forth}, for all $n\in \mathbb{N}$ (this is needed to control their asymptotic expansions).

1. Let us prove  the first part of Proposition \bref{corollary31}. In general, the asymptotic behaviour of the logarithmic conformal block depends on how the variables approach infinity. Let variables $(c,p)$ tend to infinity along a particular path $p = p(c)$ on the $(c,k)$-plane.  E.g., a path  $p = c^{\alpha}$ with exponent  $\alpha \in (0,1)_\rn$ gives a logarithmic conformal  block $f^{(\alpha)}$ for $\alpha$-heavy operators \eqref{alpha_lor}, which forms for two different $\alpha$'s are governed by Proposition \bref{prop1}. Our goal here is to show that despite the fact that both the Zamolodchikov and $\alpha$-heavy classical  conformal blocks  are obtained by choosing different paths on the $(c,p)$-plane,  they can be performed as series for different variables, but with the same expansion coefficients. 

To this end, we expand the coefficients $g_n=g_n(p,c)$ near $c=\infty$ and $p/c=0$. For convenience, let us switch to a variable $k \equiv p/c$. Following Appendix \bref{app:A22}, the coefficient $g_n$ can be expanded  around this point by  using  \eqref{g4_xy-2}, where one assigns  $x := c$ and $y := k$,  
\be
\label{expansion-gn}
 g_n(\gamma_i, \gamma \,|\, kc,c) =\sum_{m=1}^{\infty}A^{(n)}_{-1,m}(\gamma_i, \tilde{\gamma}) c^1 k^m + \sum_{l=0}^{\infty}\sum_{m=-l}^{\infty}B^{(n)}_{l,m}(\gamma_i, \tilde{\gamma}) c^{-l}k^m \,,
\ee
where $A^{(n)}_{-1,m}$ and $B^{(n)}_{l,m}$ are some expansion coefficients. If one chooses a path $k = k_0 + o(c^0) $  for some fixed $k_0$, then, since $h = O(p) = O(c)$ (i.e. $h_i = \gamma_i p$, $\tilde h = \tilde \gamma p$,  and $p = kc = k_0 c + o(c^0)$), the  leading term $\sim c$ in \eqref{expansion-gn} is by definition the Zamolodchikov classical block. Therefore, we will focus on the first  sum on  the right-hand side.

By construction, $k$ is always accompanied by  $\gamma_i$ or $\tilde{\gamma}$  which implies that $A^{(n)}_{-1,m}(\gamma_i, \tilde{\gamma})$ is a homogeneous function of degree $m$. Since $\bar{p}_n, \bar{q}_n$ are Laurent polynomials in $\tilde{\gamma}$, while $\bar{p}_n$ is also a polynomial in $\gamma_i$, then  $A^{(n)}_{-1,m}(\gamma_i, \tilde{\gamma})$ inherits the same properties in $\gamma_i, \tilde{\gamma}$. Thus, summing over all $n$ and defining $\sum_{n=1}^{\infty} z^n A^{(n)}_{-1,m} = A_{-1,m}$, one obtains 
\be
\label{expansion-with-k}
\sum_{n=1}^\infty z^n g_n(\gamma_i, \tilde \gamma \,|\, kc,c) \simeq  c^1\sum_{m=1}^{\infty} A_{-1,m}(\gamma_i, \tilde{\gamma}|z) k^m\,,
\ee
where $\simeq$ means that we keep only a leading order term in $c$. The coefficient $A_{-1,m}$ here is a homogeneous function of $\gamma_i, \tilde{\gamma}$ of degree $m$, a Laurent series in $\tilde{\gamma}$, and a Taylor series in $\gamma_i$.

Let the series in the right-hand side of \eqref{expansion-with-k} have a radius of convergence $R>1$. If not, i.e. $R\leq1$, then due to homogeneity of $A_{-1,m}(\gamma_i, \tilde{\gamma})$ one can rescale $\gamma_i, \tilde{\gamma} $ by the factor of $r<R$. {So, from the beginning one assumes that all $\gamma_i, \tilde{\gamma}$ are smaller than the initial ones by the factor of $r$ to make $R>1$}. It follows that one can choose  $k = 1 $ for which the series converges. Now, since $k$ is a finite constant, i.e. $h =  O(p) = O(c)$, then the series defines the Zamolodchikov classical block. {Since the expansion \eqref{expansion-with-k} is around $k = 0$, i.e. all classical dimensions $h/c = \gamma k$ are small,  then one obtains that the right-hand side there defines a perturbative expansion of the classical block.} Thus, 
\be
\label{A-to-P}
A_{-1,m}(\gamma_i, \tilde{\gamma}|z)c^1k^m\,\big|_{k=1}\, =\, c^1P_m\left(\gamma_i, \tilde{\gamma}\middle|z\right)\,,
\ee
where $P_m$ are defined by  \eqref{perturbative-f1-small}.  {We emphasize that here we call classical dimensions $h/c = \gamma k$ small in a sense that a common parameter $k$ lies inside the radius of convergence around $k=0$. In the same way one should treat smallness of $\delta_i, \tdelta$ in \eqref{perturbative-f1-small}.

On the other hand, let $k$ be $c$-dependent:  $k = c^{-\varkappa}$, where $\varkappa\in(0,1)_\rn$. Then,  the series \eqref{expansion-with-k} converges for $c >R^{-1/\varkappa}$. Since $h = O(p) = O(c^{1-\varkappa})$ (i.e. $h_i = \gamma_i p$, $\tilde h = \tilde \gamma p$,  and $p = kc = c^{1-\varkappa}$) the series by construction defines an  expansion of the $\alpha$-heavy logarithmic conformal block $f^{(\alpha)}$ near $c = \infty$ for $\alpha=1-\varkappa$.  For a given $m\in \mathbb{N}$ and for all $\alpha \in (\frac{m-1}{m}, 1)_\rn$, the coefficients $A_{-1,m}$ are exactly  $f^{(\alpha)}_{m\alpha-m+1} = f^{(1-\varkappa)}_{1-m\varkappa}$ which define the principal part $\pr f^{(\alpha)}$, see Proposition \bref{prop1}. Thus,    
\be
\label{A-to-f}
A_{-1,m}(\gamma_i, \tilde{\gamma}|z)c^1k^m\,\big|_{k = c^{-\varkappa}}\, =\, f_{1-m\varkappa}^{(1-\varkappa)}\left(\gamma_i, \tilde{\gamma}\middle|z\right) c^{1-m\varkappa}\,.
\ee
From \eqref{A-to-P} one finds that      
\be
\forall \varkappa \in (0, 1/m)_\rn:\qquad  P_m\left(\gamma_i, \tilde{\gamma}\middle|z\right) = f_{1-m\varkappa}^{(1-\varkappa)}\left(\gamma_i, \tilde{\gamma}\middle|z\right)\,.
\ee
This  is exactly \eqref{prop3.1_1} provided that $\alpha = 1-\varkappa$ and $\gamma_i:= \delta_i$, $\tilde \gamma := \tdelta$. 
 
2. The relation \eqref{prop3.1_2} from the second part of the  Proposition is proved similarly (to this end, one uses the expansion \eqref{g4_xy}).

\section{Fusion rules for semi-degenerate operators}
\label{app:W3-deg}

In this section we discuss constraints imposed on the 3-point correlation function when two of three  operators are semi-degenerate.  To this end, write down the Ward identities in the spin-3 sector for correlation functions of arbitrary $N$ primary operators:
\be
\label{ward-W3}
\ba{l}
\dps
\sum_{j=0}^N W_{-2}^{(j)}\Big\langle\prod_{i=1}^N \cO_i\left(z_i\right)\Big\rangle=0\,, 
\vspace{2mm} 
\\
\dps
\sum_{j=0}^N\left(z_j W_{-2}^{(j)}+W_{-1}^{(j)}\right)\Big\langle\prod_{i=1}^N \cO_i\left(z_i\right)\Big\rangle=0\,,
\vspace{2mm} 
\\
\dps
\sum_{j=0}^N\left(z_j^2 W_{-2}^{(j)}+2 z_j W_{-1}^{(j)}+w_j\right)\Big\langle\prod_{i=1}^N \cO_i\left(z_i\right)\Big\rangle=0\,, 
\vspace{2mm} 
\\
\dps
\sum_{j=0}^N\left(z_j^3 W_{-2}^{(j)}+3 z_j^2 W_{-1}^{(j)}+3 z_j w_j\right)\Big\langle\prod_{i=1}^N \cO_i\left(z_i\right)\Big\rangle=0\,, 
\vspace{2mm} 
\\
\dps
\sum_{j=0}^N\left(z_j^4 W_{-2}^{(j)}+4 z_j^3 W_{-1}^{(j)}+6 z_j^2 w_j\right)\Big\langle\prod_{i=1}^N \cO_i\left(z_i\right)\Big\rangle=0\,,
\ea
\ee
where the notation $W^{(j)}_{-n}\,\big\langle\prod_{i=1}^N \cO_i\left(z_i\right)\big\rangle \equiv  \big\langle \cO_1(z_1)...W_{-n}\cO_{j}(z_j) ...\cO_{N}(z_N) \big\rangle$ is introduced. This is the algebraic equation system for the $N$-point correlation functions of  secondary operators. Unlike the Virasoro case, such secondary correlation functions $W^{(j)}_{-n}\big\langle\prod_{i=1}^N \cO_i\left(z_i\right)\big\rangle$ can be independent of $\big\langle\prod_{i=1}^N \cO_i\left(z_i\right)\big\rangle$.\footnote{In the Virasoro case, the Ward identities allow one to express arbitrary  correlation functions of secondaries in terms of those for primaries. } 

Let us  consider the primary correlation function $\big\langle \cO_1(z_1)\cO_2(z_2)\cO_{3}(z_3) \big\rangle$. The spin-2 Ward identities fix  the $z$-dependence up to structure constants $C_{123}$ in the standard fashion as  
\be
\label{three-point-primary}
\big\langle \cO_1(z_1)\cO_2(z_2)\cO_{3}(z_3) \big\rangle  = \frac{C_{123}}{ z_{12}^{h_1+h_2-h_3} z_{23}^{h_2+h_3-h_1}z_{13}^{h_1+h_3-h_2}}\,.
\ee
To identify further limitations on the 3-point function one solves  \eqref{ward-W3} which is the linear system for 6 unknown quantities $W_{-1}^{(j)} \left\langle\prod_{i=1}^3 \cO_i\left(z_i\right)\right\rangle$ and $W_{-2}^{(j)} \left\langle\prod_{i=1}^3 \cO_i\left(z_i\right)\right\rangle$. Obviously, 5 equations in \eqref{ward-W3} are not sufficient to find 6 independent variables. This ambiguity can be fixed by  introducing  (semi-) degenerate operators. 

E.g. if $\cO_3$ is semi-degenerate \eqref{semidegenerate1}, then the secondary  correlation function of $W_{-1}\cO_3$ is expressed through the primary correlation function of $\cO_3$ as follows (see e.g. \cite{Fateev:2007ab})
\be
\label{W-1}
\big\langle \cO_1(z_1) \cO_2(z_2) W_{-1}\cO_3(z_3) \big\rangle =  \frac{3 w_3}{2 h_3} \cL_{-1}^{(3)}\big\langle \cO_1(z_1) \cO_2(z_2) \cO_3(z_3) \big\rangle \,,
\ee 
where $\cL_{-1}^{(i)} = \partial_{z_i}$ which acts on the 3-point function  \eqref{three-point-primary} as
\be
\cL_{-1}^{(i)}\big\langle \cO_i(z_i) \cO_j(z_j) \cO_k(z_k) \big\rangle =   -\left(\frac{h_i+h_j-h_k}{z_i-z_j}+\frac{h_i+h_k-h_j}{z_i-z_k}\right)\big\langle \cO_i(z_i) \cO_j(z_j) \cO_k(z_k)  \big\rangle\,.
\ee
The relation \eqref{W-1} reduces the number of independent variables by one, so that there are 5 equations for 5 variables. Applying  Ward identities \eqref{ward-W3} one  obtains a system of 5 inhomogeneous linear equations: 
{\small
\be
\label{ward-three-point-one-deg}
\begin{bmatrix}
 0 & 0 & 1 & 1 & 1 \\[1mm]
 1 & 1 & z_1 & z_2 & z_3 \\[1mm]
 2 z_1 & 2 z_2 & z_1^2 & z_2^2 & z_3^2 \\[1mm]
 3 z_1^2 & 3 z_2^2 & z_1^3 & z_2^3 & z_3^3 \\[1mm]
 4 z_1^3 & 4 z_2^3 & z_1^4 & z_2^4 & z_3^4 \\[1mm]
\end{bmatrix}
\begin{bmatrix}
    \braket{W_{-1}\cO_1(z_1) \cO_2(z_2) \cO_3(z_3)} \\[1mm]
    \braket{\cO_1(z_1) W_{-1}\cO_2(z_2) \cO_3(z_3)} \\[1mm]
    \braket{W_{-2}\cO_1(z_1) \cO_2(z_2) \cO_3(z_3)} \\[1mm]
    \braket{\cO_1(z_1) W_{-2}\cO_2(z_2) \cO_3(z_3)} \\[1mm]
    \braket{\cO_1(z_1) \cO_2(z_2) W_{-2}\cO_3(z_3)} \\[1mm]
\end{bmatrix}
= 
\boldsymbol{v}\,,
\ee}
where
{\small 
\be
\boldsymbol{v}=
\begin{bmatrix}
    0 \\
    \frac{3 w_3}{2 h_3} \cL_{-1}^{(3)}  \\
    w_1 + w_2 + w_3 + \frac{3 w_3 \cL_{-1}^{(3)}}{h_3}z_3   \\
    3(w_1 z_1+ w_2 z_2+ w_3 z_3)+ \frac{9 w_3 \cL_{-1}^{(3)}}{2 h_3}z_3   \\
    6(w_1 z_1^2+ w_2 z_2^2+ w_3 z_3^2)+ \frac{6 w_3 \cL_{-1}^{(3)}}{ h_3}z_3  
\end{bmatrix}
\big\langle \cO_1(z_1) \cO_2 (z_2) \cO_3(z_3) \big\rangle \,.
\ee}
It has a unique solution  provided that the determinant is not zero:
\be
\det = -\left(z_{12}\right){}^4 \left(z_{13}\right){}^2 \left(z_{23}\right){}^2 \neq 0 \,.
\ee
In this case one can express the 3-point functions of $W_{-1}$ and $W_{-2}$ descendants in terms of the 3-point function of primary operators. 

Making one more operator semi-degenerate allows one to reduce a number of independent variables by one  so that  there are now 5 equations for 4 variables and the system becomes overdetermined. However, notice that the 3-point structure constants \eqref{three-point-primary} are still arbitrary. Then,  the Ward identities impose restrictions on the structure constants known as the fusion rules. Let $\cO_2$ and $\cO_3$ be semi-degenerate and $\cO_1$ be arbitrary. Then applying Ward identities \eqref{ward-W3}, one obtains the  system of 5 homogeneous linear equations:
{\small
\be
\label{ward-three-point-two-deg}
\begin{bmatrix}
 0 & 0 & 1 & 1 & 1 \\
 \frac{3 w_2 \cL_{-1}^{(2)}}{2 h_2}+\frac{3 w_3 \cL_{-1}^{(3)}
   }{2 h_3} & 1 & z_1 & z_2 & z_3 \\
 \frac{3 w_2 \cL_{-1}^{(2)}}{h_2}z_2+\frac{3 w_3
  \cL_{-1}^{(3)}}{h_3}z_3+w_1+w_2+w_3 & 2 z_1 & z_1^2 & z_2^2 & z_3^2 \\
 \frac{9 w_2  \cL_{-1}^{(2)}}{2 h_2}z_2^2+\frac{9 w_3
   \cL_{-1}^{(3)}}{2 h_3} z_3^2+3( w_1 z_1+ w_2 z_2+ w_3 z_3) & 3 z_1^2 & z_1^3 & z_2^3 & z_3^3
   \\
 \frac{6 w_2 \cL_{-1}^{(2)}}{h_2}z_2^3+\frac{6 w_3 
   \cL_{-1}^{(3)}}{h_3}z_3^3+6( w_1 z_1^2+ w_2 z_2^2+ w_3 z_3^2) & 4 z_1^3 & z_1^4 & z_2^4 &
   z_3^4 \\
\end{bmatrix}
\begin{bmatrix}
    \braket{\cO_1(z_1) \cO_2(z_2) \cO_3(z_3)} \\[1.8mm]
    \braket{W_{-1}\cO_1(z_1) \cO_2(z_2) \cO_3(z_3)} \\[1.8mm]
    \braket{W_{-2}\cO_1(z_1) \cO_2(z_2) \cO_3(z_3)} \\[1.8mm]
    \braket{\cO_1(z_1) W_{-2}\cO_2(z_2) \cO_3(z_3)} \\[1.8mm]
    \braket{\cO_1(z_1) \cO_2(z_2) W_{-2}\cO_3(z_3)} 
\end{bmatrix}
= 0\,,
\ee}
It has non-trivial solutions provided the respective  determinant   
\be
\det = \frac{2 h_2 h_3  w_1  - h_3\left(-3 h_1+h_2+3 h_3\right) w_2-h_2 \left(-3 h_1+3 h_2+h_3\right) w_3}{2 h_2 h_3}\left(z_{12}\right){}^3 \left(z_{13}\right){}^3 \left(z_{23}\right)^1
\ee
is zero for any three points $z_1$, $z_2$, $z_3$ which means that the 3-point function $\braket{\cO_1(z_1) \cO_2(z_2) \cO_3(z_3)}$ is zero unless  
\be
\label{det_cond}
 w_1 = \frac{w_2}{h_2}\left(h_2+3(h_3-h_1)\right) + \frac{w_3}{h_3} \left(h_3+3(h_2 - h_1)\right)\,.
\ee
Note that \eqref{det_cond}  has  the obvious permutation symmetries  
$$
\ba{l}
\dps
(h_1, w_1;\, h_2, w_2;\, h_3,w_3) \leftrightarrow (h_1, w_1;\, h_3, w_3;\, h_2, w_2)
\vspace{2mm}
\\
\dps
(h_1, w_1;\, h_2, w_2;\,  h_3,w_3) \leftrightarrow (h_1, -w_1;\, h_3, -w_3;\, h_3, -w_2)
\ea
$$
and  $(h_2, w_2, h_1, w_1) \leftrightarrow (h_3, -w_3, h_1,  -w_1)$.  The equation  \eqref{det_cond} defines the fusion rules which, in particular, include the following:  
\begin{itemize}
\item if $\cO_2 = \cO_3$, i.e. $(h_3, w_3) = (h_2, w_2)$,  then: 
\be
\label{w12}
w_1 = \left(4-\frac{3 h_1}{h_2}\right)w_2\,, 
\qquad
\forall h_1\,;
\ee
\item if $\cO_3 = \cO_2^*$, i.e  $(h_3, w_3) = (h_2, -w_2)$, then:
\be
\label{w123}
w_1 = 0\,, 
\qquad
\forall h_1\,;
\ee
\item there exist $(h_1,w_1)$ which  satisfy \eqref{det_cond} such that \eqref{degW} holds, i.e. $\cO_1$ is semi-degenerate. E.g., if $\cO_2 = \cO_3$, then one can set $h_1 = h_2$, which leads to $w_1 = w_2$ \eqref{w12}, i.e. $\cO_1 = \cO_2 = \cO_3$. Another example is  $\cO_3 = \cO_2^*$:  one can set $h_1 = 0$, then we obtain a semi-degenerate operator $\cO_1 = \boldsymbol{1}$ \eqref{w123}.
    
\item if one of semi-degenerate operators is an identity operator, e.g., $\cO_3 = \boldsymbol{1}$, i.e. $(h_3,w_3) = (0,0)$, then one obtains a simple condition which arises when considering Ward identities for a 2-point function \cite{Fateev:2007ab}: 
\be
\label{W3-two-point-cond}
h_1 = h_2\,, \qquad w_1 = - w_2 \,.
\ee
In other words, the 3-point function $\braket{\cO_1(z_1) \cO_2(z_2) \cO_3(z_3)}$, where  $\cO_3 = \boldsymbol{1}$, goes to the 2-point function $\braket{\cO_1(z_1) \cO_2(z_2)} = \delta_{h_1-h_2,0}\delta_{w_1+w_2,0}\, z_{12}^{-2h_1}$, i.e. $\cO_2 = \cO_1^*$.  One can check that choosing $\cO_1 = \boldsymbol{1}$ or $\cO_2 = \boldsymbol{1}$ reduces the zero-determinant  condition \eqref{det_cond}  to $h_i = h_j$ and $w_i = w_j$ that is equivalent to going to the respective 2-point function.    
\end{itemize}

\noindent In Section  \bref{sec:mat_block} we consider  conformal blocks of operators \eqref{W3opers} that  are consistent with the fusion rules discussed above.

\section{The first coefficient $B_1$ of the $\cW_3$ conformal block}
\label{app:B1}
\subsection{Explicit expression of the first coefficient $B_1$}

We explicitly calculate the coefficient $B_1$  \eqref{blockW} for semi-degenerate $\cO_{2,3}$ and arbitrary $\cO_{1,4}$:
\be
B_1  =\sum_{|Y| = |Y'| = 1}\Gamma^*(Y)Q^{-1}(Y,Y') \Gamma(Y') \;.
\ee
The matrix elements $\Gamma, \Gamma^*$ and the Gram matrix $Q$ are given by 
\begin{equation}
\begin{aligned}
& \Gamma\left(L_{-1}\right)= \tilde{h}+h_3-h_4, \qquad \Gamma\left(W_{-1}\right)= \tilde{w}+2 w_3-w_4+\frac{3 w_3}{2 h_3}\left( \tilde{h}-h_3-h_4\right)  \\
& \Gamma^*\left(L_{-1}\right)= \tilde{h}+h_2-h_1, \qquad \Gamma^*\left(W_{-1}\right)= \tilde{w}+w_2+w_1-\frac{3 w_2}{2 h_2}\left( \tilde{h}+h_2-h_1\right) \\
\end{aligned}
\end{equation}
\begin{equation}
    \begin{array}{|c||c|c|}
        \hline  Q\left(Y, Y'\right) & L_{-1} & W_{-1}  \\
        \hline  L_{-1} & 2 \tilde{h} & 3 \tilde{w} \\
         \hline W_{-1} & 3 \tilde{w} &  \frac{9}{2} \left(\frac{32}{22+5c}\left(\tilde{h}+\frac{1}{5}\right)-\frac{1}{5}\right)\tilde{h} \\
         \hline
    \end{array}
\end{equation} 
Finally, 
\be
\label{general_B1}
\begin{aligned}
& B_1 = \left(\frac{\tilde{h} (c-32 \tilde{h}-2) (-h_1+h_2+\tilde{h})}{2
   \left(\tilde{h}^2 (c-32 \tilde{h}-2)+(5 c+22) \tilde{w}^2\right)}+\frac{3 \tilde{w} \left(w_1+w_2-\frac{3 w_2 (-h_1+h_2+\tilde{h})}{2 h_2}+\tilde{w}\right)}{\frac{9 \tilde{h}^2 (c-32 \tilde{h}-2)}{5 c+22}+9 \tilde{w}^2}\right)  \\
   & \times (h_3-h_4+\tilde{h}) + \frac{ (w_2 \tilde{h} (-3 h_1+h_2+3 \tilde{h})-2 h_2 w_1 \tilde{h}+h_2 \tilde{w} (-3 h_1+3 h_2+\tilde{h}))}{18 h_2 h_3 \left(\tilde{h}^2 (c-32 \tilde{h}-2)+(5 c+22) \tilde{w}^2\right)} \\
   & \times (5 c+22)(-2 h_3 w_4+w_3 (h_3-3 h_4+3 \tilde{h})+2 h_3 \tilde{w}) \, .
\end{aligned}
\ee
As a consistency check, one can choose  e.g. $\cO_1  = {\bm {1}}$ that yields the 3-point function expanded near $z=0$, i.e. $  B_1 = (\tilde h - h_3+h_4)$ (since $[\cO_3][\cO_4] = [\tilde\cO]+...$, then, according to \eqref{W3-two-point-cond}, one imposes the fusion constraints $\tilde h = h_2, \tilde{w} = -w_2 $).\footnote{Also, \eqref{general_B1} reproduces the $c=2$ block coefficient calculated  in \cite{Mironov:2009by}.} 

Even though the block coefficient \eqref{general_B1} is manifestly rational function, one should keep in mind that $(h_2,w_2)$ and $(h_3, w_3)$ are subjected to the constraint \eqref{degW} and substituting e.g. $w_{2,3} = w_{2,3}(h_{2,3})$ will produce square roots. Nevertheless, one can consider particular dimensions and charges in order to keep  $B_1$ rational even in terms of independent variables. E.g., let  all external operators  be semi-degenerate and   $\cO_1 = \cO_3 = \cO_2^* = \cO_4^*$. Then, according to the fusion rules of Appendix \bref{app:W3-deg}, one sets $\tilde{w} = 0$. Denoting $h_i \equiv h$ one finds out that  $B_1$ is a rational function of conformal dimensions $h,\tilde h$ and $c$:
\be
\label{B1_app}
B_1  = \tilde{h}\frac{16(\tilde{h}+h) +2-c }{32\tilde{h} +2-c}\,.
\ee
Choosing  a different order, $\cO_2 = \cO_3 = \cO_1^* = \cO_4^*$ (keeping $\tilde{w} = 0$), one finds the similar result,
\be
\label{B1_2}
B_1  =\frac{16(\tilde{h} - h) }{32\tilde{h}+2-c}\,.
\ee
The denominator is the same for both expressions since it is defined by the Kac determinant. 
%

\subsection{On different sign conventions}
\label{app-signs}

There are different sign conventions for the Hermitian conjugation of $W$-generators: 
\be
\label{Wpm}
W^\dagger_n = \pm W_{-n}\,.
\ee 
In this paper, we adhere to somewhat standard convention with the minus sign \eqref{W3-conj} (see e.g. \cite{Fateev:2007ab}), but the plus sign can  also be used (see e.g. \cite{Kanno:2010kj}). Here, we show that the conformal block coefficients calculated using the two sign conventions \eqref{Wpm}  remain the same.    

Consider an $n$-th order term in \eqref{blockW},  
\be
\label{block_term}
z^n:\qquad \langle h_1, -w_1| \cO_2(1) Y|\tilde{h}, \tilde{w}\rangle\; Q^{-1}_{Y, Y'}(\tilde{h}, \tilde{w})\; \langle \tilde{h}, \tilde{w} |Y'^\dagger \cO_{3}(1) | h_4, w_4 \rangle\,,
\ee
where $n = |Y| = |Y'|$, and $Q^{-1}_{Y,Y'}(\tilde{h}, \tilde{w})$ is the inverse Gram matrix. Let $m_{Y}$ and $m_{Y'}$ be numbers of $W$-generators in the basis monomials $Y= L_{-\lambda_1}...L_{-\lambda_k}W_{-\mu_1}...W_{-\mu_{m_{Y}}}$ and $Y' = L_{-\lambda_1}...L_{-\lambda_{k'}}W_{-\mu_1}...W_{-\mu_{m_{Y'}}} $. Also, choosing different sign conventions we label  respective quantities by $(\pm)$, e.g. 
\be
\label{Ypm}
(Y)_{(-)} = (Y)_{(+)}\,, \qquad (Y'^\dagger)_{(-)} = (-)^{m_{Y'}} (Y'^{\dagger})_{(+)} \,.
\ee
One finds that the  Gram matrices calculated by using the two sign conventions are related as  
\be
(Q_{Y, Y'})_{(-)} = \langle \tilde{h}, \tilde{w} | (Y^\dagger)_{(-)} (Y')_{(-)} |\tilde{h}, \tilde{w} \rangle = (-)^{m_{Y}} (Q_{Y, Y'})_{(+)}\,.
\ee
Here, every row is multiplied by the overall sign $(-)^{m_{Y}}$, whence $\det Q_{(-)} = (-)^M \det Q_{(+)}$, where $M=\sum\limits_{|Y| = n} m_{Y}$. Now, let $A_{Y, Y'}$ be a cofactor of $Q_{Y,Y'}$. Then, one obtains
\be
\label{Qpm}
(Q^{-1}_{Y,Y'})_{(-)} = \frac{(A_{Y',Y})_{(-)}}{(\det Q)_{(-)}} =  \frac{(-1)^M}{(\det Q)_{(+)}} (-)^{M-m_{Y'}} (A_{Y',Y})_{(+)} =  (-)^{m_{Y'}} (Q_{Y,Y'}^{-1})_{(+)}\,,
\ee
where the sign $(-)^{M-m_{Y'}}$ comes from $A_{Y', Y}$ because  defining this cofactor one deletes a row which corresponds to $Y'$. Combining  \eqref{Ypm}  and \eqref{Qpm}  one finds that the conformal block coefficients \eqref{block_term} are independent of the sign conventions \eqref{Wpm}.

\section{The logarithmic  $\cW_3$ conformal block coefficients}
\label{app:W3}

Here, we analyze the large-$c$ expansion of the first three coefficients of the logarithmic $\cW_3$ conformal block, $G_{n}^{(\alpha)}$, $n=1,2,3$ with $\alpha\in \mathbb{Q}^+_0$, see \eqref{def-log-W3a}. As we discussed in the previous section, the $\cW_3$ block coefficients for  particular fields are still rational functions and, therefore, rescaling dimensions and charges  and expanding the resulting coefficients  near $c=\infty$ one can apply the same methods discussed in Appendix \bref{app:second}.

The first coefficient is   $G_1^{(\alpha)} = P_1^{(\alpha)}/Q_1^{(\alpha)}$, where 
\be
P_1^{(\alpha)} = A_{0,2}c^{2\alpha} +A_{1,1}c^{1+\alpha} + A_{0,1}c^{\alpha} \;, \qquad Q_1^{(\alpha)} = a_{0,1}c^\alpha + a_{1,0}c + a_{0,0}\;,
\ee
where $A_{k,m}, a_{k,m}$ are polynomials of classical dimensions, cf. \eqref{B1_app}.  The second coefficient is given by $G_2^{(\alpha)} = P_2^{(\alpha)}/Q_2^{(\alpha)}$, where 
\be
P_2^{(\alpha)} = \sum_{k=1}^3 D_{k,1}c^{k+\alpha} + \sum_{m=2}^5 \sum_{k=0}^{4-m} D_{k,m}c^{k+m\alpha}\;, \quad Q^{(\alpha)}_2 = \sum_{k=1}^3 d_{k,0} c^{k} + \sum_{m=1}^{4} \sum_{k=0}^{4-m} d_{k,m}c^{k+m\alpha} \;,
\ee
where $D_{k,m}, d_{k,m}$ are polynomials of classical dimensions. The third coefficient is given by  $G_3^{(\alpha)} = P_3^{(\alpha)}/Q_3^{(\alpha)}$, where  
\be
P_3^{(\alpha)} = \sum_{k=1}^6 T_{k,0}c^{k} + \sum_{m=1}^4 \sum_{k=0}^{8} T_{k,m}c^{k+m\alpha}+\sum_{m=5}^{12}\sum_{k=0}^{12-m}T_{k,m}c^{k+m\alpha}\;, 
\ee
\be 
 Q^{(\alpha)}_3 = \sum_{k=1}^8 t_{k,0}c^{k} + \sum_{m=1}^3 \sum_{k=0}^{8} t_{k,m}c^{k+m\alpha}+\sum_{m=4}^{11}\sum_{k=0}^{11-m}t_{k,m}c^{k+m\alpha}  \;,
\ee
where $T_{k,m}, t_{k,m}$ are polynomials  of classical dimensions. 
\begin{figure}[h!]
    \begin{minipage}{0.48\textwidth}
    \centering
     \begin{tikzpicture}[scale = 0.5]
        \draw[thin,dotted] (0,0) grid (2,2);
        \draw[->] (0,0) -- (2.4,0) node[right] {$k$};
        \draw[->] (0,0) -- (0,2.4) node[above] {$m$};
        \foreach \x/\xlabel in { 1/1, 2/2}
    \draw (\x cm,1pt ) -- (\x cm,-1pt ) node[anchor=north,fill=white] {\xlabel};
  \foreach \y/\ylabel in { 1/1, 2/2}
    \draw (1pt,\y cm) -- (-1pt ,\y cm) node[anchor=east, fill=white] {\ylabel};
     \draw[thick, red] (1,1) -- (0,2);
        \draw[thick, red] (1,1) -- (0,1);
        \draw[thick, red] (0,1) -- (0,2);
    \draw[fill=black] (0,2) circle (0.12) node[above right] {};
    \draw[fill=black] (0,1) circle (0.12) node[above right] {};
    \draw[fill=black] (1,1) circle (0.12) node[above right] {};
    \end{tikzpicture}
    \caption{$P_1^{(\alpha)}$-polygon.}
    \label{fig-W3-g1-num}
     \end{minipage}\hfill
    \begin{minipage}{0.48\textwidth}
    \centering
     \begin{tikzpicture}[scale = 0.5]
        \draw[thin,dotted] (0,0) grid (2,2);
        \draw[->] (0,0) -- (2.4,0) node[right] {$k$};
        \draw[->] (0,0) -- (0,2.4) node[above] {$m$};
        \foreach \x/\xlabel in { 1/1, 2/2}
        \draw (\x cm,1pt ) -- (\x cm,-1pt ) node[anchor=north,fill=white] {\xlabel};
        \foreach \y/\ylabel in { 1/1, 2/2}
        \draw (1pt,\y cm) -- (-1pt ,\y cm) node[anchor=east, fill=white] {\ylabel};
        \draw[thick, red] (0,1) -- (1,0);
        \draw[thick, red] (0,1) -- (0,0);
        \draw[thick, red] (0,0) -- (1,0);
        \draw[fill=black] (0,0) circle (0.12) node[above right] {};
        \draw[fill=black] (1,0) circle (0.12) node[above right] {};
        \draw[fill=black] (0,1) circle (0.12) node[above right] {};
    \end{tikzpicture}
    \caption{$Q_1^{(\alpha)}$-polygon.}
    \label{fig-W3-g1-den}
     \end{minipage}
\end{figure}

\begin{figure}[h!]
 \begin{minipage}{0.48\textwidth}
     \centering
     \begin{tikzpicture}[scale = 0.5]
        \draw[thin,dotted] (0,0) grid (3,5);
        \draw[->] (0,0) -- (3.4,0) node[right] {$k$};
        \draw[->] (0,0) -- (0,5.4) node[above] {$m$};
        \foreach \x/\xlabel in { 1/1, 2/2, 3/3}
    \draw (\x cm,1pt ) -- (\x cm,-1pt ) node[anchor=north,fill=white] {\xlabel};
  \foreach \y/\ylabel in { 1/1, 2/2, 3/3, 4/4, 5/5}
    \draw (1pt,\y cm) -- (-1pt ,\y cm) node[anchor=east, fill=white] {\ylabel};
    \draw[thick, red] (0,5) -- (0,2);
    \draw[thick, red] (0,2) -- (1,1);
    \draw[thick, red] (3,1) -- (1,1);
    \draw[thick, red] (3,1) -- (3,2);
    \draw[thick, red] (0,5) -- (3,2);
    \draw[fill=black] (1,1) circle (0.12) node[above right] {};
    \draw[fill=black] (2,1) circle (0.12) node[above right] {};
    \draw[fill=black] (3,1) circle (0.12) node[above right] {};
    \draw[fill=black] (0,2) circle (0.12) node[above right] {};
    \draw[fill=black] (1,2) circle (0.12) node[above right] {};
    \draw[fill=black] (2,2) circle (0.12) node[above right] {};
    \draw[fill=black] (3,2) circle (0.12) node[above right] {};
    \draw[fill=black] (0,3) circle (0.12) node[above right] {};
    \draw[fill=black] (1,3) circle (0.12) node[above right] {};
    \draw[fill=black] (2,3) circle (0.12) node[above right] {};
    \draw[fill=black] (0,4) circle (0.12) node[above right] {};
    \draw[fill=black] (1,4) circle (0.12) node[above right] {};
    \draw[fill=black] (0,5) circle (0.12) node[above right] {};
    \draw[fill=black] (0,0) circle (0.00) node[above right] {};
    \end{tikzpicture}
    \caption{$P_2^{(\alpha)}$-polygon.}
    \label{fig-W3-g2-num}
 \end{minipage}\hfill
  \begin{minipage}{0.48\textwidth}
     \centering
     \begin{tikzpicture}[scale = 0.5]
        \draw[thin,dotted] (0,0) grid (3,5);
        \draw[->] (0,0) -- (3.4,0) node[right] {$k$};
        \draw[->] (0,0) -- (0,5.4) node[above] {$m$};
        \foreach \x/\xlabel in { 1/1, 2/2, 3/3}
    \draw (\x cm,1pt ) -- (\x cm,-1pt ) node[anchor=north,fill=white] {\xlabel};
  \foreach \y/\ylabel in { 1/1, 2/2, 3/3, 4/4, 5/5}
    \draw (1pt,\y cm) -- (-1pt ,\y cm) node[anchor=east, fill=white] {\ylabel};
    \draw[thick, red] (0,4) -- (3,1);
    \draw[thick, red] (0,4) -- (0,1);
    \draw[thick, red] (0,1) -- (1,0);
    \draw[thick, red] (1,0) -- (3,0);
    \draw[thick, red] (3,0) -- (3,1);
    \draw[fill=black] (1,0) circle (0.12) node[above right] {};
    \draw[fill=black] (2,0) circle (0.12) node[above right] {};
    \draw[fill=black] (3,0) circle (0.12) node[above right] {};
    \draw[fill=black] (0,1) circle (0.12) node[above right] {};
    \draw[fill=black] (1,1) circle (0.12) node[above right] {};
    \draw[fill=black] (2,1) circle (0.12) node[above right] {};
    \draw[fill=black] (3,1) circle (0.12) node[above right] {};
    \draw[fill=black] (0,2) circle (0.12) node[above right] {};
    \draw[fill=black] (1,2) circle (0.12) node[above right] {};
    \draw[fill=black] (2,2) circle (0.12) node[above right] {};
    \draw[fill=black] (0,3) circle (0.12) node[above right] {};
    \draw[fill=black] (1,3) circle (0.12) node[above right] {};
    \draw[fill=black] (0,4) circle (0.12) node[above right] {};
    \end{tikzpicture}
    \caption{$Q_2^{(\alpha)}$-polygon.}
        \label{fig-W3-g2-den}
  \end{minipage}
\end{figure}

\begin{figure}[h!]
 \begin{minipage}{0.48\textwidth}
   \centering
     \begin{tikzpicture}[scale = 0.5]
        \draw[thin,dotted] (0,0) grid (8,12);
        \draw[->] (0,0) -- (8.4,0) node[right] {$k$};
        \draw[->] (0,0) -- (0,12.4) node[above] {$m$};
        \foreach \x/\xlabel in {1/1, 2/2, 3/3, 4/4, 5/5, 6/6, 7/7, 8/8}
    \draw (\x cm,1pt ) -- (\x cm,-1pt ) node[anchor=north,fill=white] {\xlabel};
  \foreach \y/\ylabel in {1/1, 2/2, 3/3, 4/4, 5/5, 6/6, 7/7, 8/8, 9/9, 10/10, 11/11,12/12}
    \draw (1pt,\y cm) -- (-1pt ,\y cm) node[anchor=east, fill=white] {\ylabel};
    \draw[thick, red] (0,12) -- (8,4);
    \draw[thick, red] (0,12) -- (0,1);
    \draw[thick, red] (0,1) -- (1,0);
    \draw[thick, red] (1,0) -- (6,0);
    \draw[thick, red] (8,1) -- (6,0);
    \draw[thick, red] (8,1) -- (8,4);
    \draw[fill=black] (1,0) circle (0.12) node[above right] {};
    \draw[fill=black] (2,0) circle (0.12) node[above right] {};
    \draw[fill=black] (3,0) circle (0.12) node[above right] {};
    \draw[fill=black] (4,0) circle (0.12) node[above right] {};
    \draw[fill=black] (5,0) circle (0.12) node[above right] {};
    \draw[fill=black] (6,0) circle (0.12) node[above right] {};
    \draw[fill=black] (0,1) circle (0.12) node[above right] {};
    \draw[fill=black] (1,1) circle (0.12) node[above right] {};
    \draw[fill=black] (2,1) circle (0.12) node[above right] {};
    \draw[fill=black] (3,1) circle (0.12) node[above right] {};
    \draw[fill=black] (4,1) circle (0.12) node[above right] {};
    \draw[fill=black] (5,1) circle (0.12) node[above right] {};
    \draw[fill=black] (6,1) circle (0.12) node[above right] {};
    \draw[fill=black] (7,1) circle (0.12) node[above right] {};
    \draw[fill=black] (8,1) circle (0.12) node[above right] {};
    \draw[fill=black] (0,2) circle (0.12) node[above right] {};
    \draw[fill=black] (1,2) circle (0.12) node[above right] {};
    \draw[fill=black] (2,2) circle (0.12) node[above right] {};
    \draw[fill=black] (3,2) circle (0.12) node[above right] {};
    \draw[fill=black] (4,2) circle (0.12) node[above right] {};
    \draw[fill=black] (5,2) circle (0.12) node[above right] {};
    \draw[fill=black] (6,2) circle (0.12) node[above right] {};
    \draw[fill=black] (7,2) circle (0.12) node[above right] {};
    \draw[fill=black] (8,2) circle (0.12) node[above right] {};
    \draw[fill=black] (0,3) circle (0.12) node[above right] {};
    \draw[fill=black] (1,3) circle (0.12) node[above right] {};
    \draw[fill=black] (2,3) circle (0.12) node[above right] {};
    \draw[fill=black] (3,3) circle (0.12) node[above right] {};
    \draw[fill=black] (4,3) circle (0.12) node[above right] {};
    \draw[fill=black] (5,3) circle (0.12) node[above right] {};
    \draw[fill=black] (6,3) circle (0.12) node[above right] {};
    \draw[fill=black] (7,3) circle (0.12) node[above right] {};
    \draw[fill=black] (8,3) circle (0.12) node[above right] {};
    \draw[fill=black] (0,4) circle (0.12) node[above right] {};
    \draw[fill=black] (1,4) circle (0.12) node[above right] {};
    \draw[fill=black] (2,4) circle (0.12) node[above right] {};
    \draw[fill=black] (3,4) circle (0.12) node[above right] {};
    \draw[fill=black] (4,4) circle (0.12) node[above right] {};
    \draw[fill=black] (5,4) circle (0.12) node[above right] {};
    \draw[fill=black] (6,4) circle (0.12) node[above right] {};
    \draw[fill=black] (7,4) circle (0.12) node[above right] {};
    \draw[fill=black] (8,4) circle (0.12) node[above right] {};
    \draw[fill=black] (0,5) circle (0.12) node[above right] {};
    \draw[fill=black] (1,5) circle (0.12) node[above right] {};
    \draw[fill=black] (2,5) circle (0.12) node[above right] {};
    \draw[fill=black] (3,5) circle (0.12) node[above right] {};
    \draw[fill=black] (4,5) circle (0.12) node[above right] {};
    \draw[fill=black] (5,5) circle (0.12) node[above right] {};
    \draw[fill=black] (6,5) circle (0.12) node[above right] {};
    \draw[fill=black] (7,5) circle (0.12) node[above right] {};
    \draw[fill=black] (0,6) circle (0.12) node[above right] {};
    \draw[fill=black] (1,6) circle (0.12) node[above right] {};
    \draw[fill=black] (2,6) circle (0.12) node[above right] {};
    \draw[fill=black] (3,6) circle (0.12) node[above right] {};
    \draw[fill=black] (4,6) circle (0.12) node[above right] {};
    \draw[fill=black] (5,6) circle (0.12) node[above right] {};
    \draw[fill=black] (6,6) circle (0.12) node[above right] {};
    \draw[fill=black] (0,7) circle (0.12) node[above right] {};
    \draw[fill=black] (1,7) circle (0.12) node[above right] {};
    \draw[fill=black] (2,7) circle (0.12) node[above right] {};
    \draw[fill=black] (3,7) circle (0.12) node[above right] {};
    \draw[fill=black] (4,7) circle (0.12) node[above right] {};
    \draw[fill=black] (5,7) circle (0.12) node[above right] {};
    \draw[fill=black] (0,8) circle (0.12) node[above right] {};
    \draw[fill=black] (1,8) circle (0.12) node[above right] {};
    \draw[fill=black] (2,8) circle (0.12) node[above right] {};
    \draw[fill=black] (3,8) circle (0.12) node[above right] {};
    \draw[fill=black] (4,8) circle (0.12) node[above right] {};
    \draw[fill=black] (0,9) circle (0.12) node[above right] {};
    \draw[fill=black] (1,9) circle (0.12) node[above right] {};
    \draw[fill=black] (2,9) circle (0.12) node[above right] {};
    \draw[fill=black] (3,9) circle (0.12) node[above right] {};
    \draw[fill=black] (0,10) circle (0.12) node[above right] {};
    \draw[fill=black] (1,10) circle (0.12) node[above right] {};
    \draw[fill=black] (2,10) circle (0.12) node[above right] {};
    \draw[fill=black] (0,11) circle (0.12) node[above right] {};
    \draw[fill=black] (1,11) circle (0.12) node[above right] {};
    \draw[fill=black] (0,12) circle (0.12) node[above right] {};
    \end{tikzpicture}
    \caption{$P_3^{(\alpha)}$-polygon.}
    \label{fig-W3-g3-num}
 \end{minipage}\hfill
  \begin{minipage}{0.48\textwidth}
     \centering
     \begin{tikzpicture}[scale = 0.5]
        \draw[thin,dotted] (0,0) grid (8,12);
        \draw[->] (0,0) -- (8.4,0) node[right] {$k$};
        \draw[->] (0,0) -- (0,12.4) node[above] {$m$};
        \foreach \x/\xlabel in {1/1, 2/2, 3/3, 4/4, 5/5, 6/6, 7/7, 8/8}
    \draw (\x cm,1pt ) -- (\x cm,-1pt ) node[anchor=north,fill=white] {\xlabel};
  \foreach \y/\ylabel in {1/1, 2/2, 3/3, 4/4, 5/5, 6/6, 7/7, 8/8, 9/9, 10/10, 11/11,12/12}
    \draw (1pt,\y cm) -- (-1pt ,\y cm) node[anchor=east, fill=white] {\ylabel};
    \draw[thick, red] (0,11) -- (8,3);
    \draw[thick, red] (0,11) -- (0,1);
    \draw[thick, red] (0,1) -- (1,0);
    \draw[thick, red] (1,0) -- (8,0);
    \draw[thick, red] (8,0) -- (8,3);
    \draw[fill=black] (1,0) circle (0.12) node[above right] {};
    \draw[fill=black] (2,0) circle (0.12) node[above right] {};
    \draw[fill=black] (3,0) circle (0.12) node[above right] {};
    \draw[fill=black] (4,0) circle (0.12) node[above right] {};
    \draw[fill=black] (5,0) circle (0.12) node[above right] {};
    \draw[fill=black] (6,0) circle (0.12) node[above right] {};
    \draw[fill=black] (7,0) circle (0.12) node[above right] {};
    \draw[fill=black] (8,0) circle (0.12) node[above right] {};
    \draw[fill=black] (0,1) circle (0.12) node[above right] {};
    \draw[fill=black] (1,1) circle (0.12) node[above right] {};
    \draw[fill=black] (2,1) circle (0.12) node[above right] {};
    \draw[fill=black] (3,1) circle (0.12) node[above right] {};
    \draw[fill=black] (4,1) circle (0.12) node[above right] {};
    \draw[fill=black] (5,1) circle (0.12) node[above right] {};
    \draw[fill=black] (6,1) circle (0.12) node[above right] {};
    \draw[fill=black] (7,1) circle (0.12) node[above right] {};
    \draw[fill=black] (8,1) circle (0.12) node[above right] {};
    \draw[fill=black] (0,2) circle (0.12) node[above right] {};
    \draw[fill=black] (1,2) circle (0.12) node[above right] {};
    \draw[fill=black] (2,2) circle (0.12) node[above right] {};
    \draw[fill=black] (3,2) circle (0.12) node[above right] {};
    \draw[fill=black] (4,2) circle (0.12) node[above right] {};
    \draw[fill=black] (5,2) circle (0.12) node[above right] {};
    \draw[fill=black] (6,2) circle (0.12) node[above right] {};
    \draw[fill=black] (7,2) circle (0.12) node[above right] {};
    \draw[fill=black] (8,2) circle (0.12) node[above right] {};
    \draw[fill=black] (0,3) circle (0.12) node[above right] {};
    \draw[fill=black] (1,3) circle (0.12) node[above right] {};
    \draw[fill=black] (2,3) circle (0.12) node[above right] {};
    \draw[fill=black] (3,3) circle (0.12) node[above right] {};
    \draw[fill=black] (4,3) circle (0.12) node[above right] {};
    \draw[fill=black] (5,3) circle (0.12) node[above right] {};
    \draw[fill=black] (6,3) circle (0.12) node[above right] {};
    \draw[fill=black] (7,3) circle (0.12) node[above right] {};
    \draw[fill=black] (8,3) circle (0.12) node[above right] {};
    \draw[fill=black] (0,4) circle (0.12) node[above right] {};
    \draw[fill=black] (1,4) circle (0.12) node[above right] {};
    \draw[fill=black] (2,4) circle (0.12) node[above right] {};
    \draw[fill=black] (3,4) circle (0.12) node[above right] {};
    \draw[fill=black] (4,4) circle (0.12) node[above right] {};
    \draw[fill=black] (5,4) circle (0.12) node[above right] {};
    \draw[fill=black] (6,4) circle (0.12) node[above right] {};
    \draw[fill=black] (7,4) circle (0.12) node[above right] {};
    \draw[fill=black] (0,5) circle (0.12) node[above right] {};
    \draw[fill=black] (1,5) circle (0.12) node[above right] {};
    \draw[fill=black] (2,5) circle (0.12) node[above right] {};
    \draw[fill=black] (3,5) circle (0.12) node[above right] {};
    \draw[fill=black] (4,5) circle (0.12) node[above right] {};
    \draw[fill=black] (5,5) circle (0.12) node[above right] {};
    \draw[fill=black] (6,5) circle (0.12) node[above right] {};
    \draw[fill=black] (0,6) circle (0.12) node[above right] {};
    \draw[fill=black] (1,6) circle (0.12) node[above right] {};
    \draw[fill=black] (2,6) circle (0.12) node[above right] {};
    \draw[fill=black] (3,6) circle (0.12) node[above right] {};
    \draw[fill=black] (4,6) circle (0.12) node[above right] {};
    \draw[fill=black] (5,6) circle (0.12) node[above right] {};
    \draw[fill=black] (0,7) circle (0.12) node[above right] {};
    \draw[fill=black] (1,7) circle (0.12) node[above right] {};
    \draw[fill=black] (2,7) circle (0.12) node[above right] {};
    \draw[fill=black] (3,7) circle (0.12) node[above right] {};
    \draw[fill=black] (4,7) circle (0.12) node[above right] {};
    \draw[fill=black] (0,8) circle (0.12) node[above right] {};
    \draw[fill=black] (1,8) circle (0.12) node[above right] {};
    \draw[fill=black] (2,8) circle (0.12) node[above right] {};
    \draw[fill=black] (3,8) circle (0.12) node[above right] {};
    \draw[fill=black] (0,9) circle (0.12) node[above right] {};
    \draw[fill=black] (1,9) circle (0.12) node[above right] {};
    \draw[fill=black] (2,9) circle (0.12) node[above right] {};
    \draw[fill=black] (0,10) circle (0.12) node[above right] {};
    \draw[fill=black] (1,10) circle (0.12) node[above right] {};
    \draw[fill=black] (0,11) circle (0.12) node[above right] {};
    \end{tikzpicture}
    \caption{$Q_3^{(\alpha)}$-polygon.}
    \label{fig-W3-g3-den}
  \end{minipage}
\end{figure}

Now we apply the Newton polygon technique developed in Appendix  \bref{app:forth}. The  polygons associated to the numerators and denominators in  $G^{(\alpha)}_i$ are shown on fig. \bref{fig-W3-g1-num}--\bref{fig-W3-g3-den}. One can easily check that properties {\bf (a) - (e)} from Appendix  \bref{app:forth} are true for considered polygons. E.g., let us check the property {\bf (a)}. We have 
\be
\begin{aligned}
    & L_{P_i^{(\alpha)}} = L_{Q_i^{(\alpha)}}  = 0\;;   \quad D_{P_1^{(\alpha)}} = D_{P_2^{(\alpha)}} = 1 \;, \quad  D_{Q_i^{(\alpha)}} = 0\;;\\
    & R_{P_1^{(\alpha)}} = R_{Q_1^{(\alpha)}} = 1 \;, \quad R_{P_2^{(\alpha)}} = R_{Q_2^{(\alpha)}} = 3\;, \quad R_{P_3^{(\alpha)}} = R_{Q_3^{(\alpha)}} = 8 \;; \\
    & U_{P_1^{(\alpha)}} = 2, \quad U_{Q_1^{(\alpha)}} = 1, \quad U_{P_2^{(\alpha)}} = 5\;, \quad U_{Q_2^{(\alpha)}} = 4, \quad U_{P_3^{(\alpha)}} =  12\;,\quad U_{Q_3^{(\alpha)}} = 11\;.
\end{aligned}
\ee
Obviously, $R_{P_i^{(\alpha)}} - L_{P_i^{(\alpha)}} \leq U_{P_i^{(\alpha)}} - D_{P_i^{(\alpha)}}$ and  $R_{Q_i^{(\alpha)}} - L_{Q_i^{(\alpha)}} \leq U_{Q_i^{(\alpha)}} - D_{Q_i^{(\alpha)}}$ for all $i \in \{1,2,3\}$, so the property {\bf (a)} is true. Other properties can be checked directly. Consequently,  we can provide the same reasoning as in Appendix \bref{app:forth} and conclude that all $G_n^{(\alpha)}$, $n=1,2,3$ have the same Puiseux expansion as the second coefficient of Virasoro block $g_2$. Therefore, all conclusions that were made in Appendix \bref{app:second} hold for $G_n^{(\alpha)}$.

\section{Numerics}
\label{sec:numerics}

This section contains several explicit expressions for the $\alpha$-heavy classical Virasoro and $\cW_3$ conformal blocks calculated at different values of $\alpha$. 

\subsection{$Vir$ classical conformal blocks}
\label{app:vir_num}

\paragraph{$\bm{ \alpha\in (0,1)_\rn}$.}  Here, $\alpha = 1/2$ is exceptional, $\alpha = 1/3, 3/5$ are regular (see Definition \bref{def:reg}). 

\hspace{2mm}

\noindent $\cS^{(1/3)} = f^{(1/3)}_{1/3} c^{1/3} + f^{(1/3)}_0$:  
\be
\label{1/3-heavy}
\ba{l}
\dps
f_{1/3}^{(1/3)} = \frac{\tilde{\delta} }{2}z+\frac{3 \tilde{\delta}}{16}z^2+\frac{5 \tilde{\delta} }{48}z^3+\frac{35 \tilde{\delta} }{512}z^4+\frac{63 \tilde{\delta}}{1280}z^5+ O(z^6)\,, \vspace{2mm}
\\
\dps
f_{0}^{(1/3)} =\frac{z^2}{32} +\frac{z^3}{32}+  \frac{29 z^4}{1024}+\frac{13 z^5}{512}+ O(z^6)\,.
\ea
\ee
\noindent $\cS^{(1/2)} = f^{(1/2)}_1 c^{1/2} + f^{(1/2)}_0$:   
\be
\label{1/2-heavy}
\ba{l}
\dps
f^{(1/2)}_{1/2} = \frac{\tilde\delta}{2} z + \frac{3 \tilde\delta}{16} z^2 + \frac{5 \tilde\delta}{48} z^3 + \frac{35 \tilde\delta}{512} z^4 + \frac{63 \tilde\delta}{1280} z^5 + O(z^6)\,, \vspace{2mm}
\\
\dps
f^{(1/2)}_0 = \frac{1 + 4(\tilde\delta + 4 \delta)^2}{32} z^2 + \frac{1 + 4(\tilde\delta + 4 \delta)^2}{32} z^3 + \left(\frac{9 \delta ^2}{5}+\frac{141 \tilde{\delta}^2}{1280}+\frac{71 \delta  \tilde{\delta}}{80} + \frac{29}{1024}\right) z^4 
\vspace{2mm} 
\\
\dps
\hspace{15mm}+\left(\frac{8 \delta ^2}{5}+\frac{61 \tilde{\delta}^2}{640}+\frac{31 \delta  \tilde{\delta}}{40}+\frac{13}{512}\right)  z^5 + O(z^6)\,.
\ea
\ee
$\cS^{(3/5)} = f^{(3/5)}_1 c^{3/5}+ f^{(3/5)}_{3/5}c^{3/5} + f^{(3/5)}_0$:  
\be
\label{3/5-heavy}
\ba{l}
\dps
f_{3/5}^{(3/5)} = \frac{\tilde{\delta} }{2}z+\frac{3 \tilde{\delta}}{16}z^2+\frac{5 \tilde{\delta} }{48}z^3+\frac{35 \tilde{\delta} }{512}z^4+\frac{63 \tilde{\delta}}{1280}z^5+ O(z^6)\,, \vspace{2mm}
\\
\dps
f_{1/5}^{(3/5)} =\frac{(4 \delta +\tilde{\delta})^2}{8} z^2 +\frac{(4 \delta +\tilde{\delta})^2}{8}
   z^3 + \left(\frac{9 \delta ^2}{5}+\frac{141 \tilde{\delta}^2}{1280}+\frac{71 \delta  \tilde{\delta}}{80}\right)z^4 + O(z^5)\,,
\vspace{2mm} 
\\
\dps
f_{0}^{(3/5)} =\frac{z^2}{32} +\frac{z^3}{32}+  \frac{29 z^4}{1024}+\frac{13 z^5}{512}+ O(z^6)\,.
\ea
\ee
There are relations $f_{1/3}^{(1/3)} = f_{1/2}^{(1/3)} = f^{(3/5)}_{3/5}$, $f_{0}^{(1/3)} = f_{0}^{(3/5)}$ and $f_0^{(1/2)} = f_{0}^{(3/5)} + f_{1/5}^{(3/5)}$ which illustrate Propositions \bref{prop1} and \bref{universal01}. 


\paragraph{$\bm{\alpha\in(1,\infty)_\rn}$.}  Here, $\alpha = 2$ is exceptional, $\alpha = 5/3, 3$ are regular. 

\hspace{2mm}

\noindent $\cS^{(5/3)} = f^{(5/3)}_{5/3} c^{5/3} +f^{(5/3)}_{1} c^{1}+f^{(5/3)}_{1/3}c^{1/3}+ f^{(5/3)}_0$:   
\be
\ba{l}
\dps
f_{5/3}^{(5/3)} = \frac{\tilde\delta}{2} z + \frac{13 \tilde\delta^2 + 8 \tilde\delta \delta + 16 \delta ^2}{64 \tilde\delta} z^2 + \frac{23 \tilde\delta^2 + 24 \tilde\delta \delta + 48 \delta ^2}{192 \tilde\delta} z^3 + O(z^4)\,, 
\vspace{2mm}
\\
\dps
f_{1}^{(5/3)} = - \frac{(\tilde\delta + 4 \delta)^2}{512 \tilde\delta^2} z^2 - \frac{(\tilde\delta + 4 \delta)^2}{512 \tilde\delta^2} z^3 + O(z^4)\,,
\vspace{2mm}
\\
\dps 
f_{1/3}^{(5/3)} = \frac{(\tilde\delta + 4 \delta)^2}{4096 \tilde\delta^3} z^2 + \frac{(\tilde\delta + 4 \delta)^2}{4096 \tilde\delta^3} z^3 + O(z^4)\,,
\vspace{2mm}
\\
\dps 
f_0^{(5/3)} =  \frac{(\tilde\delta - 12 \delta)^2}{512 \tilde\delta^2} z^2 + \frac{(\tilde\delta -12 \delta)^2}{512 \tilde\delta^2} z^3 + O(z^4)\,,\,.
\ea
\ee
\noindent $\cS^{(2)} = f^{(2)}_{2} c^{2} +f^{(2)}_{1} c^{1}+ f^{(2)}_0$:   
\be
\ba{l}
\dps
f_2^{(2)} = \frac{\tilde\delta}{2} z + \frac{13 \tilde\delta^2 + 8 \tilde\delta \delta + 16 \delta ^2}{64 \tilde\delta} z^2 + \frac{23 \tilde\delta^2 + 24 \tilde\delta \delta + 48 \delta ^2}{192 \tilde\delta} z^3 + O(z^4)\,, 
\vspace{2mm}
\\
\dps
f_1^{(2)} = - \frac{(\tilde\delta + 4 \delta)^2}{512 \tilde\delta^2} z^2 - \frac{(\tilde\delta + 4 \delta)^2}{512 \tilde\delta^2} z^3 + O(z^4)\,,
\vspace{2mm}
\\
\dps 
f_0^{(2)} = \frac{\tilde\delta^2 (8 \tilde\delta + 1) + 8 \tilde\delta \delta (1 - 24 \tilde\delta) + 16 \delta ^2 (72 \tilde\delta + 1)}{4096 \tilde\delta^3} (z^2+z^3) + O(z^4)\,.
\ea
\ee
\noindent $\cS^{(3)} = f^{(3)}_{3} c^{3} +f^{(3)}_{1} c^{1}+ f^{(3)}_0$:  
\be
\label{3-heavy}
\ba{l}
\dps
f_3^{(3)} = \frac{\tilde\delta}{2} z + \frac{13 \tilde\delta^2 + 8 \tilde\delta \delta + 16 \delta ^2}{64 \tilde\delta} z^2 + \frac{23 \tilde\delta^2 + 24 \tilde\delta \delta + 48 \delta ^2}{192 \tilde\delta} z^3 + O(z^4)\,, 
\vspace{2mm}
\\
\dps
f_1^{(3)} = - \frac{(\tilde\delta + 4 \delta)^2}{512 \tilde\delta^2} z^2 - \frac{(\tilde\delta + 4 \delta)^2}{512 \tilde\delta^2} z^3 + O(z^4)\,,
\vspace{2mm}
\\
\dps 
f_0^{(3)} =  \frac{(\tilde\delta - 12 \delta)^2}{512 \tilde\delta^2} z^2 + \frac{(\tilde\delta -12 \delta)^2}{512 \tilde\delta^2} z^3 + O(z^4)\,.
\ea
\ee
There are relations $f_{3}^{(3)} = f_{5/3}^{(5/3)} = f^{(2)}_{2}$, $f_{1}^{(3/2)} = f_{1}^{(5/3)} = f_{1}^{(2)}$, $f_0^{(5/3)}= f_0^{(3)}$ and  $f_0^{(2)} = f_{0}^{(5/3)} + f_{1/5}^{(5/3)}$ which illustrate Propositions \bref{prop2} and \bref{universal01}. 


\paragraph{${\bm\alpha=\bf 1}$.} The classical conformal block $\cS^{(1)} = f^{(1)}_1 c + f^{(1)}_0$ is given by   
\be
\label{1-heavy}
\ba{l}
\dps
f^{(1)}_1 = \frac{\tilde\delta}{2} z + \frac{\tilde\delta (26 \tilde\delta + 3) + 16 \tilde\delta \delta + 32 \delta ^2}{16 (8 \tilde\delta + 1)} z^2 + \frac{46 \tilde\delta^2 + 48 \tilde\delta \delta + 96 \delta ^2 + 5 \tilde\delta}{48 (8 \tilde\delta + 1)} z^3 + O(z^4)\,,
\vspace{2mm}
\\
\dps
f_0^{(1)} = \frac{(2 \tilde\delta -24 \delta + 1)^2}{32 (8 \tilde\delta + 1)^2} z^2 + \frac{\big( 2 \tilde\delta - 24 \delta + 1  \big)^2}{32 (8 \tilde\delta + 1)^2} z^3 + O(z^4)\,. \ea
\ee
Higher-order  terms are quite complicated to be presented here.\footnote{In the Liouville parametrization the Zamolodchikov classical block $f^{(1)}_1$ for particular conformal dimensions  was calculated up to a high order \cite{Hadasz:2005gk}.} A perturbative expansion of $f^{(1)}_1$ \eqref{1-heavy}  around $\delta, \tilde{\delta} = 0$ is given by $f^{(1)}_1 \simeq \sum_{m=1}^{\infty} P_m (\delta, \tilde{\delta}|z)$ \eqref{perturbative-f1-small}, where 
\be
\label{1-pert}
\ba{l}
\dps
P_1 = \frac{\tilde{\delta} }{2}z+\frac{3 \tilde{\delta}}{16}z^2+\frac{5 \tilde{\delta} }{48}z^3+\frac{35 \tilde{\delta} }{512}z^4+\frac{63 \tilde{\delta}}{1280}z^5+ O(z^6)\,, \vspace{2mm}
\\
\dps
P_2 =\frac{(4 \delta +\tilde{\delta})^2}{8} z^2 +\frac{(4 \delta +\tilde{\delta})^2}{8}
   z^3 + \left(\frac{9 \delta ^2}{5}+\frac{141 \tilde{\delta}^2}{1280}+\frac{71 \delta  \tilde{\delta}}{80}\right)z^4  + O(z^5)\,.
\ea
\ee
We see that $P_1 = f^{(\alpha)}_\alpha$ and $P_2 = f^{(\alpha)}_{2\alpha-1}$ for $\forall\alpha \in (0,1)_\rn$, e.g. one can compare with \eqref{3/5-heavy} (i.e. choosing $\alpha = 3/5$). These relations illustrate the first part of Proposition \bref{corollary31}. 

Similarly, an asymptotic expansion of $f^{(1)}_1$ \eqref{1-heavy}  around $\delta, \tilde{\delta} = \infty$ is given by $f^{(1)}_1 \simeq \sum_{k=-1}^{\infty} S_k (\delta, \tilde{\delta}|z)$ \eqref{perturbative-f1-large}, where 
\be
\label{1-non-pert}
\ba{l}
\dps
S_{-1} = \frac{\tilde\delta}{2} z + \frac{13 \tilde\delta^2 + 8 \tilde\delta \delta + 16 \delta ^2}{64 \tilde\delta} z^2 + \frac{23 \tilde\delta^2 + 24 \tilde\delta \delta + 48 \delta ^2}{192 \tilde\delta} z^3 + O(z^4)\,, 
\vspace{2mm}
\\
\dps
S_{0} = - \frac{(\tilde\delta + 4 \delta)^2}{512 \tilde\delta^2} z^2 - \frac{(\tilde\delta + 4 \delta)^2}{512 \tilde\delta^2} z^3 + O(z^4)\,.
\ea
\ee
 We see that $S_{-1} = f^{(\alpha)}_\alpha$ and $S_0 = f^{(\alpha)}_1$  for $\forall\alpha\in (1, \infty)_\rn$, e.g. one can compare with \eqref{3-heavy} (i.e. choosing $\alpha =3$). These relations illustrate the second part of Proposition \bref{corollary31}.


\paragraph{Isolated coefficients.} Here, we collect isolated coefficients ($c^0$-terms) of the  classical vacuum conformal blocks at exceptional points $\tilde\alpha_{n} = \frac{n}{n+1} \in (0,1)_\rn$ for $n=1,...,7$ (using {\it Mathematica} they can be calculated up to $O(z^{16})$):
\be
\label{iso1}
\dot f_0^{(\frac12)} = 2 \delta^2 z^2\left[1+ z+\frac{9 z^2}{10}+\frac{4 z^3}{5}+\frac{5 z^4}{7}+\frac{9 z^5}{14}+\frac{7 z^6}{12} +\frac{8 z^7}{15} +\frac{27 z^8}{55} +\frac{5 z^9}{11} +O(z^{10})\right],
\ee
\be
\label{iso2}
\dot f_0^{(\frac23)} =\frac{4 \delta^3 z^4}{5}\left[1+2z+\frac{39 z^2}{14}+\frac{47 z^3}{14}+\frac{263 z^4}{70}+\frac{141 z^5}{35} +\frac{1943 z^6}{462} +\frac{1993 z^7}{462}+O(z^{8})\right],
\ee
\be
\label{iso3}
\dot f_0^{(\frac34)} = -\frac{44 \delta^4 z^4}{5}\left[1 + 2 z + \frac{2503}{924} z^2 + \frac{963 }{308}z^3 + \frac{38359}{11550}z^4 + \frac{12937}{3850}z^5 + \frac{83773}{25410}z^6 + O(z^8)\right],
\ee
\be
\dot f_0^{(\frac45)} =-\frac{592 \delta^5 z^6}{35}\left[1 + 3 z + \frac{20833}{3700} z^2 + \frac{7883}{925} z^3 + \frac{2785771}{244200}z^4 + \frac{689743}{48840}z^5+O(z^6)\right],
\ee
\be
\dot f_0^{(\frac56)} =\frac{12032 \delta^6 z^6}{105}\left[1 + 3 z + \frac{203501}{37600}z^2 + \frac{71901}{9400}z^3 + \frac{15432733}{1654400}z^4 + \frac{3399349}{330880}z^5+O(z^6)\right],
\ee
\be
\dot f_0^{(\frac67)} = \frac{375232 \delta^7 z^8}{875}\left[1+4z+\frac{1826413 z^2}{193479} +\frac{3327695 z^3}{193479}+O(z^4)\right],
\ee
\be
\label{iso7}
\dot f_0^{(\frac78)} = -\frac{1862816 \delta^8z^8}{875} \left[1+4z +\frac{34611001 z^2}{3842058}+\frac{57793265 z^3}{3842058}+O(z^4)\right].
\ee

\subsection{$\cW_3$ classical conformal blocks} 
\label{app:w3_numerics}

The operators are chosen as \eqref{W3opers}.  

\paragraph{${\bm{\alpha \in (0,1)_\rn}}$.} Here,  $\alpha = 1/2$ is exceptional, $\alpha = 1/3, 3/5$ are regular (see Definition \bref{def:reg}). 

\hspace{2mm} 

\noindent  $\ws^{(1/3)} = \cb^{(1/3)}_{1/3}c^{1/3} + \cb^{(1/3)}_{0}$: 
\be
\cb^{(1/3)}_{1/3} = \tilde{\delta}z+\frac{3 \tilde{\delta} }{4}z^2+ \frac{5 \tilde{\delta}}{6}z^3+O(z^4)\;, \qquad \cb^{(1/3)}_{0} =  \frac{1}{4}z^2 +  \frac{1}{2}z^3 + O(z^4)\,.
\ee
\noindent $\ws^{(1/2)} = \cb^{(1/2)}_{1/2}c^{1/2} + \cb^{(1/2)}_{0}$: 
\be
\ba{l}
\dps
\cb^{(1/2)}_{1/2} = \tilde{\delta}z +\frac{3 \tilde{\delta} }{4}z^2+\frac{5 \tilde{\delta}}{6}z^3 + O(z^4)\,,
\vspace{2mm}
\\
\dps
\cb^{(1/2)}_{0} =16 \tilde{\delta}  (\tilde{\delta}-\delta )z+\left(2 \delta ^2+28 \tilde{\delta}^2-28 \delta  \tilde{\delta}+\frac{1}{4}\right)z^2 +  \left(\frac{4}{15}\left(7 \delta ^2+200 \tilde{\delta}^2-192 \delta  \tilde{\delta}\right)+\frac{1}{2}\right)z^3 + O(z^4)\,.
\ea
\ee
\noindent $\ws^{(3/5)} = \cb^{(3/5)}_{3/5}c^{3/5} + \cb^{(3/5)}_{1/5}c^{1/5} + \cb^{(3/5)}_0$: 
\be
\ba{l}
\dps
\cb^{(3/5)}_{3/5} = \tilde{\delta}z + \frac{3 \tilde{\delta} }{4}z^2 + \frac{5 \tilde{\delta}}{6}z^3 + O(z^4)\,, \qquad \cb^{(3/5)}_{0} = \frac{1}{4}z^2+\frac{1}{2}z^3 + O(z^4) \,,
\vspace{2mm}
\\
\dps
\cb^{(3/5)}_{1/5} =16 \tilde{\delta}  (\tilde{\delta}-\delta )z+   (2\delta ^2+28 \tilde{\delta}^2-28 \delta  \tilde{\delta})z^2+ \frac{4}{15} \left(7 \delta ^2+200 \tilde{\delta}^2-192 \delta  \tilde{\delta}\right)z^3+O(z^4)\,.
\ea
\ee
There are relations $\cb^{(1/3)}_{1/3} = \cb^{(1/2)}_{1/2} = \cb^{(3/5)}_{3/5}$, $\cb^{(1/3)}_{0} = \cb^{(3/5)}_{0}$ and  $\cb_0^{(1/2)} = \cb_0^{(3/5)} +\cb_{1/5}^{(3/5)} $ which  illustrate Propositions from  Section  \bref{sec:structureW}. 

\paragraph{${\bm{\alpha \in (1,\infty)_\rn}}$.} Here, $\alpha = 3/2$ is exceptional, $\alpha= 3, 5/3$ are regular.

\hspace{2mm}

\noindent $\ws^{(5/3)} = \cb^{(5/3)}_{5/3}c^{5/3} + \cb^{(5/3)}_1c + \cb^{(5/3)}_{1/3}c^{1/3}+ \cb^{(5/3)}_0$:  
\be
\ba{l}
\dps
\cb^{(5/3)}_{5/3} =\frac{\delta +\tilde{\delta}}{2} z+ \frac{ (\delta +\tilde{\delta}) (9 \delta +13 \tilde{\delta})}{64 \tilde{\delta}}z^2+ \frac{1}{960}\left(250 \delta +\frac{119 \delta ^2}{\tilde{\delta}}+115 \tilde{\delta}\right)z^3+O(z^4)\,,
\vspace{2mm}
\\
\dps
\cb_1^{(5/3)} =\frac{\delta -\tilde{\delta}}{64 \tilde{\delta}}z+\frac{ 9 \delta ^2-9 \tilde{\delta}^2-2 \delta  \tilde{\delta}}{1024
   \tilde{\delta}^2}z^2+\frac{ \left(225 \delta ^3-103 \delta  \tilde{\delta}^2-190 \tilde{\delta}^3+132 \delta ^2 \tilde{\delta}\right)}{30720 \tilde{\delta}^3}z^3 + O(z^4)\,,
\vspace{2mm}
\\
\dps
\cb_{1/3}^{(5/3)} =\frac{ \delta -\tilde{\delta}}{2048 \tilde{\delta}^2}z+ \frac{23 \delta ^2-5 \tilde{\delta}^2-22 \delta  \tilde{\delta}}{65536 \tilde{\delta}^3}z^2 + \frac{15 \delta ^3-36 \delta  \tilde{\delta}^2-4 \tilde{\delta}^3-104 \delta ^2 \tilde{\delta}}{245760 \tilde{\delta}^4}z^3+O(z^4)\,,
\vspace{2mm}
\\
\dps
\cb_0^{(5/3)} = \frac{ \tilde{\delta}-\delta}{32 \tilde{\delta}}z +\frac{63 \delta ^2+9 \tilde{\delta}^2-22 \delta \tilde{\delta}}{512 \tilde{\delta}^2}z^2+\frac{1292 \delta ^2 \tilde{\delta}-225 \delta ^3-521 \delta  \tilde{\delta}^2+190 \tilde{\delta}^3}{15360 \tilde{\delta}^3}z^3+O(z^4)\,.
\ea
\ee
\noindent $\ws^{(2)} = \cb^{(2)}_{2}c^{2} + \cb^{(2)}_1c + \cb^{(2)}_0$: 
\be
\ba{l}
\dps
\cb_{2}^{(2)} =\frac{\delta +\tilde{\delta}}{2}z   +\frac{ (\delta +\tilde{\delta}) (9 \delta +13 \tilde{\delta})}{64 \tilde{\delta}}z^2+\frac{1}{960}  \left(250 \delta +\frac{119 \delta ^2}{\tilde{\delta}}+115 \tilde{\delta}\right)z^3+ O(z^4)\,,
\vspace{2mm}
\\
\dps
\cb_1^{(2)} =\frac{\delta -\tilde{\delta}}{64 \tilde{\delta}}z+\frac{ 9 \delta ^2-9 \tilde{\delta}^2-2 \delta  \tilde{\delta}}{1024
   \tilde{\delta}^2}z^2+\frac{ \left(225 \delta ^3-103 \delta  \tilde{\delta}^2-190 \tilde{\delta}^3+132 \delta ^2 \tilde{\delta}\right)}{30720 \tilde{\delta}^3}z^3 + O(z^4)\,,
\vspace{2mm}
\\
\dps
\begin{aligned}
& \cb_0^{(2)}=\frac{ (64 \tilde{\delta} -1) (\tilde{\delta} -\delta)}{2048 \tilde{\delta} ^2}z+ \frac{ (1152 \tilde{\delta} -5) \tilde{\delta} ^2+(8064 \tilde{\delta} +23) \delta^2-22 (128 \tilde{\delta} +1) \tilde{\delta}  \delta}{65536
   \tilde{\delta} ^3}z^2+ \\
& + \frac{3040 \tilde{\delta} ^4-4 \tilde{\delta} ^3 (2084 \delta+1)+4 \tilde{\delta} ^2 \delta (5168 \delta-9)-8 \tilde{\delta}  \delta^2 (450 \delta+13)+15
   \delta^3}{245760 \tilde{\delta} ^4}z^3+O(z^4) \,.
\end{aligned}
\ea
\ee
\noindent $\ws^{(3)} = \cb^{(3)}_{3}c^{3} + \cb^{(3)}_1c + \cb^{(3)}_0$: 
\be
\ba{l}
\dps
\cb_{3}^{(3)} =\frac{\delta +\tilde{\delta}}{2}z   +\frac{ (\delta +\tilde{\delta}) (9 \delta +13 \tilde{\delta})}{64 \tilde{\delta}}z^2+\frac{1}{960}  \left(250 \delta +\frac{119 \delta ^2}{\tilde{\delta}}+115 \tilde{\delta}\right)z^3+ O(z^4)\,,
\vspace{2mm}
\\
\dps
\cb_1^{(3)} =\frac{\delta -\tilde{\delta}}{64 \tilde{\delta}}z+\frac{ 9 \delta ^2-9 \tilde{\delta}^2-2 \delta  \tilde{\delta}}{1024
   \tilde{\delta}^2}z^2+\frac{ \left(225 \delta ^3-103 \delta  \tilde{\delta}^2-190 \tilde{\delta}^3+132 \delta ^2 \tilde{\delta}\right)}{30720 \tilde{\delta}^3}z^3 + O(z^4)\,,
\vspace{2mm}
\\
\dps
\cb_0^{(3)} = \frac{ \tilde{\delta}-\delta}{32 \tilde{\delta}}z +\frac{63 \delta ^2+9 \tilde{\delta}^2-22 \delta \tilde{\delta}}{512 \tilde{\delta}^2}z^2+\frac{1292 \delta ^2 \tilde{\delta}-225 \delta ^3-521 \delta  \tilde{\delta}^2+190 \tilde{\delta}^3}{15360 \tilde{\delta}^3}z^3+O(z^4)\,.
\ea
\ee
Again, following Propositions of  Section  \bref{sec:structureW}  we find  that: $\cb_{2}^{(2)} = \cb_{5/3}^{(5/3)} = \cb_{3}^{(3)}$, i.e. the leading terms  coincide;  $\cb_{1}^{(2)} = \cb_{1}^{(5/3)} = \cb_{1}^{(3)}$, $\cb_0^{(3)} = \cb_0^{(5/3)}$ and $\cb_0^{(2)} = \cb_0^{(5/3)} + \cb_{1/3}^{(5/3)} $.

\paragraph{${\bm\alpha=\bf 1}$.}  The classical block $\ws^{(1)} = \cb^{(1)}_1c + \cb^{(1)}_0$ is given by
\be
\label{f11W}
\begin{aligned}
&\cb_1^{(1)} =\frac{\tilde{\delta} (16 (\delta +\tilde{\delta})-1)}{32 \tilde{\delta}-1}z + \left(\frac{2 \delta ^2 (24 \tilde{\delta}-1) (96 \tilde{\delta}+1)+4 \delta \tilde{\delta} (8 \tilde{\delta} (352 \tilde{\delta}-35)+7)}{ (32 \tilde{\delta}-1)^3}+\right.\\
&\left.+\frac{\tilde{\delta} (16 \tilde{\delta} (4 \tilde{\delta} (416 \tilde{\delta}-57)+11)-3)}{4 (32 \tilde{\delta}-1)^3}\right)z^2+ O(z^3) \,, 
\end{aligned}
\end{equation}
\be
\label{f01W}
\begin{aligned}
& \cb_{0}^{(1)} =\frac{32 \tilde{\delta} (\tilde{\delta}-\delta )}{(1-32 \tilde{\delta})^2}z +\left(\frac{ 4\delta ^2 (112 \tilde{\delta} (288 \tilde{\delta}-29)+81)-8 \delta  \tilde{\delta} (16 \tilde{\delta} (352 \tilde{\delta}+25)-79)  -18 \delta}{ (1-32 \tilde{\delta})^4} + \right.\\
&\left.+\frac{1+32 \tilde{\delta} (\tilde{\delta} (8 \tilde{\delta}
   (288 \tilde{\delta}-7)+5)-1)}{4 (1-32 \tilde{\delta})^4}\right)z^2+O(z^3) \,.
\end{aligned}
\ee
Note that  the classical block $\ws^{(1)}$ has a pole $\tilde{\delta} = 1/32$,  $\ws^{(\alpha > 1)}$ has a pole at $\tilde\delta = 0$,  $\ws^{(\alpha <1)}$ does not have poles.

\paragraph{${\bm\alpha=\bf 0}$.}  The classical block $\ws^{(0)}$  is given by
\be
\label{f0W}
\ws^{(0)}=  \tilde{\delta}z+\frac{ \tilde{\delta}  (3 \tilde{\delta} -2)}{4 (\tilde{\delta} -1)}z^2+ \frac{\tilde{\delta}  (5 \tilde{\delta} -2)}{6 (\tilde{\delta} -1)}z^3+ O(z^4)\,.
\ee
Recall that the operators were chosen as \eqref{W3opers}. A different choice  of operators with the same fusion $\tilde{w}=0$ reads $\cO_{2,3} = \cO_{h, w},\;\; \cO_{1,4} =\cO_{h,w}^* =  \cO_{h,-w}$. In this case, the  block function  is given by    
\be
\bar{\ws}^{(0)}=  \frac{  \tilde{\delta}^2}{4 (\tilde{\delta} -1)}z^2+ O(z^4)\,,
\ee
which being exponentiated reproduces the known global $\cW_3$ conformal block calculated in \cite{Fateev:2011qa} (see eq. (2.68) therein). Note that the $z^1$-term is absent since the first coefficient \eqref{B1_2} is zero in the large-$c$ limit.

\bibliographystyle{JHEP}
\bibliography{heavy-block}

\end{document}